\documentclass[10pt]{article}
\usepackage[legalpaper]{geometry}
\usepackage{amsfonts,amssymb,amsmath,amsthm,fancybox,graphicx}
\usepackage{afterpage}
\usepackage{fancyhdr}

\usepackage{braket}
\usepackage{leftidx}

\usepackage[pdftex,plainpages=false,pdfpagelabels,hypertexnames=false,
       colorlinks=true,pdfstartview=FitV,
       linkcolor=blue,citecolor=blue,
       urlcolor=blue]{hyperref}
\usepackage{thumbpdf}
\usepackage{authblk}

\newtheorem{thm}{Theorem}%[section]
\newtheorem{cor}[thm]{Corollary}
\newtheorem{lem}[thm]{Lemma}

\newtheorem{rem}[thm]{Remark}
\newtheorem{ass}[thm]{Assumption}

\allowdisplaybreaks %%% \align allows continuous display of equations

\numberwithin{table}{section}
\numberwithin{equation}{section}
\numberwithin{figure}{section}

\newcommand{\CMPLX}{\ensuremath{\mathbb{C}}}
\newcommand{\REAL}{\ensuremath{\mathbb{R}}}

\newcommand{\BC}{\ensuremath{{\tt BC}}}
\newcommand{\p}{\ensuremath{{\tt p}}}
\newcommand{\aBC}{\ensuremath{{\tt a}}}
\newcommand{\N}{\ensuremath{{\tt N}}}
\newcommand{\D}{\ensuremath{{\tt D}}}

\newcommand{\bu}{\mathbf{u}}
\newcommand{\bv}{\mathbf{v}}
\newcommand{\bM}{\mathbf{M}}
\newcommand{\bA}{\mathbf{A}}
\newcommand{\bb}{\mathbf{b}}
\newcommand{\oh}{{\Omega_h}}

\newcommand{\norm}[1]{\left\|#1\right\|}

\usepackage{subfigure}
\graphicspath{
{./figs/}
{./figs/finalPlots/}
}

\begin{document}

\title{Incorporating local boundary conditions \\
into nonlocal theories
\thanks{Burak Aksoylu was supported in part
  by National Science Foundation DMS 1016190 grant, European Commission
  Marie Curie Career Integration Grant 293978, and Scientific and
  Technological Research Council of Turkey (T\"UB\.{I}TAK) TBAG
  112T240 and MAG 112M891 grants. Research visit of Horst R. Beyer was
  supported in part by T\"UB\.{I}TAK 2221 Fellowship for Visiting
  Scientist Program. Sabbatical visit of Fatih Celiker was supported
  in part by T\"UB\.{I}TAK 2221 Fellowship for Scientist on Sabbatical
  Leave Program.}
\thanks{Fatih Celiker was supported in part by National Science Foundation DMS
  1115280 grant.}
}
%\subtitle{Do you have a subtitle?\\ If so, write it here}

\author[1,4]{Burak Aksoylu}
\author[1,2,3]{Horst Reinhard Beyer}
\author[4]{Fatih Celiker}
\affil[1]{TOBB University of Economics and Technology, 
Department of Mathematics, Ankara, 06560, Turkey}
\affil[2]{Instituto Tecnol\'ogico Superior de Uruapan, 
Carr. Uruapan-Carapan No. 5555, Col. La Basilia, Uruapan, 
Michoac\'an. M\'exico}
\affil[3]{Theoretical Astrophysics, IAAT, Eberhard Karls 
University of T\"ubingen, T\"ubingen 72076, Germany}

\affil[4]{Wayne State University, Department of Mathematics, 656 W. Kirby, Detroit, MI 48202, USA}
\date{\today}

\maketitle
\begin{abstract}
  We study nonlocal equations from the area of peridynamics on bounded
  domains.  In our companion paper, we discover that, on
  $\mathbb{R}^n$, the governing operator in peridynamics, which
  involves a convolution, is a bounded function of the classical
  (local) governing operator.  Building on this, we define an abstract
  convolution operator on bounded domains which is a generalization of
  the standard convolution based on integrals.  The abstract
  convolution operator is a function of the classical operator,
  defined by a Hilbert basis available due to the purely discrete
  spectrum of the latter.  As governing operator of the nonlocal
  equation we use a function of the classical operator, this allows us
  to incorporate local boundary conditions into nonlocal theories.
  The governing operator is determined by what we call the regulating
  function.  By choosing different regulating functions, we can define
  governing operators tailored to the needs of the underlying
  application.

  For the homogeneous wave equation with the considered boundary
  conditions, we prove that continuity is preserved by time
  evolution. Namely, if the initial data is continuous, then the
  solution is continuous for $t \in \REAL.$ This is due to the fact
  that the solution has a unique decomposition into two parts.  The
  first part is the product of a function of time with the initial
  data.  The second part is continuous.  This decomposition is induced
  by the fact that the governing operator has a unique decomposition
  into multiple of the identity and a Hilbert-Schmidt operator.  The
  decomposition also implies that discontinuities remain stationary.

  We give explicit solution expressions for the initial value problems
  with prominent boundary conditions such as periodic, antiperiodic,
  Neumann, and Dirichlet.  In order to connect to the standard
  convolution, we give an integral representation of the abstract
  convolution operator.  We present additional ``simple'' convolutions
  based on periodic and antiperiodic boundary conditions that lead
  Neumann and Dirichlet boundary conditions.

  We present a numerical study of the solutions of the wave equation.
  For discretization, we employ a weak formulation based on a Galerkin
  projection and use piecewise polynomials on each element which
  allows discontinuities of the approximate solution at the element
  borders. We study convergence order of solutions with respect to
  polynomial order and observe optimal convergence.  We depict the
  solutions for each boundary condition.

% \keywords{Nonlocal wave equation, nonlocal operators, peridynamics, elasticity, 
% operator theory, Galerkin projection method.}
% \subclass{47G10, 35L05, 74B99}

\end{abstract}

\section{Introduction}

There are indications that a development of nonlocal theories is
necessary for description of certain natural phenomena.  It is part of
the folklore in physics that the point particle model, which is the
root for \emph{locality} in physics, is the cause of unphysical
singular behavior in the description of phenomena.  On the other hand,
all fundamental theories of physics are local.  There are many more
alternatives for formulating nonlocal theories compared to local ones.
Therefore, the task of formulating a viable nonlocal theory consistent
with experiment seems much harder than that of a local one.  In any
case, for such formulation, an understanding of the capacities of
nonlocal theories appears inevitable.

%%% Vision statement:
Considering wave phenomena only partially successfully described by a
classical wave equation, it seems reasonable to expect that a more
successful model can be obtained by employing the functional calculus of
self-adjoint operators, i.e., by replacing the classical governing
operator $A$ by a suitable function $f(A)$.  We call $f$ the
\emph{regulating function}. Since classical boundary conditions (BCs) is an integral part of
the classical operator, these BCs are automatically inherited by
$f(A)$.  In this way, we vision to model wave phenomena by using
appropriate $f(A)$ and, as a consequence, need to study the effect of
$f(A)$ on the solutions.  One advantage of our approach is that every
symmetry that commutes with $A$ also commutes with $f(A)$.  As a
result, required invariance with respect to classical symmetries
such as translation, rotation and so forth is
preserved.  The choice of regulating functions appropriate for the
physical situation at hand is an object for future research.

We are interested in studying instances of successful modeling by
nonlocal theories of phenomena that cannot be captured by local
theories.  There are noteworthy developments in the area of nonlocal
modeling.  For instance, crack propagation
\cite{silling2000PDfirstPaper} and viscoelastic damping
\cite{beyerKempfle1995} are modeled by peridynamics (PD) and fractional
derivatives, respectively, both of which are nonlocal.  Similar
classes of operators are used in numerous applications such as
nonlocal diffusion
\cite{rossiEtAl2010_book,duEtal2012_sirev,selesonGunzburgerParks2013_nonlocalDomains},
population models
\cite{carrilloFife2005,mogilnerEdelstein-Keshet1999}, image processing
\cite{Gilboa:2008:NonlocalImage,kindermannOsherJones2005}, particle
systems \cite{bodnarVelazquez2006}, phase transition
\cite{albertiBellettini1998,albertiBellettini1998-2}, and coagulation
\cite{fournierLaurencot2006}.  Further applications are in the context
of multiscale modeling, where PD has been shown to be an upscaling of
molecular dynamics
\cite{selesonGunzburgerParks2014,Seleson:2009:UpscalingMDtoPD} and has
been demonstrated as a viable multiscale material model for length
scales ranging from molecular dynamics to classical elasticity
\cite{askariEtAl2008scidac}.  Also see other related engineering
applications
\cite{celikGuvenMadenci2011,kilicAgwaiMadenci2009,kilicMadenci2010,oterkusMadenciEtAl2012,oterkusMadenci2012_initiation},
the review and news articles
\cite{duEtal2012_sirev,duLipton2014_siamNews,lehoucqSilling2010} for a
comprehensive discussion, and the recent book
\cite{madenciOterkus2014_book}.
In addition, we witness a major effort
to meet the need for a mathematical theory for PD applications and
related nonlocal problems addressing, for instance, conditioning
analysis, domain decomposition and variational theory
\cite{aksoyluMengesha2010,aksoyluParks2011,aksoyluUnlu2014_nonlocal},
volume constraints
\cite{duEtal2012_sirev},
nonlinearity
\cite{durukErbayErkip2009,durukErbayErkip2010,durukErbayErkip2011_general,lipton2014},
discretization
\cite{aksoySenocak2011,aksoyluUnlu2014_nonlocal,Emmrich:2007:DiscretizePD_,tianDu2013},
numerical methods
\cite{chenGunzburger2011,duJuTianZhou2013_aPosterioriErrorAnal,duTianZhao2013_adaptiveFEM,selesonBeneddinePrudhomme2013_couplingScheme}, and
various other aspects 
\cite{Alali:2009:HeterogeneousPD,duKammLehoucqParks2012,emmrichLehoucqPuhst2013,Weckner:2007:PDConverge_,hindsRadu2012,lehoucqSilling2008_PD_elasticity,Lehoucq:2008:PDStress,mengesha2012,duMengesha2013_signChangingKernel,mikata2012,selesonGunzburgerParks2014,selesonParks2011_influenceFunction,zhouDu2010}.
%%% End of PD application commercial

Disturbances in solids propagate in form of waves.  The wave equation
is the basic model for the description of the evolution of
deformations.  This paper focuses on the class of nonlocal wave
equations from PD.  Classical elasticity has been successful in
characterizing and measuring the resistance of materials to crack
growth. On the other hand, PD, a nonlocal extension of continuum
mechanics developed by Silling~\cite{silling2000PDfirstPaper}, is
capable of quantitatively predicting the dynamics of propagating
cracks, including bifurcation.  Its effectiveness has been established
in sophisticated applications such as successful description of
results of Kalthoff-Winkler experiments of the fracture of a steel
plate with notches \cite{kalthoffWinkler1988,silling2003}, fracture
and failure of composites, nanofiber networks, and polycrystal
fracture
\cite{kilicMadenci2009,oterkusMadenci2012,sillingBobaru2005,sillingBobaru2004}.
Since PD is a nonlocal theory, one might expect only the appearance of
nonlocal BCs, and, indeed, so far the concept of local BC does not apply to
PD. Instead, external forces must be supplied through the loading
force density $b$ \cite{silling2000PDfirstPaper}.  On the other hand,
we demonstrate that the anticipation that local BCs are incompatible
with nonlocal operators is not quite correct.

The rest of the article is structured as follows.  In Section
\ref{sec:Convo}, we define convolutions on Hilbert spaces and study
properties of related operators.  In Section
\ref{sec:diagonalization}, we provide diagonalizations of the classical
governing operator $A$ and
functions, $f(A)$, of $A$.  The governing nonlocal operator is a
function of $A$. On the other hand, the
representation of the solution of the initial value problem
corresponding to $f(A)$ contains bounded functions of $f(A)$.  Hence,
in Section \ref{sec:diagonalizationFunc}, we prove that bounded
functions of $f(A)$ are bounded functions of $A$. In Section
\ref{sec:soluInhomog}, we give a representation of the solution of the
inhomogeneous wave equation in terms of an eigenbasis.  In addition, in
Section \ref{sec:soluInhomogConvo}, we make the connection to functions
of $A$  from the peridynamic governing operator in the unbounded domain
case.  

In Section~\ref{sec:hilbertSchmidt}, we address the important question
under what conditions solutions will satisfy prescribed BCs.  We find
that Hilbert-Schmidt operators play a crucial role in satisfying BCs.
The Hilbert-Schmidt property
leads to a uniform convergence argument which allows us to interchange
limits.  Hence, BCs are automatically satisfied.  In addition, the
Hilbert-Schmidt property leads to \emph{smoothing} of the input in the
sense an $L^2$ function is mapped into a function that is continuous
up to the boundary.  With additional decay
conditions of the eigenvalues, we reach a uniform convergence argument
also for derivatives which allows us to interchange limits.  As a
consequence, BCs that involve derivatives are also automatically satisfied. 

In Section \ref{sec:convoBC}, we apply the material from previous
sections to study prominent BCs such as periodic,
antiperiodic, Neumann, and Dirichlet BCs. In the case
of periodic and antiperiodic BCs, integral
representations of the abstract convolutions are relatively
straightforward to establish.  On the other hand, for Neumann boundary
condition, this representation is considerably more involved,
requiring arguments related to half-way symmetry of functions.  For
Dirichlet BC, we give representation in terms limits
of integral convolutions.  For both Neumann and Dirichlet conditions, 
we also give other plausible simple definitions of convolutions that are
related to periodic and antiperiodic extensions
of the micromodulus function.

In Section \ref{sec:numerics}, we present a comprehensive numerical
treatment of the nonlocal wave equation.  We have two goals in
numerical experiments.  First, we want to demonstrate that
discontinuities of the initial data remain stationary for $t \in
\REAL$. Second, solutions satisfy the BCs also for $t \in \REAL$.  In
order to show that the two goals are accomplished, we choose
discontinuous initial data and run experiments showing wave evolutions
for all of the considered BCs; periodic, antiperiodic, Neumann, and
Dirichlet.  Furthermore, by choosing continuous initial data, we
draw parallels between the local and nonlocal wave equations for
Neumann and Dirichlet BCs.

\section{Construction for the Bounded Domain}
\label{sec:abstractConstruction}

Practical applications call for a bounded domain. In the unbounded
domain case, in the companion paper
\cite{beyerAksoyluCeliker2014_unbounded}, we discovered that PD uses
as governing operator a function $f(A)$ of the classical operator $A$:
\begin{eqnarray} 
f(A) u(x,t) & := & \int_{{\mathbb{R}}^n} C(y-x) \cdot
\left(u(y,t) - u(x,t) \right) dy \label{governingOp} \\
& = & \left( \int_{{\mathbb{R}}^n} C(y) dy \right) u(x,t) - 
\int_{{\mathbb{R}}^n} C(x-y) u(y,t) dy \nonumber.
\end{eqnarray}
The question is the generalization of convolutions to functions on
bounded domains, preserving the Banach algebra structure of
convolutions for functions on $L^1(\REAL^n)$.  Indeed, such a
generalization is known for periodic functions.  In Section
\ref{periodicboundaryconditions}, we show that this definition is a
special case of Theorem \ref{convolutionsinhilbertspaces}.  

\subsection{Convolutions on Hilbert Spaces}
\label{sec:Convo}

In the unbounded domain case, in our companion paper
\cite{beyerAksoyluCeliker2014_unbounded}, we discovered that the
governing nonlocal operator is a function of a multiple of the Laplace
operator, the classical governing operator.  Therefore, for the
bounded domain case, it is natural to define the governing operator as
a function of the corresponding classical (local) operator.  This
opens a gateway to incorporate local BCs to nonlocal theories.

For simplicity, we choose the classical (local) operator to be a
multiple of the Laplace operator with appropriate BCs. 
In the bounded domain case, the spectrum of the Laplace operator with
classical BCs such as periodic, antiperiodic, Neumann, Dirichlet, and
Robin is purely discrete.  Furthermore, we can explicitly calculate
the eigenfunctions $e_k$ corresponding to each BC and the 
subscript signifies the BC used;  
$\BC \in \{\p, \aBC, \N, \D\}$ where $\p, \aBC, \N$, and $\D$ stand for
periodic, antiperiodic, Neumann, Dirichlet, respectively.
These
eigenfunctions form a Hilbert basis (complete orthonormal basis)
through which the abstract convolution can be defined as follows:
\begin{equation} \label{abstractConvo}
C *_{{\tt BC}} u := \sum \braket{e_k|C} \braket{e_k|u} e_k,
\end{equation}
where 
$$
\braket{e_k | u} := \int_{-1}^{1} e_k^*(y) u(y)dy.
$$
The nonlocal wave equation we solve is given as follows:
\begin{equation} \label{governingEqu}
u_{tt}(x,t) + \varphi(A_{{\tt BC}}) u(x,t) = 0,\quad x \in (-1,1),~t \in [0,T],
\end{equation}
where $T \geqslant 0$, $\varphi: \sigma(A_{{\tt BC}}) \rightarrow \REAL$ is a bounded function
and $A_{{\tt BC}}$ is the (local) classical operator with spectrum $\sigma(A_{{\tt BC}})$.  For instance,
the convolution in \eqref{abstractConvo} defines  the
governing operator $c-C*_{{\tt BC}}$ where $c$ is an appropriate constant.
The \emph{regulating function} 
$\varphi: \sigma(A_{{\tt BC}}) \rightarrow \REAL$ is determined
by the following equation:
\begin{equation} \label{regulatingFunc}
\varphi(A_{{\tt BC}}):= c-C*_{{\tt BC}}  \; .
\end{equation}
Representing $u$ in this Hilbert basis, 
$$
u = \sum \braket{e_k | u} e_k,
$$  
we arrive at
\begin{equation*}
(c - C*_{{\tt BC}}) u =   \sum [c - \braket{e_k|C}] \braket{e_k|u} e_k
= \sum \varphi(\lambda_k) \braket{e_k|u} e_k,
\end{equation*}
where $(\lambda_k, e_k)$ denotes an eigenpair of the classical operator.
Since the last relation regulates which function of the classical
operator is used to define the nonlocal operator, for simplicity, we also 
call $\varphi_C(k):=\varphi(\lambda_k)$ the \emph{regulating function}:
\begin{equation} \label{regulatingFunc}
\varphi_C (k) = c - \braket{e_k|C}.
\end{equation}

The following is a formal definition of convolutions on Hilbert spaces
which leads to a Banach algebra structure.
\begin{thm} 
\label{convolutionsinhilbertspaces} {\bf (Convolutions on Hilbert Spaces)}
Let ${\mathbb{K}} \in \{{\mathbb{R}},{\mathbb{C}}\}$, 
$(X,\braket{\,|\,})$ a non-trivial ${\mathbb{K}}$-Hilbert space with 
induced norm $\|\,\,\|$ and 
$M \subset X$ a Hilbert basis. By 
\begin{equation*}
\xi * \eta := \sum_{e \in M(\xi) \cap M(\eta)} \braket{e|\xi} \braket{e|\eta} . \, e
\end{equation*}
for $\xi,\eta \in X$, where 
\begin{equation*}
M(\xi) := \{e \in M :\braket{e|\xi} \neq 0 \} \, \, , \, \,
M(\eta) := \{e \in M :\braket{e|\eta} \neq 0 \} \, \, , 
\end{equation*}
and the sum over the empty set is defined as $0_{X}$,
there is defined a bilinear, commutative and associative map 
\begin{equation*}
\begin{pmatrix}
* \, \: \! \! \! \! & \! \! \! \! \! \! \! \! \! \! \! X^2 \rightarrow X \\
 & (\xi,\eta) \mapsto \xi * \eta 
\end{pmatrix}
\end{equation*} 
such that 
\begin{equation*}
\|\xi * \eta\| \leqslant \|\xi\| \cdot \|\eta\|
\end{equation*}
for all $\xi,\eta \in X$ and, as a consequence, $(X, +, .\, , *, \|\, \|)$
is a commutative Banach algebra.
\end{thm}

\begin{proof}
We note for $\xi,\eta \in X$ that, as intersection of at most countable subsets of $X$,  $M(\xi) \cap M(\eta)$ is at most countable. Furthermore, for every finite 
subset $S$ of $M(\xi) \cap M(\eta)$
\begin{align*}
& \sum_{e \in S} |\braket{e|\xi} \braket{e|\eta}|^2 = 
\sum_{e \in S} |\braket{e|\xi}|^2 \cdot |\braket{e|\eta}|^2 
\leqslant \left( \sum_{e \in S} |\braket{e|\xi}|^2 \right) \cdot 
\left( \sum_{e \in S} |\braket{e|\eta}|^2 \right) \\
& \leqslant 
\left( \sum_{e \in M(\xi)} |\braket{e|\xi}|^2 \right) \cdot 
\left( \sum_{e \in M(\eta)} |\braket{e|\eta}|^2 \right) =
\|\xi\|^2 \cdot \|\eta\|^2 \, \, .
\end{align*}
Hence the sequence
\begin{equation*}
(|\braket{e|\xi} \braket{e|\eta}|^2)_{e \in M(\xi) \cap M(\eta)} 
\end{equation*}
is (absolutely) summable. As a consequence, the sequence 
\begin{equation*}
(\braket{e|\xi} \braket{e|\eta} . \, e)_{e \in M(\xi) \cap M(\eta)}
\end{equation*}
is summable, with a sum that  is independent of the order of summation. 
Therefore, $\xi * \eta$ is well-defined and satisfies 
\begin{equation*}
\| \xi * \eta\|^2 = \sum_{e \in M(\xi) \cap M(\eta)} |\braket{e|\xi} \braket{e|\eta}|^2  \leqslant \| \xi\|^2 \cdot  \|\eta\|^2 \, \, .
\end{equation*}
That $*$ is bilinear, commutative and associative is obvious.
\end{proof}

We present properties of operators induced by convolutions on Hilbert spaces.

\begin{cor} \label{cor:convoOps}
Let ${\mathbb{K}}$, 
$(X,\braket{\,|\,})$, $\|\,\,\|$, 
$M$ and $*$ as in Theorem~\ref{convolutionsinhilbertspaces}. 
Then  for every $\xi \in X$,
$\xi * \cdot \in L(X,X)$ and, in addition, Hilbert-Schmidt 
and therefore also compact.
\end{cor}

\begin{proof}
Let $\xi \in X$. 
First, it follows from Theorem~\ref{convolutionsinhilbertspaces}
that $\xi * \cdot \in L(X,X)$. Furthermore, from the definition of 
$*$, it follows
for every $e \in M$ that 
\begin{equation*}
\xi * e = \braket{e|\xi} . e \, \, .
\end{equation*}
In particular, this implies that the set 
\begin{equation*}
M(A) := \{e \in M: \xi * e \neq 0\}
\end{equation*}
is at most countable and that 
\begin{equation*}
\left(\|\braket{e|\xi} . e \|^2 \right)_{e \in M(A)}
= \left(|\braket{e|\xi}|^2\right)_{e \in M(A)}
\end{equation*}
is summable. Hence $\xi * \cdot$ is in addition Hilbert-Schmidt 
and therefore also compact. 
\end{proof}

\begin{rem}
In the context of $L^2$ spaces, the Hilbert-Schmidt property leads to
smoothing of input functions; see Theorem \ref{regularizingfunction}.
\end{rem}

In the case that the Hilbert basis is an eigenbasis of an operator
$A$, we give a condition that the operators in Corollary
\ref{cor:convoOps} are functions of $A$.

\begin{cor}  \label{convolutionsasfunctionsofoperators}
Let ${\mathbb{K}} \in \{{\mathbb{R}},{\mathbb{C}}\}$, 
$(X,\braket{\,|\,})$ a non-trivial ${\mathbb{K}}$-Hilbert space, with 
induced norm $\|\,\,\|$, $A$ a densely-defined, linear and 
self-adjoint operator in $X$ with a purely discrete spectrum 
$\sigma(A)$, 
$M \subset X$ the Hilbert basis consisting of the eigenvectors of $A$ and, finally, $*$ the convolution corresponding to 
$M$. Then, for every $\xi \in X$, satisfying 
\begin{align*}
\braket{e|\xi} = \braket{e^{\prime}|\xi}
\end{align*}
for every $e,e^{\prime} \in M$ corresponding to the same 
eigenvalue of $A$, 
$\xi * \cdot$ is a bounded function of $A$.
\end{cor}

\begin{proof}
Let $\xi \in X$ and for every $e \in M$, $\lambda(e) \in 
{\mathbb{R}}$ the corresponding eigenvalue of $A$. Then the 
spectrum  $\sigma(A)$ of $A$ is given by 
\begin{equation*}
\sigma(A) := \{\lambda(e): e \in M\} \, \, .
\end{equation*}
We define $f \in B(\sigma(A),{\mathbb{C}})$, where $B(\sigma(A),{\mathbb{C}})$
denotes complex valued bounded functions on $\sigma(A)$,  by 
\begin{equation*}
f(\lambda(e)) := \braket{e|\xi}
\end{equation*}
for every $e \in M$. This definition leads to a well-defined 
$f$, since 
according to the assumptions,
\begin{align*}
\braket{e|\xi} = \braket{e^{\prime}|\xi}
\end{align*}
for every $e,e^{\prime} \in M$ satisfying $\lambda(e) = \lambda(e^{\prime})$. Also, we note that $f$ is bounded since 
\begin{equation*}
|f(\lambda(e))| = |\braket{e|\xi}| \leqslant \|\xi\| \, \, ,
\end{equation*}
for every $e \in M$.
Furthermore, from the spectral theorem for densely-defined, 
linear and self-adjoint Hilbert spaces and the definition of 
$*$, it follows for every $e \in M$ that 
\begin{equation*}
f(A) e = f(\lambda(e)) e = \braket{e|\xi} . e = \xi * e
\end{equation*}
and hence that 
\begin{equation*}
\xi * \cdot = f(A) \, \, .
\end{equation*}
\end{proof}

We note that simple rotations of the eigenbasis lead to different
convolutions. 

\begin{rem} \label{crucialremark}
Note that if $M \subset X$ is a Hilbert basis and $\alpha : M \rightarrow S^1$,
where $S^1 \subset {\mathbb C}$ denotes the unit circle, then 
\begin{equation*}
M_{\alpha} := \{\alpha(e) . e : M\}
\end{equation*}
is also a Hilbert basis, and for every $\xi, \eta \in X$,
\begin{equation*}
\xi *_{\alpha} \eta = \sum_{e \in M(\xi) \cap M(\eta)} (\alpha(e))^{*}  \braket{e|\xi} \braket{e|\eta} . \, e \, \, , 
\end{equation*}
where $*_{\alpha}$ denotes the convolution that is associated to $M_{\alpha}$. 
\end{rem}

To avoid repetition in the upcoming discussions, we make the following
assumptions with corresponding abstract homogeneous governing equation:
\begin{equation} \label{governingEquAbstract}
u^{\prime \prime}(t) + f(A)u(t) = 0.
\end{equation}
On the other hand, in motivations we use symbols from
the governing equation \eqref{governingEqu} such as $\varphi$ and
$A_{\BC}$ instead of $f$ and $A$.

\begin{ass}
In the following, let $(X,\braket{\,|\,})$ be a non-trivial complex Hilbert Space, 
$A : D(A) \rightarrow X$ a densely-defined, linear and self-adjoint 
operator with a purely discrete spectrum $\sigma(A)$, i.e., such 
that the (non-empty, closed and real) $\sigma(A)$ is discrete and
contains only eigenvalues of finite multiplicity. Note that this implies that, 
$\sigma(A)$ is at most countable,  there is an at most countable 
Hilbert basis $M \subset X$ of $X$ consisting of eigenvectors of 
$A$ and that $(X,\braket{\,|\,})$ is separable. In the following,
we consider the particular case that $X$ is infinite dimensional 
and hence that $M$ is countable. As a consequence, 
\begin{equation*}
M = \{e_1,e_2,\dots \} \, \, , 
\end{equation*}
where $e_1,e_2,\dots$ are pairwise orthogonal normalized 
eigenvectors of $A$. In particular, for every $k \in {\mathbb{N}}^{*}$,
let $\lambda_k$ be the eigenvalue corresponding to $e_k$. Then 
\begin{equation*}
\sigma(A) = \{\lambda_k : k \in {\mathbb{N}}^{*} \} \, \, .
\end{equation*}

\end{ass}

\subsection{Diagonalization of $A$ and Induced Representation of  Bounded Functions of $A$}
\label{sec:diagonalization}

With the knowledge of the eigenbasis, the diagonalization of $A$
is straightforward.

\begin{thm} \label{thm:diagonalization}
Define $U : L^2_{\mathbb{C}}(I) \rightarrow \sum_{k \in {\mathbb{N}}^{*}} {\mathbb{C}} = l^2_{\mathbb{C}}$ by 
\begin{equation*}
U(\xi) := (\braket{e_{k}|\xi}_2)_{k \in {\mathbb{N}}^{*}}
\end{equation*}
for every $\xi \in X$. Then $U$ is a Hilbert space 
isomorphism. In particular, 
\begin{align} \label{diagonalizationofA}
& U \circ A \circ U^{-1} 
= \sum_{k \in {\mathbb{N}}^{*}} \lambda_k .{\textrm{id}}_{\mathbb{C}} \, \, .
\end{align}
\end{thm}

\begin{proof}
In particular, for every $\xi \in D(A)$
\begin{align} \label{diagonalizationofA}
& U \circ A \circ U^{-1} (\braket{e_{k}|\xi}_2)_{k \in {\mathbb{N}}^{*}} 
= U \circ A \xi = (\lambda_k \braket{e_k|\xi}_2)_{k \in {\mathbb{N}}^{*}} \nonumber 
\\
&
= \left[ \sum_{k \in {\mathbb{N}}^{*}} \lambda_k .{\textrm{id}}_{\mathbb{C}}
\right] (\braket{e_{k}|\xi}_2)_{k \in {\mathbb{N}}^{*}}
\, \, .
\end{align}
For the proof of the latter, we note for every $k \in {\mathbb{N}}^{*}$ that 
\begin{equation*}
\braket{e_k|A \xi}_2 = \braket{A e_k|\xi}_2 = \lambda_k \braket{e_k|\xi}_2 \, \, .
\end{equation*}
The  identity (\ref{diagonalizationofA}) implies that 
\begin{equation*}
\sum_{k \in {\mathbb{N}}^{*}} \lambda_k  .{\textrm{id}}_{\mathbb{C}} \supset 
U \circ A \circ U^{-1}
\end{equation*}

and hence, since $U \circ A \circ U^{-1}$ and $\sum_{k \in {\mathbb{N}}^{*}} (\lambda_k .{\textrm{id}}_{\mathbb{C}})$ are both densely-defined, linear and 
self-adjoint operators in $l^2_{\mathbb{C}}$, that 
\begin{equation*} 
U \circ A \circ U^{-1} = \sum_{k \in {\mathbb{N}}^{*}} \lambda_k .{\textrm{id}}_{\mathbb{C}} \, \, .
\end{equation*}
\end{proof}

The Hilbert space isomorphism $U$ in Theorem \ref{thm:diagonalization}
also diagonalizes bounded functions of $A$.  In the solution of the
initial value problem of the wave equation with governing operator
$A$, only bounded functions appear.  Namely, for  
$f \in B(\sigma(A),{\mathbb{C}})$ and $\xi \in X$,
\begin{equation} \label{boundedfunctionsofA}
f(A) \, \xi = 
\sum_{k=1}^{\infty} f(\lambda_k) \braket{e_k|\xi}_2 . e_{k}.
\end{equation}

\subsection{Functions of Functions of $A$}
\label{sec:diagonalizationFunc}

The governing operator in the nonlocal equation \eqref{governingEqu} is a
bounded function, $\varphi(A_{\BC})$, of the classical operator $A_{\BC}$.  As a
consequence, in the solution of the initial value problem, bounded
functions, $g(\varphi(A_{\BC}))$, of bounded function, $\varphi(A_{\BC})$, 
appear.  Since $A_{\BC}$ has a purely discrete spectrum, it is easy to see
that $g(\varphi(A_{\BC}))$ is a bounded function of $A_{\BC}$, 
$(g \circ \varphi)(A_{\BC})$, as indicated in the following. 

\begin{thm}
Let $f \in B(\sigma(A),{\mathbb{R}})$.

\begin{enumerate}
\item[(i)] Then $f(A)$ is
self-adjoint with {\it pure point spectrum}, $\sigma(f(A))$, given by 
\begin{equation*}
\sigma(f(A)) = \overline{\{f(\lambda_k): k \in {\mathbb{N}}^{*}\}} \, \, .
\end{equation*}

\item[(i)] If $g \in B(\sigma(f(A)),{\mathbb{C}})$, $k \in {\mathbb{N}}^{*}$ and $\eta \in X$, then
\begin{equation} \label{compositions}
g(f(A)) = (g \circ f)(A) \, \, .
\end{equation}

\end{enumerate}

\end{thm}

\begin{proof}
  If $f \in B(\sigma(A),{\mathbb{R}})$, we conclude from the spectral
  theorem for densely-defined, linear and self-adjoint Hilbert spaces,
  that $f(A)$ is, in particular, self-adjoint and from
  (\ref{boundedfunctionsofA}) that every member of the Hilbert basis
  $(e_{k})_{k \in {\mathbb{N}}^{*}}$ is an eigenvector of $f(A)$.
  Hence $f(A)$ has a {\it pure point spectrum}, and its spectrum
  $\sigma(f(A))$ is given by
\begin{equation*}
\sigma(f(A)) = \overline{\{f(\lambda_k): k \in {\mathbb{N}}^{*}\}} \, \, .
\end{equation*}

Furthermore, if $g \in B(\sigma(f(A)),{\mathbb{C}})$, $k \in {\mathbb{N}}^{*}$ and $\eta \in X$, it follows from 
the spectral theorem for densely-defined, linear and self-adjoint Hilbert spaces that 
\begin{align*} \label{boundedfunctionsoff(A)}
g(f(A)) \, e_{k} & = g(f(\lambda_k)) . e_{k} \, \, , \nonumber \\
g(f(A)) \, \eta & = g(f(A)) \sum_{k=1}^{\infty} \braket{e_{k}|\eta}_2 . e_{k} =
\sum_{k=1}^{\infty} \braket{e_{k}|\eta}_2 . g(f(A)) \, e_{k} \nonumber \\
& =
\sum_{k=1}^{\infty} g(f(\lambda_k)) \braket{e_{k}|\eta}_2 . e_{k} \\
& =
\sum_{k=1}^{\infty} (g \circ f)(\lambda_k) \braket{e_{k}|\eta}_2 . e_{k} = 
(g \circ f)(A) \eta 
\end{align*}
and hence that 
\begin{equation} \label{compositions}
g(f(A)) = (g \circ f)(A) \, \, .
\end{equation}
\end{proof}

The following functions appear in the solution of the initial value
problem for \eqref{governingEquAbstract} with abstract governing
operator $f(A)$.

\begin{rem}
In particular,
for all $t \in {\mathbb{R}}$, $\eta \in X$,
\begin{align} \label{representations}
& \left[\overline{\cos \left(t \sqrt{\phantom{ij}} \right)}\,
\bigg|_{\sigma(f(A))}\right]\!(f(A)) \, \eta =
\sum_{k=1}^{\infty} \overline{\cos \left(t \sqrt{\phantom{ij}}\right)}(f(\lambda_k)) 
\braket{e_{k}|\eta}_2 . e_{k} \, \, , \nonumber \\
& \left[\,\overline{\frac{\sin \left(t \sqrt{\phantom{ij}} \right)}{\sqrt{\phantom{ij}}}} \, \bigg|_{\sigma(f(A))}\right]\!(f(A)) \, \eta 
=
\sum_{k=1}^{\infty} \overline{\frac{\sin \left(t \sqrt{\phantom{ij}} \right)
}{\sqrt{\phantom{ij}}}}(f(\lambda_k)) 
\braket{e_{k}|\eta}_2 . e_{k} \, \, . \nonumber
\end{align}
\end{rem}

\subsection{Solution of the Inhomogeneous Wave Equation}
\label{sec:soluInhomog}

A solution for vanishing initial data, of the inhomogeneous wave equation 
with abstract governing operator
$f(A)$ corresponding to continuous inhomogeneity is given by the
following theorem.

\begin{thm} \label{solutionoftheinhomogenousequation}
In addition, let $A : D(A) \rightarrow X$ be positive
and 
$$
b: {\mathbb{R}} \rightarrow L^2_{\mathbb{C}}(I)
$$ 
be continuous. Furthermore, let $f: \sigma(A) \rightarrow \REAL$ be bounded.
Then, $v : {\mathbb{R}} \rightarrow X$, for every defined by
\begin{align*}
v(t) := \int_{I_t} 
\left[\,\overline{\frac{\sin \left((t - \tau) \sqrt{\phantom{ij}} \right)}{\sqrt{\phantom{ij}}}} \, \bigg|_{\sigma(f(A))}\!\right]\! \left( f(A) \right) b(\tau)
\, d\tau \, \, ,
\end{align*}
where $\int$ denotes weak integration in $X$,
\begin{equation*}
I_{t} := 
\begin{cases}
[0,t] & \text{if $t \geqslant 0$} \\
[t,0] & \text{if $t < 0$}
\end{cases} \, \, , 
\end{equation*}
is twice continuously differentiable, such that 
\begin{equation*}
v(0) = v^{\prime}(0) = 0 \, \, ,
\end{equation*} 
and 
\begin{equation*} 
v^{\, \prime \prime}(t) + A \, v(t) = b(t), \quad t \in {\mathbb{R}}.
\end{equation*}
\end{thm}

\begin{proof}
See the proof in the companion paper 
\cite[Thm. 2.5]{beyerAksoyluCeliker2014_unbounded}.
\end{proof}

The general inhomogeneous solution is given by the superposition of the 
general homogeneous solution with $v$ in \eqref{inhomogSpecialSolu}.

\begin{rem}
For every $k \in {\mathbb{N}}^{*}$, we calculate the
expansion coefficients $\braket{e_k|v(t)}_2$ with respect to Hilbert
basis corresponding to $A$:
\begin{align*}
\braket{e_k|v(t)}_2 & = \int_{0}^{t} 
\overline{\frac{\sin \left((t - \tau) \sqrt{\phantom{ij}} \right)
}{\sqrt{\phantom{ij}}}}\,(f(\lambda_k)) 
 \braket{e_k|b(\tau)} \, d\tau.
\end{align*}
Hence,
\begin{equation} \label{inhomogSpecialSolu}
v(t) = \sum_{k=1}^{\infty} 
\left\{ 
\int_{0}^{t} \overline{\frac{\sin \left((t - \tau) \sqrt{\phantom{ij}} \right)
}{\sqrt{\phantom{ij}}}}\,(f(\lambda_k)) \braket{e_k|b(\tau)} \, d\tau
\right\} e_k \, \, .
\end{equation}
\end{rem}

\subsection{The Case of Nonlocal Governing Operators Involving Convolutions }
\label{sec:soluInhomogConvo}

Let $*$ denote the convolution in $X$ that, according to
Theorem~\ref{convolutionsinhilbertspaces}, is associated to the
Hilbert basis $(e_{k})_{k \in {\mathbb{N}}^{*}}$.  We consider
operators that are analogous to peridynamic governing operators
in the unbounded domain case and make the connection to
functions of $A$. 

\begin{thm} \label{thm:positivity}
Let $C \in X$ and $c \in \REAL$. Then, the following holds.
\begin{enumerate}
\item[(i)] $c - C * \cdot$ is a bounded operator. 
\item[(ii)] $c - C * \cdot$ is self-adjoint if and only if 
\begin{equation*}
\braket{e_k|C}
\end{equation*}
is real for every $k \in {\mathbb{N}}^{*}$ and, if self-adjoint, positive, if and only if 
\begin{equation*}
\braket{e_k|C} \leqslant c 
\end{equation*}
for every $k \in {\mathbb{N}}^{*}$. 
\item[(iii)] If in addition, 
\begin{equation*}
\braket{e_k|C} = \braket{e_l|C}
\end{equation*}
for every $k, l \in {\mathbb{N}}^{*}$ 
such that $\lambda_k = \lambda_l$, then
\begin{equation*}
c - C * \cdot =  (c - f)(A) \, \, ,  
\end{equation*}
where $f \in B(\sigma(A),{\mathbb{C}})$ is defined by 
\begin{equation*}
f(\lambda_k) := \braket{e_k|C}.
\end{equation*}

\end{enumerate}
\end{thm}

\begin{proof}
For $C, \xi, \eta \in X$ and $c \in {\mathbb{R}}$, 
\begin{equation*}
c - C * \cdot
\end{equation*}
defines a linear operator in $X$ that, as a consequence of 
\begin{equation*}
\|(c - C * ) \xi \| = \|c . \xi - C * \xi \| \leqslant 
|c| \cdot \|\xi\| + \|C * \xi\| \leqslant (\,|c| + \|C\|) \cdot \|\xi\| \, \, , 
\end{equation*}
is bounded. Furthermore, 
\begin{align*}
& C * \xi = \sum_{k \in {\mathbb{N}}^{*}} \braket{e_k|C}_2 \braket{e_k|\xi} . \, e_k  
\, \, , \, \, c . \xi = 
\sum_{k \in {\mathbb{N}}^{*}} c \, \braket{e_k|\xi} . \, e_k \, \, , \\
& c . \xi - C * \xi = \sum_{k \in {\mathbb{N}}^{*}} (c - \braket{e_k|C}) \braket{e_k|\xi} . \, e_k
\, \, , \\
&
U (c . \xi - C * \xi) = \big(\,(c - \braket{e_k|C}) \braket{e_k|\xi}\big)_{k \in {\mathbb{N}}^{*}} \, \, , \\
& U (c - C * \cdot) U^{-1} = \sum_{k \in {\mathbb{N}}^{*}} (c - \braket{e_k|C}) .{\textrm{id}}_{\mathbb{C}} \, \, .
\end{align*}
Therefore, $c - C * \cdot$ is self-adjoint if and only if 
\begin{equation*}
\braket{e_k|C}
\end{equation*}
is real for every $k \in {\mathbb{N}}^{*}$ and, is self-adjoint,  positive, if and only if 
\begin{equation*}
\braket{e_k|C} \leqslant c 
\end{equation*}
for every $k \in {\mathbb{N}}^{*}$. If in addition, 
\begin{equation*}
\braket{e_k|C} = \braket{e_l|C}
\end{equation*}
for every $k, l \in {\mathbb{N}}^{*}$ 
such that $\lambda_k = \lambda_l$, we conclude from the proof
of 
Corollary~\ref{convolutionsasfunctionsofoperators} that 
\begin{equation*}
c - C * \cdot =  (c - f)(A) \, \, ,  
\end{equation*}
where $f \in B(\sigma(A),{\mathbb{C}})$ is defined by 
\begin{equation*}
f(\lambda_k) := \braket{e_k|C}
\end{equation*}
for every $k \in {\mathbb{N}}^{*}$. 
\end{proof}

\begin{rem}
If $c - C * \cdot$ is self-adjoint and positive, then, for all 
$t \in {\mathbb{R}}$, $\eta \in X$, it follows from 
\eqref{compositions} that 
\begin{align*}
& \left[\overline{\cos \left(t \sqrt{\phantom{ij}} \right)}\,
\bigg|_{\sigma(c - C * \cdot)}\right]\!(c - C * \cdot) \, \eta =
\sum_{k=1}^{\infty} \overline{\cos \left(t \sqrt{\phantom{ij}}\right)}(c - \braket{e_k|C}) 
\braket{e_{k}|\eta}_2 . e_{k} \, \, , \\
& \left[\,\overline{\frac{\sin \left(t \sqrt{\phantom{ij}} \right)}{\sqrt{\phantom{ij}}}} \, \bigg|_{\sigma(c - C * \cdot)}\right]\!(c - C * \cdot) \, \eta 
=
\sum_{k=1}^{\infty} \overline{\frac{\sin \left(t \sqrt{\phantom{ij}} \right)
}{\sqrt{\phantom{ij}}}}(c - \braket{e_k|C}) 
\braket{e_{k}|\eta}_2 . e_{k} \, \, .
\end{align*}
Finally, the expression of $v$ follows from \eqref{inhomogSpecialSolu}:

$$
v(t) = \sum_{k=1}^{\infty} 
\left\{ 
\int_{0}^{t} \overline{\frac{\sin \left((t - \tau) \sqrt{\phantom{ij}} \right)
}{\sqrt{\phantom{ij}}}}\,(c - \braket{e_k|C}) \braket{e_k|b(\tau)} \, d\tau
\right\} e_k \, \, .
$$

\end{rem}

\section{Smoothing Functions of Operators and Boundary Conditions}
\label{sec:hilbertSchmidt}

For motivation, we consider the Dirichlet eigenfunction expansion of 
$u = \chi_{[-1/2,1/2]} + 1$ on the interval $I=[-1,1]$
\begin{equation*}
u = \sum_{k = 1}^{\infty} \braket{e_{k}^{\D}|u}_2 e_{k}^{\D}.
\end{equation*}
Although 
\begin{equation*}
\sum_{k = 1}^{N} \braket{e_{k}^{\D}|u}_2 e_{k}^{\D}
\end{equation*}
is infinitely differentiable on $I$ and satisfies the Dirichlet
BCs, we find that $u$ is neither continuous nor
satisfies the BCs.  This opens the important question under what
conditions the solution will satisfy the BCs.  We address this
question in this section and find that Hilbert-Schmidt operators play
a crucial role in satisfying the BCs.  The basis for this is provided
by the fact that the governing operators in Section~\ref{sec:convoBC}
are of the form $c-C$ where $c \in \REAL$ and $C$ is Hilbert-Schmidt
operator.  Consequently, the assumptions on the Hilbert-Schmidt
property made in the following apply to all cases discussed in
Section~\ref{sec:convoBC}.  BCs involving derivatives require stronger
conditions on the decay of the eigenvalues of the operator $C$ than
that of provided by the Hilbert-Schmidt property.  Indeed, we find
that this strong decay is satisfied in the case of Neumann BCs in
Section~\ref{neumannboundaryconditions}.

\subsection{Strategy to Satisfy the Boundary Conditions}
\label{sec:strategy}

The solution is explicitly given in terms of the governing operator 
$c-C$ as follows \cite[Thm. 2.1]{beyerAksoyluCeliker2014_unbounded}: 

\begin{equation} \label{representationofthesolution}
u(x,t) = \left[\overline{\cos \left(t \sqrt{\phantom{ij}} \right)}\,
\bigg|_{\sigma(c-C)}\right]\!(c-C) u(x,0) + 
\left[\, \overline{\frac{\sin \left(t \sqrt{\phantom{ij}} \right)}{\sqrt{\phantom{ij}}}} \, \bigg|_{\sigma(c-C)}\right]\!(c-C) u_t(x,0) 
\end{equation} 
for all $t \in {\mathbb{R}}$, where 
\begin{equation*} 
\overline{\cos(t \sqrt{\phantom{ij}} \,)} \, \, \, \, \textrm{and} \, \, \, \, 
\overline{\frac{\sin(t \sqrt{\phantom{ij}} \,)}{ \sqrt{\phantom{ij}}}}
\end{equation*}
denote the unique extensions of 
$
\cos(t \sqrt{\phantom{ij}} \,) \, \, \, \, \textrm{and}  \, \,
\, \, 
\sin(t \sqrt{\phantom{ij}}) / \sqrt{\phantom{ij}},
$
respectively, to entire holomorphic functions.  These functions are
called \emph{solution operators}.  For brevity of discussion, let us
denote either one of these solution operators as $g(c-C)$.  We follow a
two-step strategy to show how BCs are going to be satisfied:

\begin{enumerate} 
\item Decompose the solution operator as follows:
\begin{equation} \label{decompositionBC}
g(c-C) = [g(c-C) - g(c)] + g(c),
\end{equation}  
so that $g(c-C) - g(c)$ becomes a Hilbert-Schmidt
operator because $C$ is Hilbert-Schmidt.  This leads to a uniform
convergence argument which allows us to interchange limits; see
Theorem \ref{regularizingfunction} and Corollary
\ref{corregularizingfunction}.  We immediately see that 
$g(c-C) - g(c)$ part enforces the BCs.

\item  For the remaining part $g(c)$, we choose initial data
$u(x,0)$ and $u_t(x,0)$ that satisfy the BCs. 
\end{enumerate} 

In order to show that $g(c-C) - g(c)$ is
Hilbert-Schmidt when $C$ is Hilbert-Schmidt, we can utilize power series
expansions.  The operator $c$ commutes with any operator $Z$.  Let
us define $h(Z):= g(c-Z)$.  Since $g(Z)$ is entire, so is
$h(Z)$. Furthermore, we have a power series representation of 
$g(c-Z)$ in powers of $Z$ as follows:
\begin{equation*}
g(c-Z) = h(Z) = \sum_{k=0}^{\infty} \frac{h^{(k)}(0)}{k!} Z^k,
\end{equation*}
where $h^{(k)}(Z) = (-1)^k g^{(k)}(c-Z)$. Hence,
\begin{equation*}
g(c-C) = \sum_{k=0}^{\infty} \frac{(-1)^k g^{(k)}(c)}{k!} C^k.
\end{equation*}
Consequently, we have an expression for $g(c-C) - g(c)$
that contains strictly positive powers of $C$:
\begin{equation*}
g(c-C) - g(c) = 
\sum_{k=1}^{\infty} \frac{(-1)^k g^{(k)}(c)}{k!} C^k.
\end{equation*}
We have shown in Corollary~\ref{cor:convoOps} that abstract
convolution operators are Hilbert-Schmidt. Since $C$ is
Hilbert-Schmidt due to the definition by abstract convolution, any
power series in $C$ that contains strictly positive powers is
also Hilbert-Schmidt.  Consequently, $g(c-C) - g(c)$ is
Hilbert-Schmidt.  For details, see Lemma
\ref{holomorphicfunctionalcalculushilbertschmidt}.

\begin{rem}
  Since the governing operator $\varphi(A_{\BC})= c-C$ is a
  perturbation of the multiple of the identity operator by a compact
  operator, the static form of the inhomogeneous governing equation
  \eqref{governingEqu} satisfies the Fredholm alternative.
\end{rem}

\subsection{Tools to Establish Hilbert-Schmidt Property}

We start with reminding the reader of the holomorphic functional
calculus from the companion paper
\cite{beyerAksoyluCeliker2014_unbounded} and the fact that the
functions,
$\overline{\cos \left(t \sqrt{\phantom{ij}} \right)}\,$ and 
$\overline{\frac{\sin \left(t \sqrt{\phantom{ij}} \right)}{\sqrt{\phantom{ij}}}}$, 
appearing in the solution of the initial value problem in \eqref{representationofthesolution} are entire functions.

\begin{lem} {\bf (Holomorphic Functional Calculus)}
\label{holomorphicfunctionalcalculusI}
Let $(X,\braket{\,|\,})$ be a non-trivial complex 
Hilbert space, $A \in L(X,X)$ self-adjoint and $\sigma(A) \subset {\mathbb{R}}$ the (non-empty, compact) spectrum of $A$. Furthermore, 
let $R > \|A\|$ and  
$g : U_{R}(0) \rightarrow {\mathbb{C}}$ be holomorphic. Then,
the sequence 
\begin{equation*}
\left(\frac{g^{(k)}(0)}{k !} . A^{k} \right)_{k \in {\mathbb{N}}}
\end{equation*}
is absolutely summable in $L(X,X)$ and 
\begin{equation*}
(g|_{{\sigma}(A)})(A) = 
\sum_{k=0}^{\infty} \frac{g^{(k)}(0)}{k !} . A^{k} 
\, \, .
\end{equation*}  
\end{lem}

\begin{proof}
See \cite{beyerAksoyluCeliker2014_unbounded}. 
\end{proof}

In addition, if $A$ is Hilbert-Schmidt, more can be said. Namely,
if $g(0)=0$, then $g(A)$ is Hilbert-Schmidt.

\begin{lem} {\bf (Holomorphic Functional Calculus for Hilbert-Schmidt Operators)}
\label{holomorphicfunctionalcalculushilbertschmidt}
Let $(X,\braket{\,|\,})$ be a non-trivial complex 
Hilbert space and ${\cal I}_2$ be the complex Hilbert space 
consisting of the Hilbert-Schmidt operators on $X$ with induced norm $\|\, \,\|_2$. Furthermore, let 
$A \in {\cal I}_2$ be self-adjoint, $\sigma(A) \subset {\mathbb{R}}$ the (non-empty, compact) spectrum of $A$. Finally, 
let $R > \|A\|_{2}$ and  
$g : U_{R}(0) \rightarrow {\mathbb{C}}$ be holomorphic 
such that $g(0) = 0$. Then
\begin{equation*}
(g|_{{\sigma}(A)})(A) \in 
{\cal I}_2 \, \, .
\end{equation*}
\end{lem}

\begin{proof}
Since $R > \|A\|_2 \geqslant \|A\|$ and $g(0)=0$, an application of 
Lemma~\ref{holomorphicfunctionalcalculusI} gives 
\begin{equation*}
(g|_{{\sigma}(A)})(A) = 
\sum_{k=1}^{\infty} \frac{g^{(k)}(0)}{k !} . A^{k} 
\, \, .
\end{equation*} 
Also, since $g : U_{R}(0) \rightarrow {\mathbb{C}}$ is holomorphic  
and $\|A\|_2 < R$, the sequence  
\begin{equation*}
\left( \frac{g^{(k)}(0)}{k !} \cdot
\|A\|_2^{k} \right)_{k \in {\mathbb{N}}^{*}}
\end{equation*}
is absolutely summable. Using that ${\cal I}_2$ is a $*$-ideal
in $L(X,X)$ and that for all $B, C \in {\cal I}_2$
\begin{equation*}
\|B \circ C\|_2 \leqslant 
\|B\|_2 \cdot \|C\|_2 \, \, ,
\end{equation*} 
see e.g., 
\cite[Vol.~II,~Prop.~5,~p.~41]{reedSimon_books}, it follows for every 
non-empty finite subset $J \subset {\mathbb{N}}^{*}$, 
\begin{equation*}
\sum_{k \in J} \bigg\|\frac{g^{(k)}(0)}{k !} . A^{k} \bigg\|_2
\leqslant \sum_{k \in J} \frac{|g^{(k)}(0)|}{k !} \cdot
\|A\|_2^{k} \leqslant \sum_{k=1}^{\infty} \frac{|g^{(k)}(0)|}{k !} . \|A\|_2^{k} \, \, .
\end{equation*}
As a consequence, 
\begin{equation*}
\left( \frac{g^{(k)}(0)}{k !} \cdot
A^{k} \right)_{k \in {\mathbb{N}}^{*}}
\end{equation*}
is absolutely summable in $({\cal I}_2,\|\, \, \|_2)$. Finally, 
since ${\cal I}_2 \hookrightarrow L(X,X)$ is continuous, 
we conclude that  
\begin{equation*}
(g|_{{\sigma}(A)})(A) \in {\cal I}_2 \, \, .
\end{equation*}
\end{proof}

\begin{rem}
  The same statement holds true for ${\cal I}_1$ which denotes complex
  Banach space of the trace class operators on $X$ and $\|\, \,\|_1$
  is the corresponding trace norm.  The proof is virtually identical to
  the previous one.
\end{rem}

As a consequence of Lemma
\ref{holomorphicfunctionalcalculushilbertschmidt}, we observe that
$\cos(t \sqrt{A})$ is a perturbation of the identity operator by a
Hilbert-Schmidt operator. Likewise, $\frac{\sin(t
  \sqrt{A})}{\sqrt{A}}$ is a perturbation of a multiple of the
identity operator by a Hilbert-Schmidt operator.  
The discussion so far involved the sum of two operators, a multiple of
the identity and a convolution type operator, i.e., $c - C$.  More
generally, applying similar methods used for a proof in the
companion paper \cite[Thm. 4.3]{beyerAksoyluCeliker2014_unbounded},
the following theorem gives that $\cos(t \sqrt{A+B})$ is a
perturbation of $\cos(t\sqrt{A})$ by a Hilbert-Schmidt operator if $B$
is Hilbert-Schmidt for $t \in \REAL.$ Likewise, $\frac{\sin(t
  \sqrt{A+B})}{\sqrt{A+B}}$ is a perturbation of $\frac{\sin(t
  \sqrt{A})}{\sqrt{A}}$ by a Hilbert-Schmidt operator if $B$ is
Hilbert-Schmidt for $t \in \REAL.$

The proof of the following theorem utilizes expansion of solution operators 
in terms of generalized hypergeometric functions given in
the companion paper \cite[Thm. 4.3]{beyerAksoyluCeliker2014_unbounded}.

\begin{thm} \label{coradditiontheorem}
Let $(X,\braket{\,|\,})$ be a non-trivial complex 
Hilbert space, $\sqrt{\phantom{ij}}$
the complex square-root function, with domain ${\mathbb{C}}
\setminus ((- \infty,0] \times \{0\})$. $A, B \in L(X,X)$ self-adjoint
such that $[A,B] = 0$ and $\sigma(A), \sigma(A+B) \subset {\mathbb{R}}$ the 
(non-empty, compact) spectra of $A$ and $A+B$, respectively.  
Let $B \in {\cal I}_2$,
then the operators 
\begin{align*}
& \left[\overline{\cos \left(t \sqrt{\phantom{ij}} \right)}\,
\bigg|_{\sigma(A+B)}\right]\!(A + B)  - \left[\overline{\cos \left(t \sqrt{\phantom{ij}} \right)}\,
\bigg|_{\sigma(A)}\right]\!(A ) \, \, , \\
& \left[\overline{\frac{\sin \left(t \sqrt{\phantom{ij}} \right)}{\sqrt{\phantom{ij}}}} \, \bigg|_{\sigma(A+B)}\right]\!(A+B)  - \left[\overline{\frac{\sin \left(t \sqrt{\phantom{ij}} \right)}{\sqrt{\phantom{ij}}}} \, \bigg|_{\sigma(A)}\right]\!(A)
\end{align*} 
are elements of ${\cal I}_2$. 
\end{thm}

\begin{proof}
For the proof, we use that ${\cal I}_2$ is a $*$-ideal
in $L(X,X)$ and that for all $B \in L(X,X)$, $C \in {\cal I}_2$
\begin{equation*}
\|B \circ C\|_2 \leqslant 
\|B\| \cdot \|C\|_2 \, \, ,
\end{equation*} 
see e.g., \cite[Vol.~II,~Prop.~5,~p.~41]{reedSimon_books}. 
In the following, $\leftidx{_0}{F}{_1}$ denotes the generalized hypergeometric
function, defined as in \cite{olverEtAl2010_book}.
In a first step, we note for every 
$k \in {\mathbb{N}}, z \in {\mathbb{C}}$ that
\begin{align*}
\bigg|\leftidx{_0}{F}{_1}\!\!\left(-; k + \frac{1}{2}; z\right)
\!\!\bigg|
& = \bigg|\sum_{l=0}^{\infty} \frac{z^l}{(k + \frac{1}{2})_l \cdot l!} \bigg| \leqslant
\sum_{l=0}^{\infty} \frac{|z|^l}{(k + \frac{1}{2})_l \cdot l!}
\leqslant \sum_{l=0}^{\infty} \frac{|z|^l}{(\frac{1}{2})^{l} \cdot l!} \\
& = \sum_{l=0}^{\infty} \frac{|2 z|^l}{l!} = e^{2|z|} \, \, , \\
\bigg|\leftidx{_0}{F}{_1}\!\!\left(-; k + \frac{3}{2}; z\right)
\!\!\bigg|
& = \bigg|\sum_{l=0}^{\infty} \frac{z^l}{(k + \frac{3}{2})_l \cdot l!} \bigg| \leqslant
\sum_{l=0}^{\infty} \frac{|z|^l}{(k + \frac{3}{2})_l \cdot l!}
\leqslant \sum_{l=0}^{\infty} \frac{|z|^l}{(\frac{3}{2})^{l} \cdot l!} \\
& = \sum_{l=0}^{\infty} \frac{|2 z/3|^l}{l!} = e^{2|z|/3} \, , 
\end{align*}
and hence, if in addition $k > 0$, for $t \in {\mathbb{R}}$
that 
\begin{align*}
& \bigg\| (-1)^{k} \cdot \frac{t^{2k}}{(2k)!} \, . 
\left\{
\left[\leftidx{_0}{F}{_1}\!\!\left(-; k + \frac{1}{2}; - \, \frac{t^2}{4} \, . {\textrm id}_{\sigma(A)}\right)\right]\!\!(A)
\right\} \! B^k \bigg\|_2 \\
& \leqslant
e^{\,t^2
\|A\| / 2} \, \frac{|t|^{2k}}{(2k)!} \, \|B\|_2^k
\, \, , \\
& \bigg\|(-1)^{k} \cdot \frac{t^{2k+1}}{(2k+1)!} \, . 
\left\{
\left[\leftidx{_0}{F}{_1}\!\!\left(-; k + \frac{3}{2}; - \, \frac{t^2}{4} \, . {\textrm id}_{\sigma(A)}\right)\right]\!\!(A)
\right\} \! B^k \bigg\|_2 \\
& \leqslant
e^{\, 3 t^2
\|A\| / 2} \, \frac{|t|^{2k+1}}{(2k+1)!} \, \|B\|_2^k
\, \, . 
\end{align*} 
As a consequence, the sequences 
\begin{align*}
& \left((-1)^{k} \cdot \frac{t^{2k}}{(2k)!} \, . 
\left\{
\left[\leftidx{_0}{F}{_1}\!\!\left(-; k + \frac{1}{2}; - \, \frac{t^2}{4} \, . {\textrm id}_{\sigma(A)}\right)\right]\!\!(A)
\right\} \! B^k  \right)_{k \in {\mathbb{N}}^{*}} \, \, , \\
& \, \, 
\left( 
(-1)^{k} \cdot \frac{t^{2k+1}}{(2k+1)!} \, . 
\left\{
\left[\leftidx{_0}{F}{_1}\!\!\left(-; k + \frac{3}{2}; - \, \frac{t^2}{4} \, . {\textrm id}_{\sigma(A)}\right)\right]\!\!(A)
\right\} \! B^k
\right)_{k \in {\mathbb{N}}^{*}}
\end{align*}
are absolutely summable in ${\cal I}_2$. Since
$\|C\|_2 \geqslant \|C\|$ for every $C \in {\cal I}_2$, this implies that
the operators 
\begin{align*}
& \left[\overline{\cos \left(t \sqrt{\phantom{ij}} \right)}\,
\bigg|_{\sigma(A+B)}\right]\!(A + B)  - \left[\leftidx{_0}{F}{_1}\!\!\left(-;\frac{1}{2}; - \, \frac{t^2}{4} \, . {\textrm id}_{\sigma(A)}\right)\right]\!\!(A) \, \, , \\
& \left[\overline{\frac{\sin \left(t \sqrt{\phantom{ij}} \right)}{\sqrt{\phantom{ij}}}} \, \bigg|_{\sigma(A+B)}\right]\!(A+B)  - \left[\leftidx{_0}{F}{_1}\!\!\left(-;\frac{3}{2}; - \, \frac{t^2}{4} \, . {\textrm id}_{\sigma(A)}\right)\right]\!\!(A) 
\end{align*} 
are elements of ${\cal I}_2$. From this, the statement follows with the help 
of \cite[Lemma 4.4]{beyerAksoyluCeliker2014_unbounded}.
\end{proof}

In particular, $\cos(t \sqrt{c-C})$ is a perturbation of $\cos(t \sqrt{c})$ by
a Hilbert-Schmidt operator if $C$ is Hilbert-Schmidt for $t \in \REAL.$
Likewise, $\frac{\sin(t \sqrt{c-C})}{\sqrt{c-C}}$ is a perturbation of 
$\frac{\sin(t \sqrt{c})}{\sqrt{c}}$ by
a Hilbert-Schmidt operator if $C$ is Hilbert-Schmidt for $t \in \REAL.$
From the functional calculus for
bounded, linear, self-adjoint operators on Hilbert spaces, it is easy
to conclude that functions of $c - C$ are functions of $C$ in the 
following obvious way.

\begin{cor} \label{hilbertschmidtproperty}
Let $C \in {\cal I}_2$ and $c>0$.  In addition, let 
$A=c$ and $B=-C$, then for every $t \in {\mathbb{R}}$, the operators 
\begin{align*} 
& \left[\overline{\cos \left(t \sqrt{\phantom{ij}} \right)}\,
\bigg|_{\sigma(c-C)}\right]\!(c - C) - \cos(\sqrt{c} \, t \,)
\\ 
& = \left[\overline{\cos \left(t \sqrt{\phantom{ij}} \right)}\,
\circ (c - {\textrm{id}}_{\sigma(C)}) - \cos(\sqrt{c} \, t \,) \right]\!\!(C) 
\, \, , \\
& \left[\overline{\frac{\sin \left(t \sqrt{\phantom{ij}} \right)}{\sqrt{\phantom{ij}}}} \, \bigg|_{\sigma(c - C)}\right]\!(c - C)  - \frac{\sin(\sqrt{c} \, t \,)}{\sqrt{c}} \\
& = 
\left[\overline{\frac{\sin \left(t \sqrt{\phantom{ij}} \right)}{\sqrt{\phantom{ij}}}} \circ (c - {\textrm{id}}_{\sigma(C)}) - 
\frac{\sin(\sqrt{c}\,t\,)}{\sqrt{c} \,} \right]\!\!(C) 
\, \, ,
\end{align*}
are elements of ${\cal I}_2$, where we
define $\sin 0 / 0 := 1$. 
\end{cor}

\begin{proof} The statement is an immediate consequence of 
Theorem~\ref{coradditiontheorem}.
\end{proof}

The previous results enable the treatment of the solutions of the
homogeneous equation with periodic, antiperiodic, and Dirichlet BCs.
In order to treat inhomogeneous equation and BCs that include
derivatives such the Neumann BCs, we need more detailed information on
the eigenvalues of the operators in Corollary
\ref{hilbertschmidtproperty}, as will be given in 
Theorem~\ref{decayofeigenvalues}.

\begin{lem} \label{decayestimate}
Let $c > 0$ and $\lambda \leqslant \min\{c,1\}$. Then 
for every $t \in {\mathbb{R}}$
\begin{align*}
& |\cos(t \sqrt{c - \lambda}\,) - \cos(t \sqrt{c}\,)| \leqslant
\left(\frac{t^2}{2 c}
+\frac{|t|}{\sqrt{c}}\right) \! |\lambda| \, \, , \\
& \left|\frac{\sin(t \sqrt{c - \lambda}\,)}{\sqrt{c - \lambda}}
- \frac{\sin(t \sqrt{c}\,)}{\sqrt{c}}\right| \leqslant 
\left(\frac{t^2}{6c} + \frac{|t|}{2 \sqrt{c}}\right) \! |\lambda| \, \, .
\end{align*}
\end{lem}

\begin{proof}
For $\lambda < c$, we conclude that 
\begin{align*}
& \cos(t \sqrt{c - \lambda}\,) - \cos(t \sqrt{c}\,) =
\cos(t \, [\sqrt{c - \lambda} - \sqrt{c}\,] + t \sqrt{c} \,)
- \cos(t \sqrt{c}\,) \\
& = \cos(t \, [\sqrt{c - \lambda} - \sqrt{c}\,])
\cos(t \sqrt{c} \,) - \sin(t \, [\sqrt{c - \lambda} - \sqrt{c}\,])
\sin(t \sqrt{c} \,) - \cos(t \sqrt{c}\,) \\
& = \{\cos(t \, [\sqrt{c - \lambda} - \sqrt{c}\,]) - 1\} 
\cos(t \sqrt{c} \,) - \sin(t \, [\sqrt{c - \lambda} - \sqrt{c}\,]) \sin(t \sqrt{c} \,) \\
& = -2 \sin^2\!\!\left(\frac{t}{2} \, [\sqrt{c - \lambda} - \sqrt{c}\,]\right) 
\cos(t \sqrt{c} \,) - \sin(t \, [\sqrt{c - \lambda} - \sqrt{c}\,]) \sin(t \sqrt{c} \,) 
\end{align*}
and hence 
\begin{align*}
& |\cos(t \sqrt{c - \lambda}\,) - \cos(t \sqrt{c}\,)| \leqslant
\frac{t^2}{2} \, [\sqrt{c - \lambda} - \sqrt{c}\,]^2 
+ |t| \cdot |\sqrt{c - \lambda} - \sqrt{c}\,| \\
& = \frac{t^2}{2} \left[\frac{-\lambda}{\sqrt{c - \lambda} + \sqrt{c}}\right]^2 
+ |t| \cdot \left|\frac{-\lambda}{\sqrt{c - \lambda} + \sqrt{c}}\,\right| \leqslant \frac{t^2 \lambda^2}{2 c}
+\frac{|t| \cdot |\lambda|}{\sqrt{c}} \, \, .
\end{align*}
Furthermore, using that for $x, y > 0$
\begin{align*}
& \frac{\sin(x)}{x} - \frac{\sin(y)}{y} =
\int_{0}^{1} [\cos(x u) - \cos(y u)] \, du \\
& =
\int_{0}^{1} [\cos((x - y) u + y u) - \cos(y u)] \, du \\
& = \int_{0}^{1} \{[\cos((x - y) u)  -1 ] \cos(y u) + 
\sin((x - y) u) \sin(y u)\} \, du 
\\
& = \int_{0}^{1} \left[- 2 \sin^2\!\!\left(\frac{x - y}{2} \, u \right) \cos(y u) + 
\sin((x - y) u) \sin(y u)\right] du 
\end{align*}
and hence that 
\begin{align*} 
& \left| 
\frac{\sin(x)}{x} - \frac{\sin(y)}{y}
\right| \leqslant 
\int_{0}^{1} \left|- 2 \sin^2\!\!\left(\frac{x - y}{2} \, u \right) \cos(y u) + 
\sin((x - y) u) \sin(y u)\right| du \\
& \leqslant \frac{|x-y|^2}{2} \, \int_{0}^{1} u^2 \, du +
|x-y| \int_{0}^{1} u \, du =
\frac{|x-y|^2}{6} + \frac{|x-y|}{2} \, \, , 
\end{align*}
we conclude that 
\begin{align*}
& \left|\frac{\sin(t \sqrt{c - \lambda}\,)}{\sqrt{c - \lambda}}
- \frac{\sin(t \sqrt{c}\,)}{\sqrt{c}}\right| \leqslant 
\frac{t^2 \, |\sqrt{c - \lambda} - \sqrt{c}|^2}{6} 
+ \frac{|t|\cdot |\sqrt{c - \lambda} - \sqrt{c}|}{2} \\
& = \frac{t^2}{6} \, \left[\frac{-\lambda}{\sqrt{c - \lambda} + \sqrt{c}}\right]^2 + \frac{|t|}{2} \, \left|\frac{-\lambda}{\sqrt{c - \lambda} + \sqrt{c}}\,\right| \leqslant 
\frac{t^2 \lambda^2}{6c} + \frac{|t| \cdot |\lambda|}{2 \sqrt{c}}
\end{align*}
\end{proof}

\begin{thm} \label{decayofeigenvalues}
Let $(X,\braket{\,|\,})$ be a non-trivial complex 
Hilbert space, $\sqrt{\phantom{ij}}$
the complex square-root function, with domain ${\mathbb{C}}
\setminus ((- \infty,0] \times \{0\})$, $c > 0$ and $t \in {\mathbb{R}}$. Furthermore, let $C \in L(X,X)$ be Hilbert-Schmidt, 
\begin{equation*}
(e_k)_{k \in {\mathbb{N}}^{*}}
\end{equation*}
a corresponding basis of eigenvectors and, for every $k \in {\mathbb{N}}^{*}$, $\lambda_k$ 
the eigenvalue of $C$ that corresponds to $e_k$.
\begin{itemize}
\item[(i)] Then
\begin{equation*}
\lim_{k \rightarrow \infty} \lambda_k = 0 \, \, ;
\end{equation*}
\item[(ii)] if $N \in {\mathbb{N}}^{*}$ is such that 
$\lambda_k \leqslant \min\{c,1\}$ for every 
$k \in {\mathbb{N}}^{*}$ satisfying $k \geqslant N$, 
then 
\begin{align*}
&  \left| \left[\overline{\cos \left(t \sqrt{\phantom{ij}} \right)}\,
\circ (c - {\textrm{id}}_{\sigma(C)}) - \cos(\sqrt{c} \, t \,) \right]\!\!(\lambda_k) \right| \leqslant 
\left(\frac{t^2}{2 c}
+\frac{|t|}{\sqrt{c}}\right) \! |\lambda_k|
\, \, , \\
& 
\left| \left[\overline{\frac{\sin \left(t \sqrt{\phantom{ij}} \right)}{\sqrt{\phantom{ij}}}} \circ (c - {\textrm{id}}_{\sigma(C)}) - 
\frac{\sin(\sqrt{c}\,t\,)}{\sqrt{c} \,} \right]\!\!(\lambda_k)
\right| 
\leqslant \left(\frac{t^2}{6c} + \frac{|t|}{2 \sqrt{c}}\right) \! |\lambda_k| \, \, ,
\end{align*}
for every 
$k \in {\mathbb{N}}^{*}$ satisfying $k \geqslant N$.
\end{itemize}
\end{thm}

\begin{proof}
Part~(i): First, we note that 
\begin{align*} 
& C \, e_{k} = \lambda_k . e_{k} \, \, , \, \, 
C \, \xi = C \sum_{k=1}^{\infty} \braket{e_k|\xi}_2 . e_{k} =
\sum_{k=1}^{\infty} \braket{e_k|\xi}_2 . C e_{k} =
\sum_{k=1}^{\infty} \lambda_k \braket{e_k|\xi}_2 . e_{k} 
\, \, ,
\end{align*}
for every $\xi \in X$.
In particular, since for every $k \in {\mathbb{N}}^{*}$
\begin{equation*}
\|C e_{k}\|_2^2 = \|C e_{k}\|_2^2 = |\lambda_k|^2 
\end{equation*}
and $C$ is Hilbert-Schmidt, 
\begin{equation*}
(|\lambda_k|^2)_{k \in {\mathbb{N}}^{*}}
\end{equation*}
is summable. As a consequence,
\begin{equation*}
\lim_{k \rightarrow \infty} \lambda_k = 0 \, \, .
\end{equation*} 
The statement of Part~(ii) is a direct consequence 
of Lemma~\ref{decayestimate}.
\end{proof}

\begin{rem}
Note that the proof of Theorem~\ref{decayofeigenvalues},
together with an application
of the spectral theorem for bounded self-adjoint 
operators on Hilbert spaces, provides an independent proof of 
Corollary~\ref{hilbertschmidtproperty}. In addition, 
Theorem~\ref{decayofeigenvalues} provides the basis for the application of Corollary~\ref{corregularizingfunction}.
\end{rem}

\subsection{Satisfying Boundary Conditions not Involving Derivatives}

Now we are in a position to study BCs not involving derivatives for
operators with a pure point spectrum.  The Hilbert-Schmidt property
leads to a uniform convergence argument which allows us to interchange
limits.  Hence, BCs are automatically satisfied.  In addition, the
Hilbert-Schmidt property leads to \emph{smoothing} of the input, in the
sense that an $L^2$ function is mapped into a function that is continuous
up to the boundary.

\begin{thm} \label{regularizingfunction}
{\bf (Smoothing Functions of an Operator)}
Let $c, K > 0$, $n \in {\mathbb{N}}^{*}$, $\Omega \subset {\mathbb{R}}^n$ be non-empty, bounded and open, 
$A$ be a densely-defined, linear and self-adjoint operator 
in $L^2_{\mathbb{C}}(\Omega)$ with a pure point spectrum $\sigma(A)$, i.e., for which there is a Hilbert basis
\begin{equation*}
(e_k)_{k \in {\mathbb{N}}^{*}}
\end{equation*}
of eigenvectors. In particular, for every $k \in {\mathbb{N}}^{*}$, let $\lambda_k$ be the eigenvalue corresponding to $e_k$. Furthermore, let 
$\hat{\Omega} \supset \bar{\Omega}$ be bounded and open, and 
for every 
$k \in {\mathbb{N}}^{*}$ let $e_k$ be the restriction of some $\hat{e}_k \in C(\hat{\Omega},{\mathbb{C}})$ satisfying 
\begin{equation*}
\|\hat{e}_k\|_{\infty} \leqslant K \, \, .
\end{equation*}
Finally, let $f \in U_{\mathbb{C}}^s(\sigma(A),{\mathbb{C}})$,
$u \in L^2_{\mathbb{C}}(\Omega)$ and 
$b : {\mathbb{R}} \rightarrow L^2_{\mathbb{C}}(\Omega)$ be continuous.
\begin{itemize}
\item[(i)] $f(A)$ is Hilbert-Schmidt if and only 
if
\begin{equation*}
(|f(\lambda_k)|^2)_{k \in {\mathbb{N}}^{*}}
\end{equation*}
is summable.
\item[(ii)] If $f(A)$ is Hilbert-Schmidt, then $f(A) u$ has  
an extension to a continuous function on $\bar{\Omega}$, and 
for every limit point $x$ of
$\Omega$ 
\begin{equation*}
\lim_{y \rightarrow x} [f(A) u](y)
= \sum_{k=1}^{\infty} f(\lambda_k) \braket{e_k|u} \left(\lim_{y \rightarrow x} e_k(y)\right) \, \, . 
\end{equation*}
\item[(iii)] If in addition, 
$f$ is real-valued such that $f(A)$ is Hilbert-Schmidt and 
$f(A) \leqslant c $, 
\begin{itemize}
\item[a)]
then $v : {\mathbb{R}} \rightarrow X$,
for every $t \in {\mathbb{R}}$
defined by 
\begin{align*}
v(t) := \int_{I_t} 
\left[\overline{\frac{\sin \left((t - \tau) \sqrt{\phantom{ij}} \right)}{\sqrt{\phantom{ij}}}} \, \bigg|_{\sigma(c - f(A))}\!\right]\!(c - f(A)) b(\tau)
\, d\tau \, \, ,
\end{align*}
where $\int$ denotes weak integration in $X$,
\begin{equation*}
I_{t} := 
\begin{cases}
[0,t] & \text{if $t \geqslant 0$} \\
[t,0] & \text{if $t < 0$}
\end{cases} \, \, , 
\end{equation*}
for every $t \in {\mathbb{R}}$, satisfies 
\begin{equation*}
v(t) = \sum_{k=1}^{\infty} 
\left\{ 
\int_{I_t} \overline{\frac{\sin \left((t - \tau) \sqrt{\phantom{ij}} \right)
}{\sqrt{\phantom{ij}}}}\,(c - f(\lambda_k)) \braket{e_k|b(\tau)} \, d\tau
\right\} e_k \, \, .
\end{equation*}
\item[b)] then for every $t \in {\mathbb{R}}$, 
$v(t) - v_c(t)$ has  
an extension to a continuous function on $\bar{\Omega}$, and 
for every limit point $x$ of
$\Omega$ 
\begin{align*}
& \lim_{y \rightarrow x} [v(t) - v_c(t)](y) \\
& = \sum_{k=1}^{\infty} \left\{ 
\int_{I_t} \left[ \overline{\frac{\sin \left((t - \tau) \sqrt{\phantom{ij}} \right)
}{\sqrt{\phantom{ij}}}}\,(c - f(\lambda_k)) \right. \right. \\
& \qquad \qquad \quad \, \, \, \, \left. \left. - \overline{\frac{\sin \left((t - \tau) \sqrt{\phantom{ij}} \right)
}{\sqrt{\phantom{ij}}}}\,(c) \right]
\braket{e_k|b(\tau)} \, d\tau
\right\} \left( \lim_{y \rightarrow x }{e}_k(y) \right)
\, \, . 
\end{align*}

where, $v_c : {\mathbb{R}} \rightarrow L^2_{\mathbb{C}}(\Omega)$ is defined by 
\begin{equation*}
v_c(t) := 
\int_{I_t} \overline{\frac{\sin \left((t - \tau) \sqrt{\phantom{ij}} \right)
}{\sqrt{\phantom{ij}}}}\,(c) . b(\tau) \, d\tau
\, \, .
\end{equation*}
\end{itemize}
\end{itemize}
\end{thm}

\begin{proof}
Part~(i):
From the spectral theorem for densely-defined, 
linear and self-adjoint Hilbert spaces, it follows for every $k \in {\mathbb{N}}^{*}$ and $h \in L^2_{\mathbb{C}}(\Omega)$ that 
\begin{align*} 
& f(A) \, e_{k} = f(\lambda_k) . e_{k} \, \, , \\
& f(A) \, h = f(A)  \sum_{k=1}^{\infty} \braket{e_k|h}_2 . e_{k} =
\sum_{k=1}^{\infty} \braket{e_k|h}_2 . f(A) \, e_{k} =
\sum_{k=1}^{\infty} f(\lambda_k) \braket{e_k|h}_2 . e_{k} 
\, \, .
\end{align*}
Hence $f(A)$ has a pure point spectrum, and its spectrum  $\sigma(f(A))$ is 
given by 
\begin{equation*}
\sigma(f(A)) = \overline{\{f(\lambda_k): k \in {\mathbb{N}}^{*}\}} \, \, .
\end{equation*}
In particular, since for every $k \in {\mathbb{N}}^{*}$
\begin{equation*}
\|f(A) e_{k}\|_2^2 = \|f(\lambda_k) e_{k}\|_2^2 = |f(\lambda_k)|^2 \,\, , 
\end{equation*}
$f(A)$ is Hilbert-Schmidt if and only if 
\begin{equation*}
(|f(\lambda_k)|^2)_{k \in {\mathbb{N}}^{*}}
\end{equation*}
is summable. 
\newline
\linebreak
Part~(ii):
It follows for 
$m, m^{\prime} \in {\mathbb{N}}^{*}$ satisfying $N \leqslant m \leqslant m^{\prime}$, where $N \in {\mathbb{N}}^{*}$ is sufficiently large, 
that 
\begin{align*}
& \bigg\|\sum_{k=m}^{m^{\prime}} f(\lambda_k) \braket{e_k|u}_2 \hat{e}_k \bigg\|_{\infty}
\leqslant K \sum_{k=m}^{m^{\prime}} |f(\lambda_k)| \cdot |\braket{e_k|u}_2| \\
& \leqslant K \cdot \left(\,\sum_{k=m}^{m^{\prime}} |f(\lambda_k)|^2 \right)^{\!\!\!1/2} \cdot \left(\,\sum_{k=m}^{m^{\prime}} |\braket{e_k|u}_2|^2 \right)^{\!\!\!1/2} \leqslant K \cdot \|u\|_2 \cdot \left(\,\sum_{k=m}^{m^{\prime}} |f(\lambda_k)|^2 \right)^{\!\!\!1/2} \\
&  
\left(\leqslant K \cdot \|u\|_2 \cdot \left(\,\sum_{k=1}^{\infty} |f(\lambda_k)|^2 \right)^{\!\!\!1/2} 
\, \, \right)\, \, .
\end{align*}
As a consequence, 
\begin{equation*}
\left(\, \sum_{k=1}^{N} f(\lambda_k) \braket{e_k|u}_2 {\hat e}_k \right)_{N \in {\mathbb{N}}^{*}} 
\end{equation*}
is a Cauchy sequence in $(BC(\hat{\Omega},{\mathbb{C}}),\|\, \, \|_{\infty})$, 
(i.e., bounded continuous functions defined on $\hat{\Omega}$) 
and hence uniformly convergent to an extension of $f(A) u$ to a bounded 
continuous function on $\hat{\Omega}$. In particular, this implies
that for every limit point $x$ of
$\Omega$ that 
\begin{align*}
& \lim_{y \rightarrow x} [f(A) u](y) \\
& =
\lim_{y \rightarrow x} 
\lim_{N \rightarrow \infty} \sum_{k=1}^{N} f(\lambda_k) \braket{e_k|u} {\hat e}_k(y) =
\lim_{N \rightarrow \infty} \lim_{y \rightarrow x}  \sum_{k=1}^{N} f(\lambda_k) \braket{e_k|u} {\hat e}_k(y) \\
& = \sum_{k=1}^{\infty} f(\lambda_k) \braket{e_k|u} \left(\lim_{y \rightarrow x} e_k(y)\right) \, \, .
\end{align*}
For the latter, see, e.g.,  in \cite[Thm~2.41]{beyer2007_book}.
\newline
\linebreak
Part~(iii)a):
First, for all $t \in {\mathbb{R}}$ and 
every $k \in {\mathbb{N}}^{*}$, it follows from 
the spectral theorem for densely-defined, linear and self-adjoint in Hilbert spaces that 
\begin{align*}
\braket{e_k|v(t)} & = \int_{I_t} 
\braket{e_k|\left[\overline{\frac{\sin \left((t - \tau) \sqrt{\phantom{ij}} \right)}{\sqrt{\phantom{ij}}}} \, \bigg|_{\sigma(c - f(A))}\right]\!(c - f(A)) b(\tau)} \, d\tau \\
& = \int_{I_t} \braket{ \left[\overline{\frac{\sin \left((t - \tau) \sqrt{\phantom{ij}} \right)}{\sqrt{\phantom{ij}}}} \, \bigg|_{\sigma(c - f(A))}\right]\!(c - f(A))
 e_{k}| b(\tau)} \, d\tau \\
& = \int_{I_t} 
\overline{\frac{\sin \left((t - \tau) \sqrt{\phantom{ij}} \right)
}{\sqrt{\phantom{ij}}}}\,(c - f(\lambda_k)) 
 \braket{e_k|b(\tau)} \, d\tau
\end{align*}
and hence that
\begin{equation*}
v(t) = \sum_{k=1}^{\infty} 
\left\{ 
\int_{I_t} \overline{\frac{\sin \left((t - \tau) \sqrt{\phantom{ij}} \right)
}{\sqrt{\phantom{ij}}}}\,(c - f(\lambda_k)) \braket{e_k|b(\tau)} \, d\tau
\right\} e_k \, \, .
\end{equation*}
Part(iii)b):
If $k \in {\mathbb{N}}^{*}$ is such that 
$f(\lambda_k) \leqslant \min\{c,1\}$, then 
\begin{align*}
\left|\,\overline{\frac{\sin \left((t - \tau) \sqrt{\phantom{ij}} \right)
}{\sqrt{\phantom{ij}}}}\,(c - f(\lambda_k)) -
\overline{\frac{\sin \left((t - \tau) \sqrt{\phantom{ij}} \right)
}{\sqrt{\phantom{ij}}}}\,(c)\,\right| \leqslant 
\left[\frac{(t - \tau)^2}{6c} + \frac{|t - \tau|}{2 \sqrt{c}}\right] \cdot |f(\lambda_k)| 
\end{align*}
and hence 
\begin{align*}
& \left|\int_{I_t} \left[\overline{\frac{\sin \left((t - \tau) \sqrt{\phantom{ij}} \right)
}{\sqrt{\phantom{ij}}}}\,(c - f(\lambda_k)) 
- \overline{\frac{\sin \left((t - \tau) \sqrt{\phantom{ij}} \right)
}{\sqrt{\phantom{ij}}}}\,(c) \right]
\braket{e_k|b(\tau)} \, d\tau \right| \\
& \leqslant 
\int_{I_t} \left[\frac{(t - \tau)^2}{6c} + \frac{|t - \tau|}{2 \sqrt{c}}\right] \cdot |f(\lambda_k)| \cdot |\!\braket{e_k|b(\tau)}\!| \, d\tau \\
& \leqslant \left(\frac{t^2}{6c} + \frac{|t|}{2 \sqrt{c}}\right) |f(\lambda_k)| \int_{I_t} 
|\!\braket{e_k|b(\tau)}\!| \, d\tau
\end{align*}
It follows for 
$m, m^{\prime} \in {\mathbb{N}}^{*}$ satisfying $N \leqslant m \leqslant m^{\prime}$, where $N \in {\mathbb{N}}^{*}$
is sufficiently large, 
that 
\begin{align*}
& \bigg\|\sum_{k=m}^{m^{\prime}} 
\left\{ 
\int_{I_t} \left[ \overline{\frac{\sin \left((t - \tau) \sqrt{\phantom{ij}} \right)
}{\sqrt{\phantom{ij}}}}\,(c - f(\lambda_k)) - \overline{\frac{\sin \left((t - \tau) \sqrt{\phantom{ij}} \right)
}{\sqrt{\phantom{ij}}}}\,(c) \right] \braket{e_k|b(\tau)} \, d\tau
\right\} \hat{e}_k \bigg\|_{\infty} \\
& \leqslant K \left(\frac{t^2}{6c} + \frac{|t|}{2 \sqrt{c}}\right) \int_{I_t} \left[ \, \sum_{k=m}^{m^{\prime}}  |f(\lambda_k)| \cdot 
|\!\braket{e_k|b(\tau)}\!| \right] d\tau \\
& \leqslant K \left(\frac{t^2}{6c} + \frac{|t|}{2 \sqrt{c}}\right) \cdot 
\left[\, \sum_{k=m}^{m^{\prime}}  |f(\lambda_k)|^2 \right]^{1/2} 
\cdot  \int_{I_t} \left[\,\sum_{k=m}^{m^{\prime}} 
|\!\braket{e_k|b(\tau)}\!|^2 \right]^{1/2} \! d\tau
\\
& 
\leqslant K \left(\frac{t^2}{6c} + \frac{|t|}{2 \sqrt{c}}\right) \cdot 
\left[\, \sum_{k=m}^{m^{\prime}}  |f(\lambda_k)|^2 \right]^{1/2} 
\cdot  \int_{I_t} \|b(\tau)\|_2 \, d\tau
\, \, .
\end{align*}
As a consequence, 
\begin{align*}
& \left(\, \sum_{k=1}^{N} \left\{ 
\int_{I_t} \left[ \overline{\frac{\sin \left((t - \tau) \sqrt{\phantom{ij}} \right)
}{\sqrt{\phantom{ij}}}}\,(c - f(\lambda_k)) \right. \right. \right. \\
& \left. \left. \left. \qquad \qquad \quad \, \, - \overline{\frac{\sin \left((t - \tau) \sqrt{\phantom{ij}} \right)
}{\sqrt{\phantom{ij}}}}\,(c) \right]
\braket{e_k|b(\tau)} \, d\tau
\right\} {\hat e}_k \right)_{N \in {\mathbb{N}}^{*}} 
\end{align*}
is a Cauchy sequence in $(BC(\hat{\Omega},{\mathbb{C}}),\|\, \, \|_{\infty})$
and hence uniformly convergent to an extension of $f(A) g$ to a bounded 
continuous function on $\hat{\Omega}$. In particular, this implies
that for every limit point $x$ of
$\Omega$ that 
\begin{align*}
& \lim_{y \rightarrow x} [v(t) - v_c(t)](y)  \\
& =
\lim_{y \rightarrow x} 
\lim_{N \rightarrow \infty} \sum_{k=1}^{N} \left\{ 
\int_{I_t} \left[ \overline{\frac{\sin \left((t - \tau) \sqrt{\phantom{ij}} \right)
}{\sqrt{\phantom{ij}}}}\,(c - f(\lambda_k)) \right. \right . \\
& \left. \left. \qquad \qquad \qquad \qquad \quad \, \, \, \,  - \overline{\frac{\sin \left((t - \tau) \sqrt{\phantom{ij}} \right)
}{\sqrt{\phantom{ij}}}}\,(c) \right]
\braket{e_k|b(\tau)} \, d\tau
\right\} {e}_k(y) \\
& =
\lim_{N \rightarrow \infty} \lim_{y \rightarrow x} 
\sum_{k=1}^{N} \left\{ 
\int_{I_t} \left[ \overline{\frac{\sin \left((t - \tau) \sqrt{\phantom{ij}} \right)
}{\sqrt{\phantom{ij}}}}\,(c - f(\lambda_k)) \right. \right. \\
& \left. \left.\qquad \qquad \qquad \qquad \quad \, \, \, \, - \overline{\frac{\sin \left((t - \tau) \sqrt{\phantom{ij}} \right)
}{\sqrt{\phantom{ij}}}}\,(c) \right]
\braket{e_k|b(\tau)} \, d\tau
\right\} {e}_k(y) \\
& = \sum_{k=1}^{\infty} \left\{ 
\int_{I_t} \left[ \overline{\frac{\sin \left((t - \tau) \sqrt{\phantom{ij}} \right)
}{\sqrt{\phantom{ij}}}}\,(c - f(\lambda_k)) \right. \right. \\
& \left. \left. \qquad \qquad \quad \, \, \, \, - \overline{\frac{\sin \left((t - \tau) \sqrt{\phantom{ij}} \right)
}{\sqrt{\phantom{ij}}}}\,(c) \right]
\braket{e_k|b(\tau)} \, d\tau
\right\} \left( \lim_{y \rightarrow x }{e}_k(y) \right) \, \, .
\end{align*}
For the latter, see, e.g., \cite[Thm~2.41]{beyer2007_book}.
\end{proof}

As a consequence, for the example of the Dirichlet BCs
in Section \ref{dirichletboundaryconditions}, solutions to the wave
equation corresponding to data $u(0,\cdot), u^{\prime}(0,\cdot) \in
L^2_{\mathbb{C}}(I)$ satisfying pointwise for $x$ in a some
neighborhood of $-1$ and $1$
\begin{equation*}
\lim_{x \rightarrow -1} u(x,0) = \lim_{x \rightarrow 1} u(x,0) = 0 \, \, , 
\end{equation*}
in the same sense, 
satisfy the Dirichlet BCs for all
$t \in {\mathbb{R}}$. In addition, to the micromoduli considered
in Section~\ref{neumannboundaryconditions}, Corollary~\ref{corregularizingfunction} is applicable. 

\subsection{Satisfying Boundary Conditions Involving Derivatives}

Now we are in a position to study BCs involving derivatives for
operators with a pure point spectrum.  With additional decay
conditions of the eigenvalues of $f(A)$, we reach a uniform
convergence argument also for derivatives which allows us to
interchange limits.  As a consequence, BCs are automatically
satisfied.  For instance, solutions to the wave equation in Section
\ref{neumannboundaryconditions}, for data $u(0,\cdot),
u^{\prime}(0,\cdot) \in L^2_{\mathbb{C}}(I)$ satisfying pointwise for
$x$ in a some neighborhood of $-1$ and $1$
\begin{equation*}
\lim_{x \rightarrow - 1} u^{\prime}(x,0) = \lim_{x \rightarrow 1} u^{\prime}(x,0) = 0 \, \, , 
\end{equation*}
in the same sense, satisfy the Neumann BCs for all $t \in {\mathbb{R}}$.

\begin{cor} \label{corregularizingfunction}
{\bf (Smoothing Functions of an Operator II)} In addition to the assumptions of Theorem~\ref{regularizingfunction}, we assume that 
$\Omega = I$, ${\hat \Omega} = {\hat I}$  and $I, \hat{I}$ are  non-empty open intervals of
${\mathbb{R}}$.
Furthermore, we assume that $f(A)$ is Hilbert-Schmidt 
and for every 
$k \in {\mathbb{N}}^{*}$ that $e_k$ is differentiable 
with a derivative that has an extension ${\hat e}^{\prime}_k$
to a continuous function on ${\bar{\hat I}}$. Finally, we assume that 
\begin{equation*}
(\,|\,\|{e}^{\prime}_k\|_{\infty} \, f(\lambda_k)|^2
)_{k \in {\mathbb{N}}^{*}}
\end{equation*}
is summable.
\begin{itemize}
\item[(i)]
Then $f(A) u \in C^1({\bar I},{\mathbb{C}})$, and 
for every limit point $x$ of
$I$ 
\begin{align*}
& \lim_{y \rightarrow x} [f(A) u](y)
= \sum_{k=1}^{\infty} f(\lambda_k) \braket{e_k|u} \left(\lim_{y \rightarrow x} e_k(y)\right) \, \, , \\
& \lim_{y \rightarrow x} [f(A) u]^{\prime}(y)
= \sum_{k=1}^{\infty} f(\lambda_k) \braket{e_k|u} \left(\lim_{y \rightarrow x} e_k^{\prime}(y)\right)
\, \, . 
\end{align*}
\item[(ii)] If in addition, 
$f$ is real-valued such that 
$f(A) \leqslant c$ and $v, v_c$ are defined as in Theorem~\ref{regularizingfunction}~(iii)a), 
then for every $t \in {\mathbb{R}}$, 
$v(t) - v_c(t) \in C^1({\bar{I}},{\mathbb{C}})$ and 
for every limit point $x$ of
$\Omega$ 
\begin{align*}
& \lim_{y \rightarrow x} [v(t) - v_c(t)](y) \\
& = \sum_{k=1}^{\infty} \left\{ 
\int_{I_t} \left[ \overline{\frac{\sin \left((t - \tau) \sqrt{\phantom{ij}} \right)
}{\sqrt{\phantom{ij}}}}\,(c - f(\lambda_k)) \right. \right. \\
& \qquad \qquad \qquad \left. \left. - \overline{\frac{\sin \left((t - \tau) \sqrt{\phantom{ij}} \right)
}{\sqrt{\phantom{ij}}}}\,(c) \right]
\braket{e_k|b(\tau)} \, d\tau
\right\} \left( \lim_{y \rightarrow x }{e}_k(y) \right)
\, \, , \\
& \lim_{y \rightarrow x} [v(t) - v_c(t)]^{\prime}(y) \\
& = \sum_{k=1}^{\infty} \left\{ 
\int_{I_t} \left[ \overline{\frac{\sin \left((t - \tau) \sqrt{\phantom{ij}} \right)
}{\sqrt{\phantom{ij}}}}\,(c - f(\lambda_k)) \right. \right. \\
& \qquad \qquad \qquad \left. \left. - \overline{\frac{\sin \left((t - \tau) \sqrt{\phantom{ij}} \right)
}{\sqrt{\phantom{ij}}}}\,(c) \right]
\braket{e_k|b(\tau)} \, d\tau
\right\} \left( \lim_{y \rightarrow x }{e}_k^{\prime}(y) \right)
\, \, .
\end{align*}
\end{itemize}
\end{cor}

\begin{proof}
Part~(i):
It follows for 
$m, m^{\prime} \in {\mathbb{N}}^{*}$ satisfying $N \leqslant m \leqslant m^{\prime}$, where $N \in {\mathbb{N}}^{*}$ is sufficiently large, 
that 
\begin{align*}
& \bigg\|\sum_{k=m}^{m^{\prime}} f(\lambda_k) \braket{e_k|u}_2 {\hat e}^{\prime}_k \bigg\|_{\infty}
\leqslant \sum_{k=m}^{m^{\prime}} |f(\lambda_k)| \, \|{\hat e}^{\prime}_k\|_{\infty} \cdot |\braket{e_k|u}_2| \\
& \leqslant \cdot \left(\,\sum_{k=m}^{m^{\prime}} |\,\|{\hat e}^{\prime}_k\|_{\infty} \, f(\lambda_k)|^2 \right)^{\!\!\!1/2} \cdot \left(\,\sum_{k=m}^{m^{\prime}} |\braket{e_k|u}_2|^2 \right)^{\!\!\!1/2} \\
& \leqslant \|u\|_2 \cdot \left(\,\sum_{k=m}^{m^{\prime}} |\,\|{\hat e}^{\prime}_k\|_{\infty} \, f(\lambda_k)|^2 \right)^{\!\!\!1/2} 
\left(\leqslant \cdot \|u\|_2 \cdot \left(\,\sum_{k=1}^{\infty} 
|\,\|{\hat e}^{\prime}_k\|_{\infty} \, f(\lambda_k)|^2
\right)^{\!\!\!1/2} 
\, \, \right)\, \, .
\end{align*}
As a consequence, 
\begin{equation*}
\left(\, \sum_{k=1}^{N} f(\lambda_k) \braket{e_k|u}_2 {\hat e}^{\prime}_k \right)_{N \in {\mathbb{N}}^{*}} 
\end{equation*}
is a Cauchy sequence in $(C({\bar{\hat I}},{\mathbb{C}}),\|\, \, \|_{\infty})$
and hence uniformly convergent to a continuous function on
${\bar{\hat I}}$. In particular, this implies
that $f(A)u$ is continuously differentiable and 
for every limit point $x$ of
$I$ that 
\begin{align*}
& \lim_{y \rightarrow x} [f(A) u]^{\prime}(y) = \sum_{k=1}^{\infty} f(\lambda_k) \braket{e_k|u} \left(\lim_{y \rightarrow x} e_k(y) \right) \, \, . 
\end{align*}
For the latter, see, e.g., \cite[Thm~2.42]{beyer2007_book}.
\newline
\linebreak
Part~(ii): 
It follows for $t \in {\mathbb{R}}$, 
$m, m^{\prime} \in {\mathbb{N}}^{*}$ satisfying $N \leqslant m \leqslant m^{\prime}$, where $N \in {\mathbb{N}}^{*}$
is sufficiently large, 
that 
\begin{align*}
& \bigg\|\sum_{k=m}^{m^{\prime}} 
\left\{ 
\int_{I_t} \left[ \overline{\frac{\sin \left((t - \tau) \sqrt{\phantom{ij}} \right)
}{\sqrt{\phantom{ij}}}}\,(c - f(\lambda_k)) - \overline{\frac{\sin \left((t - \tau) \sqrt{\phantom{ij}} \right)
}{\sqrt{\phantom{ij}}}}\,(c) \right] \braket{e_k|b(\tau)} \, d\tau
\right\} {\hat e}^{\prime}_k \bigg\|_{\infty} \\
& \leqslant K \left(\frac{t^2}{6c} + \frac{|t|}{2 \sqrt{c}}\right) \int_{I_t} \left[ \, \sum_{k=m}^{m^{\prime}}  \|{\hat e}^{\prime}_k\|_{\infty} \, |f(\lambda_k)| \cdot 
|\!\braket{e_k|b(\tau)}\!| \right] d\tau \\
& \leqslant K \left(\frac{t^2}{6c} + \frac{|t|}{2 \sqrt{c}}\right) \cdot 
\left[\, \sum_{k=m}^{m^{\prime}}  |\|{\hat e}^{\prime}_k\|_{\infty} \, f(\lambda_k)|^2 \right]^{1/2} 
\cdot  \int_{I_t} \left[\,\sum_{k=m}^{m^{\prime}} 
|\!\braket{e_k|b(\tau)}\!|^2 \right]^{1/2} \! d\tau
\\
& 
\leqslant K \left(\frac{t^2}{6c} + \frac{|t|}{2 \sqrt{c}}\right) \cdot 
\left[\, \sum_{k=m}^{m^{\prime}}  |\|{\hat e}^{\prime}_k\|_{\infty} \, f(\lambda_k)|^2 \right]^{1/2} 
\cdot  \int_{I_t} \|b(\tau)\|_2 \, d\tau
\, \, .
\end{align*}
As a consequence, 
\begin{align*}
& \left(\, \sum_{k=1}^{N} \left\{ 
\int_{I_t} \left[ \overline{\frac{\sin \left((t - \tau) \sqrt{\phantom{ij}} \right)
}{\sqrt{\phantom{ij}}}}\,(c - f(\lambda_k)) \right. \right. \right. \\
& \left. \left. \left. \qquad \qquad \quad \, \, - \overline{\frac{\sin \left((t - \tau) \sqrt{\phantom{ij}} \right)
}{\sqrt{\phantom{ij}}}}\,(c) \right]
\braket{e_k|b(\tau)} \, d\tau
\right\} {\hat e}^{\prime}_k  \right)_{N \in {\mathbb{N}}^{*}} 
\end{align*}
is a Cauchy sequence in $(C({\bar{\hat I}},{\mathbb{C}}),\|\, \, \|_{\infty})$
and hence uniformly convergent to a continuous function on
${\bar{\hat I}}$. In particular, this implies
that $v(t) - v_c(t)$ is continuously differentiable and 
for every limit point $x$ of
$I$ that  
\begin{align*}
& \lim_{y \rightarrow x} [v(t) - v_c(t)]^{\prime}(y)  \\
& =
\lim_{y \rightarrow x} 
\lim_{N \rightarrow \infty} \sum_{k=1}^{N} \left\{ 
\int_{I_t} \left[ \overline{\frac{\sin \left((t - \tau) \sqrt{\phantom{ij}} \right)
}{\sqrt{\phantom{ij}}}}\,(c - f(\lambda_k)) \right. \right . \\
& \left. \left. \qquad \qquad \qquad \qquad \quad \, \, \, \,  - \overline{\frac{\sin \left((t - \tau) \sqrt{\phantom{ij}} \right)
}{\sqrt{\phantom{ij}}}}\,(c) \right]
\braket{e_k|b(\tau)} \, d\tau
\right\} {\hat e}^{\prime}_k(y) \\
& =
\lim_{N \rightarrow \infty} \lim_{y \rightarrow x} 
\sum_{k=1}^{N} \left\{ 
\int_{I_t} \left[ \overline{\frac{\sin \left((t - \tau) \sqrt{\phantom{ij}} \right)
}{\sqrt{\phantom{ij}}}}\,(c - f(\lambda_k)) \right. \right. \\
& \left. \left.\qquad \qquad \qquad \qquad \quad \, \, \, \, - \overline{\frac{\sin \left((t - \tau) \sqrt{\phantom{ij}} \right)
}{\sqrt{\phantom{ij}}}}\,(c) \right]
\braket{e_k|b(\tau)} \, d\tau
\right\} {\hat e}^{\prime}_k(y) \\
& = \sum_{k=1}^{\infty} \left\{ 
\int_{I_t} \left[ \overline{\frac{\sin \left((t - \tau) \sqrt{\phantom{ij}} \right)
}{\sqrt{\phantom{ij}}}}\,(c - f(\lambda_k)) \right. \right. \\
& \left. \left. \qquad \qquad \quad \, \, \, \, - \overline{\frac{\sin \left((t - \tau) \sqrt{\phantom{ij}} \right)
}{\sqrt{\phantom{ij}}}}\,(c) \right]
\braket{e_k|b(\tau)} \, d\tau
\right\} \left( \lim_{y \rightarrow x }{\hat e}^{\prime}_k(y) \right) \, \, .
\end{align*}
For the latter, see, e.g., \cite[Thm~2.41]{beyer2007_book}. 
\end{proof}

\iftrue
\section{Study of Convolutions with Various Boundary Conditions}
\label{sec:convoBC}

We study one-dimensional elasticity, which is an instance of regular
Sturm-Liouville theory with prominent BCs such as periodic,
antiperiodic, Dirichlet, and Neumann.  In regular Sturm-Liouville
problems, all BCs leading to self-adjoint operators are known
\cite[Thm. 13.14]{weidmann2003_book}.  If needed, all associated BCs
can be considered.  All regular Sturm-Liouville operators are known to
have a purely discrete spectrum, in particular, there is a Hilbert basis of
eigenfunctions.  There are a number of standard problems in higher
dimensions that can be reduced to regular Sturm-Liouville problems on
bounded domains. Also, generically, a differential operator with regular
coefficients on $\REAL^n$ has a purely discrete spectrum, providing an
eigenbasis of the underlying space.  Since the essential ingredient is
a self-adjoint operator with a purely discrete spectrum, hence, our
approach can easily cover higher spatial dimensions.

The choice of a Hilbert basis determines an abstract convolution,
which we refer to as \emph{canonical}.  The most relevant BCs in
applications are Dirichlet and Neumann BCs.  In these cases, the
connection of the abstract convolution to an integral form is not
direct.  On the other hand, for periodic and antiperiodic BCs, that
connection is direct, needing a periodic and antiperiodic extension of
the micromodulus function, respectively.  Because of this directness,
we choose to include periodic and antiperiodic BCs.

In the case of Neumann and Dirichlet BCs, we study additional
convolutions that we refer to as ``\emph{simple}.''  These are
inspired by the convolutions from the periodic and antiperiodic BCs.
Certain combinations of convolutions derived from periodic and
antiperiodic BCs of even micromoduli with even and odd input function
enforce Neumann and Dirichlet BCs in these simple convolutions.  For
instance, we sketch the case of Dirichlet BC. Let $\hat{C_{\p}}$ and
$\hat{C_{\aBC}}$ denote periodic and antiperiodic extensions of $C$,
respectively. It is easy to see that
\begin{equation*}
C *_{\p} u (1) = C *_{\p} u(-1),
\end{equation*}
for any $C$.  In addition, if $C$ is even, $C *_{\p} u (1)=0$ when $u$ is
odd.  Likewise, for the antiperiodic case,
\begin{equation*}
C *_{\aBC} u (1) = - C *_{\aBC} u(-1),
\end{equation*}
holds for any $C$.  If $C$ is even, $C *_{\aBC} u (1)=0$ when $u$ is even.
This suggests that $C *_{\p}  P_{\textrm{odd}}$ and 
$C *_{\aBC}  P_{\textrm{even}}$ are functions of the
classical operator.  For the Neumann BC, the situation is similar. 
Namely, $C *_{\p}  P_{\textrm{even}}$ and
$C *_{\aBC}  P_{\textrm{odd}}$ are functions of the
classical operator as well, where
$P_{\textrm{odd}}$ and $P_{\textrm{even}}$ are orthogonal projections 
onto odd and even functions, respectively.
We elaborate on these examples in Sections \ref{simpleConvoDirichlet} 
and \ref{simpleConvoNeumann}, respectively.

Integral operators on bounded domains are often Hilbert-Schmidt, and
hence, compact.  Indeed, for all these BCs, we show that a simple
decay condition on the regulating function leads to a Hilbert-Schmidt
operator.

In Sections \ref{standardFourierNeumann} and \ref{standardFourierDirichlet},
we connect the eigenfunction expansions for the Neumann and Dirichlet BCs
to that of (periodic) Fourier expansions on the extended domain $(-2,2)$.
This enables the application of standard results from Fourier theory
to the particular eigenfunction expansions in these cases.

We define the minimal operator 
$A_0: C^2_0(I,\mathbb{C}) \rightarrow L^2_{\mathbb{C}}(I)$ by 
\begin{equation*}
A_0 u := - a_0 \, u^{\, \prime \prime}, 
\end{equation*} 
where $a_0$ is a suitable real number and $u \in C^2_0(I,\mathbb{C})$.  
The operator $A_0$ is densely defined, linear, and symmetric, but not
essentially self-adjoint.  We give self-adjoint extensions $A_0$
by the closure of essentially self-adjoint operators. The extension
process is depicted in Figure \ref{extensionProcess}.

\begin{figure}[t] \label{extensionProcess}
\centering
\scalebox{0.8}{\includegraphics{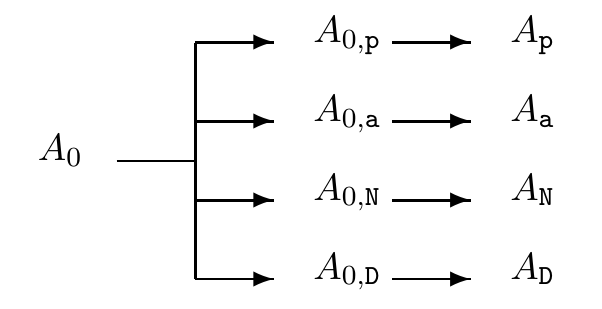}}
\caption{We extend the minimal operator $A_0$, specifying
boundary conditions such as periodic, antiperiodic, Neumann, and 
Dirichlet boundary conditions, to an essentially self-adjoint 
operator $A_{0,\p}, A_{0,\aBC}, A_{0,\N}, A_{0,\D}$, respectively.  
Finally, we arrive at self-adjoint 
operators $A_{\p}, A_{\aBC}, A_{\N}, A_{\D}$ by taking the closure of 
$A_{0,\p}, A_{0,\aBC}, A_{0,\N}, A_{0,\D}$, respectively.}
\end{figure}

\subsection{Periodic Boundary Conditions} 
\label{periodicboundaryconditions}
\begin{flushleft}
We define the operator $A_{0,\p} : D(A_{0,\p}) \rightarrow L^2_{\mathbb{C}}(I)$
by 
\end{flushleft}
\begin{equation*}
D(A_{0,\p}) := \left\{ u \in C^2(\bar{I},{\mathbb{C}}): 
\lim_{x \rightarrow -1} u(x) = \lim_{x \rightarrow 1} u(x) \, \, , \, \, \lim_{x \rightarrow -1} u^{\prime}(x) = \lim_{x \rightarrow 1} u^{\prime}(x)
 \right\}
\end{equation*} 
and  
\begin{equation*}
A_{0,\p} u := - \frac{1}{\pi^2} \, u^{\, \prime \prime} 
\end{equation*}  
for every $u \in D(A_{0,\p})$, where $I := (-1,1)$,
$C^2(\bar{I},{\mathbb{C}})$ 
consists of the restrictions of the elements of
$C^2(J,{\mathbb{C}})$ to $I$, 
where $J$ runs through 
all open intervals of ${\mathbb{R}}$ containing $\bar{I}$. 
Note that 
$C^2(\bar{I},{\mathbb{C}})$ is a {\it dense subspace} of $X$.
$A_{0,\p}$ is densely-defined, linear and symmetric.

\subsubsection{Associated Hilbert Basis and Properties}
We note that $A_{0,\p}$ is a special case of 
a regular Sturm-Liouville operator. In particular, $A_{0,\p}$ is essentially 
self-adjoint. The closure $A_{\p}$ of $A_{0,\p}$ is given by
\begin{equation*}
A_{\p} u = - \frac{1}{\pi^2} \, u^{\, \prime \prime}, 
\end{equation*}
where $\prime$ denotes the weak derivative and $u$ is a restriction to
$I$ of an periodic element of $W^2(\REAL,\CMPLX)$.  $A_{\p}$ has a purely
discrete spectrum $\sigma(A_{\p})$ consisting of the eigenvalues,
\begin{equation*}
\sigma(A_{\p}) = \left\{k^2 : k \in {\mathbb{N}}
\right\} \, \, . 
\end{equation*}
For every $k \in {\mathbb{Z}}$, a normalized eigenvector 
corresponding 
to the eigenvalue $k^2$ is given by  
\begin{equation*}
e_{k}(x) := \frac{1}{\sqrt{2}} \, e^{i \pi k x}.
\end{equation*}
Hence, $(e_k)_{k \in {\mathbb{Z}}}$ is a Hilbert basis of $L^2_{\mathbb{C}}(I)$, $0$ is a simple eigenvalue and for every $k \in {\mathbb{N}}^{*}$, $k^2$ is an eigenvalue of geometric multiplicity $2$,
with corresponding linearly independent eigenvectors $e_k,e_{-k}$. 

\subsubsection{Compactness of $f(A_{\p})$}

For every $f \in B(\sigma(A_{\p}),{\mathbb{C}})$, if 
\begin{equation*}
(|f(k^2)|^2)_{k \in {\mathbb{N}}}
\end{equation*}
is summable, then $f(A_{\p})$ is a Hilbert-Schmidt operator and hence compact. The latter
is the case if 
\begin{equation*}
|f(\lambda)| \leqslant  c \, \lambda^{-\alpha} 
\end{equation*}
for every $\lambda \in \sigma(A_{\p})$, where $\alpha > 1/2$, $c \geqslant 0$.

\subsubsection{Properties of Canonical Convolutions and Integral Representations}
In the following, $*_{\p}$ denotes the convolution in $L^2_{\mathbb{C}}(I)$
that, according to Theorem~\ref{convolutionsinhilbertspaces}, is associated
to the Hilbert basis $(e_{k})_{k \in {\mathbb{Z}}}$. In particular, for
even  $C \in L^2(I)$, $c \in {\mathbb{R}}$,
\begin{align*}
\braket{e_k|C}_2  & = \braket{e_{-k}|C} 
= \frac{1}{\sqrt{2}} \int_{-1}^{1} \cos(\pi k y) \cdot C(y) \, dy, 
\quad k \in {\mathbb{N}}.
\end{align*}
Hence, $c - C *_{\p} \cdot$ is a bounded self-adjoint function of $A_{\p}$. Furthermore, 
if $C$ is in addition positive and 
\begin{equation*}
c := \frac{1}{\sqrt{2}} \int_{-1}^{1} C(y) dy \, \, , 
\end{equation*}
then $c - C *_{\p} \cdot$ is in addition positive, with a spectrum
that contains $0$.

In addition, for $C, u \in L^2_{\mathbb{C}}(I)$
\begin{align*}
(e_{k}(x))^{*} \cdot \braket{C|e_{k}}_{2} &  = 
\frac{1}{2} \cdot e^{- i \pi k x} \cdot \int_{-1}^{1} C^{*}(y)  e^{i \pi k y} \, dy 
= \frac{1}{2} \cdot \int_{-1}^{1} C^{*}(y) e^{i \pi k (y - x)} \, dy \\
& = \int_{-1+x}^{1+x} {\hat C_{\p}}^{*}(y) e^{i \pi k (y - x)} \, dy 
= \frac{1}{2} \cdot \int_{-1}^{1} {\hat C_{\p}}^{*}(y + x) \cdot 
e^{i \pi k y} \, dy \\
& = \frac{1}{2} \cdot \int_{-1}^{1} e^{ - i \pi k y} 
\cdot {\hat C_{\p}}^{*}(x - y ) 
 \, dy 
= \frac{1}{\sqrt{2}}
\braket{e_{k}|{\hat C_{\p}}^{*}({ x - \cdot})}_{2} 
\, \, ,
\end{align*}
where ${\hat C_{\p}}$ denotes the extension of $C$ to a $2$-periodic
function on ${\mathbb{R}}$.  Since for every finite subset
$S \subset {\mathbb{N}}$,
\begin{equation*}
\sum_{k \in S} |(e_k(x))^{*} \cdot \braket{C|e_k}_{2}|^2 \leqslant \sum_{k \in S}
|\braket{e_k|C}_{2}|^2  \leqslant 
\sum_{k \in {\mathbb{N}}} |\braket{e_k|C}_{2}|^2
\, \, ,
\end{equation*}
$ \left(|(e_k(x))^{*} \cdot \braket{C|e_k}_{2}|^2\right)_{k \in \mathbb{N}}$ is
summable.

Hence, we note that   
\begin{align*}
(C *_{\p} u)(x) & = \sum_{l \in {\mathbb{N}}} \braket{e_{\beta(l)}|C}_{2} \braket{e_{\beta(l)}|u}_{2} . \, e_{\beta(l)}(x)  =
\big \langle \sum_{l \in {\mathbb{N}}} (e_{\beta(l)}(x))^{*} \cdot \braket{C|e_{\beta(l)}}_{2} . e_{\beta(l)}|u \big \rangle_2 \\
& =
\big \langle \sum_{l \in {\mathbb{N}}} \frac{1}{\sqrt{2}}
\braket{e_{\beta(l)}|{\hat C_{\p}}^{*}({ x - \cdot})}_{2} . e_{\beta(l)}|u \big \rangle_2  =
\frac{1}{\sqrt{2}} \, \, \big \langle \sum_{l \in {\mathbb{N}}}
\braket{e_{\beta(l)}|{\hat C_{\p}}^{*}({ x - \cdot})}_{2} . e_{\beta(l)}|u \big \rangle_2 \\
& =  \frac{1}{\sqrt{2}} \, \braket{{\hat C_{\p}}^{*}({x - \cdot})|u}_{2}
= \frac{1}{\sqrt{2}}  \int_{-1}^{1} {\hat C_{\p}}(x - y) \cdot u(y) \, dy
\, \, ,
\end{align*}
where $\beta : {\mathbb{N}} \rightarrow {\mathbb{Z}}$ is some bijection.

\subsection{Antiperiodic Boundary Conditions} 
\label{antiperiodicboundaryconditions}
\begin{flushleft}
We define the operator $A_{0,\aBC} : D(A_{0,\aBC}) \rightarrow L^2_{\mathbb{C}}(I)$
by 
\end{flushleft}
\begin{equation*}
D(A_{0,\aBC}) := \left\{ u \in C^2(\bar{I},{\mathbb{C}}): 
\lim_{x \rightarrow -1} u(x) = - \lim_{x \rightarrow 1} u(x) \, \, , \, \, \lim_{x \rightarrow -1} u^{\prime}(x) = - \lim_{x \rightarrow 1} u^{\prime}(x)
 \right\}
\end{equation*} 
and  
\begin{equation*}
A_{0,\aBC} u := - \frac{1}{\pi^2} \, u^{\, \prime \prime} 
\end{equation*}  
for every $u \in D(A_{0,\aBC})$, where $I := (-1,1)$,
$C^2(\bar{I},{\mathbb{C}})$ 
consists of the restrictions of the elements of
$C^2(J,{\mathbb{C}})$ to $I$, 
where $J$ runs through 
all open intervals of ${\mathbb{R}}$ containing $\bar{I}$. 
Note that 
$C^2(\bar{I},{\mathbb{C}})$ is a {\it dense subspace} of $X$.
$A_{0,\aBC}$ is densely-defined, linear and symmetric.

\subsubsection{Associated Hilbert Basis and Properties}
We note that  $A_{0,\aBC}$ is a special case of 
a regular Sturm-Liouville operator. In particular, $A_{0,\aBC}$ is essentially 
self-adjoint. The closure $A_{\aBC}$ of $A_{0,\aBC}$ is given by
\begin{equation*}
A_{\aBC} u = - \frac{1}{\pi^2} \, u^{\, \prime \prime}, 
\end{equation*}
where $\prime$ denotes the weak derivative
and $u$ is a restriction to $I$ of an antiperiodic element of $W^2(\REAL,\CMPLX)$.
$A_{\aBC}$ has a purely discrete spectrum
$\sigma(A_{\aBC})$  consisting of the eigenvalues,
\begin{equation*}
\sigma(A_{\aBC}) = \left\{\left(k + \frac{1}{2} \right)^{\!2} : k \in {\mathbb{N}}
\right\} \, \, . 
\end{equation*}
For every $k \in {\mathbb{Z}}$, a normalized eigenvector 
corresponding 
to the eigenvalue $\left(k + \frac{1}{2}\right)^{\!2}$ is given by  
\begin{equation*}
e_{k}(x) := \frac{1}{\sqrt{2}} \, e^{i \pi \left(k + \frac{1}{2}\right) x}
\end{equation*}

Hence $(e_k)_{k \in {\mathbb{Z}}}$ is a Hilbert basis of $L^2_{\mathbb{C}}(I)$, and for every $k \in {\mathbb{N}}$, $(k+ (1/2))^2$ is an eigenvalue of geometric multiplicity $2$, with 
corresponding linearly independent eigenvalues $e_{k}, e_{-k-1}$.

\subsubsection{Compactness of $f(A_{\aBC})$}
For every $f \in B(\sigma(A_{\aBC}),{\mathbb{C}})$, if 
\begin{equation*}
\left(\left|f\left(\left[k + \frac{1}{2}\right]^2 \right)\right|^2\right)_{k \in {\mathbb{Z}}}
\end{equation*}
is summable, then $f(A_{\aBC})$ is a Hilbert-Schmidt operator and hence compact. 
In particular, the latter
is the case if 
\begin{equation*}
|f(\lambda)| \leqslant  c \, \lambda^{-\alpha} 
\end{equation*}
for every $\lambda \in \sigma(A_{\aBC})$, where $\alpha > 1/2$, $c \geqslant 0$.

\subsubsection{Properties of Canonical Convolutions and Integral Representations}
In the following, $*_{\aBC}$ denotes the convolution in $L^2_{\mathbb{C}}(I)$
that, according to Theorem~\ref{convolutionsinhilbertspaces}, is associated
to the Hilbert basis $(e_{k})_{k \in {\mathbb{Z}}}$. In particular, for even $C \in L^2(I)$ and $c \in {\mathbb{R}}$, 
\begin{align*}
\braket{e_k|C}_2  & = \braket{e_{-k-1}|C}_2 = \frac{1}{\sqrt{2}} \int_{-1}^{1} \cos\left[\pi \left(k+\frac{1}{2}\right) y\right]
\cdot C(y) \, dy, \quad k \in {\mathbb{N}}.
\end{align*}
Hence all members of the sequence $(\braket{e_k|C}_2)_{k \in {\mathbb{Z}}}$ 
are real-valued.
Therefore $c - C *_{\aBC} \cdot$ is a self-adjoint bounded function of $A_{\aBC}$. Furthermore, if $C$
is in addition positive and 
\begin{equation*}
c := \frac{1}{\sqrt{2}} \int_{-1}^{1} C(y) dy \, \, , 
\end{equation*}
then $c - C *_{\aBC} \cdot$ is in addition positive. 
\newline
\linebreak
In addition, for $C, u \in L^2_{\mathbb{C}}(I)$ 
\begin{align*}
(e_{k}(x))^{*} \cdot \braket{C|e_{k}}_{2} & = 
\frac{1}{2} \cdot e^{- i \pi \left(k+\frac{1}{2}\right) x} \cdot \int_{-1}^{1} C^{*}(y)  e^{i \pi \left(k + \frac{1}{2}\right)y} \, dy 
= \frac{1}{2} \cdot \int_{-1}^{1} C^{*}(y) e^{i \pi \left(k+\frac{1}{2}\right) (y - x)} \, dy \\
& = \int_{-1+x}^{1+x} {\hat C_{\aBC}}^{*}(y) e^{i \pi \left(k+\frac{1}{2}\right) (y - x)} \, dy= \frac{1}{2} \cdot \int_{-1}^{1} {\hat C_{\aBC}}^{*}(y + x) \cdot 
e^{i \pi \left(k+\frac{1}{2}\right) y} \, dy \\
& = 
\frac{1}{2} \cdot \int_{-1}^{1} e^{ - i \pi \left(k+\frac{1}{2}\right) y} 
\cdot {\hat C_{\aBC}}^{*}(x - y ) 
 \, dy
= \frac{1}{\sqrt{2}}
\braket{e_{k}|{\hat C_{\aBC}}^{*}({ x - \cdot})}_{2} 
\, \, , \\
\end{align*}
where ${\hat C_{\aBC}}$ denotes the extension of $C$ to a $2$-antiperiodic function 
on ${\mathbb{R}}$. 
Since for every finite subset $S \subset {\mathbb{N}}$,
$$
\sum_{k \in S} |(e_k(x))^{*} \cdot \braket{C|e_k}_{2}|^2 \leqslant \sum_{k \in S}
|\braket{e_k|C}_{2}|^2  \leqslant 
\sum_{k \in {\mathbb{N}}} |\braket{e_k|C}_{2}|^2
\, \, , 
$$ 
$ \left(|(e_k(x))^{*} \cdot \braket{C|e_k}_{2}|^2\right)_{k \in \mathbb{N}}$ is
summable.

Hence,   
we note that 
\begin{align*}
(C *_{\aBC} u)(x) & = \sum_{l \in {\mathbb{N}}} \braket{e_{\beta(l)}|C}_{2} \braket{e_{\beta(l)}|u}_{2} . \, e_{\beta(l)}(x)  =
\big \langle \sum_{l \in {\mathbb{N}}} (e_{\beta(l)}(x))^{*} \cdot \braket{C|e_{\beta(l)}}_{2} . e_{\beta(l)}|u \big \rangle_2 \\
& =
\big \langle \sum_{l \in {\mathbb{N}}} \frac{1}{\sqrt{2}}
\braket{e_{\beta(l)}|{\hat C_{\aBC}}^{*}({ x - \cdot})}_{2} . e_{\beta(l)}|u \big \rangle_2 =
\frac{1}{\sqrt{2}} \, \, \big \langle \sum_{l \in {\mathbb{N}}}
\braket{e_{\beta(l)}|{\hat C_{\aBC}}^{*}({ x - \cdot})}_{2} . e_{\beta(l)}|u \big \rangle_2 \\
& =  \frac{1}{\sqrt{2}} \, \braket{{\hat C_{\aBC}}^{*}({x - \cdot})|u}_{2}
= \frac{1}{\sqrt{2}}  \int_{-1}^{1} {\hat C_{\aBC}}(x - y) \cdot u(y) \, dy
\, \, .
\end{align*}
where $\beta : {\mathbb{N}} \rightarrow {\mathbb{Z}}$ is some bijection.

\subsection{Neumann Boundary Conditions}
\label{neumannboundaryconditions}

\begin{flushleft}
We define the operator $A_{0,\N} : D(A_{0,\N}) \rightarrow L^2_{\mathbb{C}}(I)$
by 
\end{flushleft}
\begin{equation*}
D(A_{0,\N}) := \left\{ u \in C^2(\bar{I},{\mathbb{C}}): 
\lim_{x \rightarrow -1} u^{\prime}(x) = \lim_{x \rightarrow 1} u^{\prime}(x) = 0
 \right\}
\end{equation*} 
and  
\begin{equation*}
A_{0,\N} u := - \frac{4}{\pi^2} \, u^{\, \prime \prime} 
\end{equation*}  
for every $u \in D(A_{0,\N})$, where $I := (-1,1)$,
$C^2(\bar{I},{\mathbb{C}})$ 
consists of the restrictions of the elements of
$C^2(J,{\mathbb{C}})$ to $I$, 
where $J$ runs through 
all open intervals of ${\mathbb{R}}$ containing $\bar{I}$. 
Note that 
$C^2(\bar{I},{\mathbb{C}})$ is a {\it dense subspace} of $X$.
$A_{0,\N}$ is densely-defined, linear and positive symmetric.

\subsubsection{Associated Hilbert Basis and Properties}

We note that  $A_{0,\N}$ is a special case of 
a regular Sturm-Liouville operator. In particular, $A_{0,\N}$ is essentially 
self-adjoint.  The closure $A_{\N}$ of $A_{0,\N}$ is given by
\begin{equation*}
A_{\N} u = - \frac{4}{\pi^2} \, u^{\, \prime \prime}, 
\end{equation*}
where $u \in W^2(I,\CMPLX)$ where $\prime$ denotes the weak derivative
and $u \in W^1_0(I,\CMPLX)$.  $A_{\N}$ has a purely discrete spectrum
$\sigma(A_{\N})$  consisting of simple eigenvalues,
\begin{equation*}
\sigma(A_{\N}) = \left\{k^2 : k \in {\mathbb{N}}
\right\} \, \, . 
\end{equation*}
For every $k \in {\mathbb{N}}$, a normalized eigenvector corresponding 
to the eigenvalue $k^2$ is given by  
\begin{equation*}
e_{k}(x) := 
\begin{cases}
\quad \qquad \frac{1}{\sqrt{2}} & \text{if $k = 0$} \\
\cos\left(\frac{k \pi}{2}(x+1)\right) & \text{if $k \neq 0$}
\end{cases}
\, \, , 
\end{equation*}
for every $x \in I$,  
\begin{equation*}
A_{0,\N} e_{k} = \frac{4}{\pi^2} \, \frac{\pi^2 k^2}{4} \, e_{k}  =
k^2 \, e_{k}  \, \, .
\end{equation*}
Hence $e_0,e_1,\dots $ is a Hilbert basis of $L^2_{\mathbb{C}}(I)$.
Furthermore, we note for $k \in {\mathbb{N}}^{*}$
and $x \in I$ that 
\begin{align*}
e_{2k}(x) & =  (-1)^{k} \cos(k \pi x) \, \, , \\
e_{2k - 1}(x) & =  
(-1)^{k} \sin\left( \pi \left(k -  \frac{1}{2}\right) x \right),
\end{align*}
and hence that 
\begin{align*}
& \text{$e_k$ is even and periodic with period $2$, for even  $k \in {\mathbb{N}}$} \, \, , \\
& \text{$e_k$ is odd and antiperiodic with period $2$, for odd $k \in {\mathbb{N}}$} \, \, .
\end{align*}
Also, $\cos\left(\frac{k \pi}{2}({\textrm{id}}_{\mathbb{R}}+1)\right)$
is periodic with period $4$ for every
$k \in {\mathbb{N}}$. 

\subsubsection{Compactness of $f(A_{\N})$}
For every $f \in B(\sigma(A_{\N}),{\mathbb{C}})$, if 
\begin{equation*}
(|f(k^2)|^2)_{k \in {\mathbb{N}}}
\end{equation*}
is summable, then $f(A_{\N})$ is a Hilbert-Schmidt operator and hence compact. The latter
is the case if 
\begin{equation*}
|f(\lambda)| \leqslant  c \, \lambda^{-\alpha} 
\end{equation*}
for every $\lambda \in \sigma(A_{\N})\setminus \{0\}$, where $\alpha > 1/2$, $c \geqslant 0$.

\subsubsection{Properties of Simple Convolutions and Integral Representations} 
\label{simpleConvoNeumann}

In the following, we give connections to the 
convolutions $*_{\p}$ from Section~\ref{periodicboundaryconditions},
for periodic BCs,
and
$*_{\aBC}$ from Section~\ref{antiperiodicboundaryconditions}, for
antiperiodic BCs. For every $k \in {\mathbb{Z}}$, $x \in I$,  the
corresponding eigenfunctions are as follows:
\begin{equation*}
e_{k}^{\p}(x) := \frac{1}{\sqrt{2}} \, e^{i \pi k x} \, \, , \, \,
e_{k}^{\aBC}(x) := \frac{1}{\sqrt{2}} \, e^{i \pi \left(k + \frac{1}{2}\right) x} \, \, .
\end{equation*}  
We note for even $C \in L^2(I)$, even $u \in L^2_{\mathbb{C}}(I)$, $x \in I$, $k \in {\mathbb{N}}^{*}$ that 
\begin{align*}
\braket{e_k^{\p}|C}_2 & = \braket{e_{-k}^{\p}|C}_2 = 
\frac{1}{\sqrt{2}} \, \int_{-1}^{1} \cos(\pi k y) C(y) \, dy 
\, \, \, \, \left( \leqslant 
\frac{1}{\sqrt{2}} \int_{-1}^{1} C(y) \, dy
\right) \, \, ,
\end{align*}
\begin{align*}
\braket{e_k^{\p}|u}_2  & = \braket{e_{-k}^{\p}|u}_2 = 
\frac{(-1)^{k}}{\sqrt{2}} \braket{e_{2k}|u}_2 \, \, .
\end{align*}
As a consequence, for $k \in {\mathbb{N}}^{*}$
\begin{align*}
& \braket{e_k^{\p}|C}_2 \braket{e_k^{\p}|u}_2 e_k^{\p}(x) + 
\braket{e_{-k}^{\p}|C}_2 \braket{e_{-k}^{\p}|u}_2 e_{-k}^{\p}(x) \\
& = \sqrt{2} \, (-1)^{k} \braket{e_k^{\p}|C}_2 \braket{e_k^{\p}|u}_2 e_{2k}(x) = \braket{e_k^{\p}|C}_2 \braket{e_{2k}|u}_2 e_{2k}(x) \, \, .
\end{align*}
Hence 
\begin{align*}
C *_{\p} u  & = \sum_{k \in {\mathbb{Z}}} \braket{e_k^{\p}|C}_2 \braket{e_k^{\p}|u}_2 e_{k}^{\p}
= 
\sum_{k = 0}^{\infty} \braket{e_k^{\p}|C}_2 \braket{e_{2k}|u}_2 e_{2k} \\
& = \sum_{k = 0}^{\infty} \varphi_1(k^2) \braket{e_{k}|u}_2 e_{k} \, \, , 
\end{align*}
where $\varphi_1 \in B(\sigma(A_{\N}),{\mathbb{C}})$ is defined by 
\begin{equation*}
\varphi_1(k^2) := 
\begin{cases}
\quad \, \, 0 & \text{if $k \in {\mathbb{N}}$ is odd} \\
\braket{e_{k/2}^{\p}|C}_2 & \text{if $k \in {\mathbb{N}}$ is even} 
\end{cases} \, \, ,
\end{equation*}
and 
\begin{equation*}
C *_{\p} P_{\text{even}} = \varphi_1(A_{\N}) \, \, , 
\end{equation*}
where the orthogonal projection $P_{\text{even}} : L^2_{\mathbb{C}}(I) \rightarrow L^2_{\mathbb{C}}(I)$ is defined by 
\begin{equation*}
P_{\text{even}} h := \frac{1}{2} \left(h + h \circ (-{\text{id}}_{I}) \right) 
\, \, ,
\end{equation*}
for every $h \in L^2_{\mathbb{C}}(I)$. \\
\linebreak
Also, we note for even $C \in L^2(I)$ and 
odd $u \in L^2_{\mathbb{C}}(I)$, $x \in I$, $k \in {\mathbb{N}}$ that 
\begin{align*}
\braket{e_k^{\aBC}|C}_2 & = \braket{e_{-k-1}^{\aBC}|C}_2 = \frac{1}{\sqrt{2}} \int_{-1}^{1} \cos\left[\pi \left(k+\frac{1}{2}\right) y\right] 
\cdot C(y) \, dy  \, \, \, \, \left( \leqslant 
\frac{1}{\sqrt{2}} \int_{-1}^{1} C(y) \, dy
\right) 
\, \, , \\
\braket{e_k^{\aBC}|u}_2 & = 
- \braket{e_{-k-1}^{\aBC}|u}_2 = 
\frac{(-1)^{k} \, i}{\sqrt{2}} \braket{e_{2k+1}|u}_2  \, \, .
\end{align*}
As a consequence, 
\begin{align*}
& \braket{e_k^{\aBC}|C}_2 \braket{e_k^{\aBC}|u}_2 e_k^{\aBC}(x) + 
\braket{e_{-k-1}^{\aBC}|C}_2 \braket{e_{-k-1}^{\aBC}|u}_2 e_{-k-1}^{\aBC}(x) = 
\braket{e_k^{\aBC}|C}_2 \braket{e_{2k+1}|u}_2 e_{2k+1}(x) \, \, .
\end{align*}
Hence,
\begin{align*}
C *_{\aBC} u  & = \sum_{k \in {\mathbb{Z}}} \braket{e_k^{\aBC}|C}_2 \braket{e_k^{\aBC}|u}_2 e_{k}^{\aBC} = 
\sum_{k=0}^{\infty} \braket{e_k^{\aBC}|C}_2 \braket{e_{2k+1}|u}_2 e_{2k+1}  \\
& = 
\sum_{k=0}^{\infty} \varphi_2(k^2) \braket{e_{k}|u}_2 e_{k} \, \, ,
\end{align*}
where $\varphi_2 \in B(\sigma(A_{\N}),{\mathbb{C}})$ is defined by 
\begin{equation*}
\varphi_2(k^2) := 
\begin{cases}
\braket{e_{(k-1)/2}^{\aBC}|C}_2 & \text{if $k \in {\mathbb{N}}^{*}$ is odd} \\
\qquad \, \, \, \, 0 & \text{if $k \in {\mathbb{N}}^{*}$ is even} 
\end{cases} \, \, ,
\end{equation*}
and 
\begin{equation*}
C *_{\aBC} P_{\text{odd}} = \varphi_2(A_{\N}) \, \, , 
\end{equation*}
where the orthogonal projection $P_{\text{odd}} : L^2_{\mathbb{C}}(I) \rightarrow L^2_{\mathbb{C}}(I)$ is defined by 
\begin{equation*}
P_{\text{odd}} h := \frac{1}{2} \left(h - h \circ (-{\text{id}}_{I}) \right) 
\, \, ,
\end{equation*}
for every $h \in L^2_{\mathbb{C}}(I)$.

\subsubsection{Properties of Canonical Convolutions and Integral Representations} 
In the following, $*_{\N}$ denotes the convolution in $L^2_{\mathbb{C}}(I)$
that, according to Theorem~\ref{convolutionsinhilbertspaces}, is associated
to the Hilbert basis $(e_{k})_{k \in {\mathbb{N}}^{*}}$. In particular, for $C \in L^2(I)$ and $c \in {\mathbb{R}}$,
\begin{equation*}
\braket{e_k|C}_2 = \int_{-1}^{1} 
\cos\left(\frac{k \pi}{2}(y+1)\right) C(y)  \, dy \, \, , 
\end{equation*}
is real-valued for every $k \in {\mathbb{N}}$ and $c - C *_{\N} \cdot$ is a bounded self-adjoint function of $A_{\N}$. Furthermore, since 
\begin{equation*}
\int_{-1}^{1} C(y) \, dy - 
\braket{e_k|C}_2 = \int_{-1}^{1} \left[1 -
\cos\left(\frac{k \pi}{2}(y+1)\right) \right] C(y)  \, dy \, \, , 
\end{equation*}
for every $k \in {\mathbb{N}}^{*}$,  
if $C \geqslant 0$ and 
\begin{equation*}
c = \int_{-1}^{1} C(y) \, dy \, \, , 
\end{equation*}
then the operator 
$c - C *_{\N} \cdot$
is in particular positive.
\newline
\linebreak
If $C$ is even, 
\begin{equation*}
\braket{e_k|C}_2 = 0 \, \, , 
\end{equation*}
for every odd $k \in {\mathbb{N}}$. 
\newline
\linebreak
In addition, since for every $C, g \in L^2_{\mathbb{C}}(I)$, $x \in I$ and every finite subset
$S \subset {\mathbb{N}}$
\begin{align*}
& \sum_{k \in S} |(e_k(x))^{*} \cdot \braket{C|e_k}_2|^2 \leqslant \sum_{k \in S}
|\braket{e_k|C}_2|^2  \leqslant 
\sum_{k \in {\mathbb{N}}^{*}} |\braket{e_k|C}_2|^2
\, \, , 
\end{align*}
we note that 
\begin{equation*}
(C *_{\N} u)(x) = \sum_{k \in {\mathbb{N}}} \braket{e_k|C}_2 \braket{e_k|u}_2 . \, e_k(x) =
\big \langle \sum_{k \in {\mathbb{N}}} (e_k(x))^{*} \cdot \braket{C|e_k}_2 . e_k|u \rangle_2
\end{equation*}
and, 
since for $k \in {\mathbb{N}}$, $x, u \in {\mathbb{R}}$
\begin{align*}
& \cos\left(\frac{k \pi}{2}(x+1)\right) \cos\left(\frac{k \pi}{2}(y+1)\right) \\
& = \frac{1}{2} \left\{\left[\cos\left(\frac{k \pi}{2}(x - y + 1)\right) 
+ \cos\left(\frac{k \pi}{2}(x + y + 1)\right) \right]
\cos\left(\frac{k \pi}{2}\right) \right. \\
& \left.\qquad \, \, \, \, + \left[\sin\left(\frac{k \pi}{2}(x - y + 1)\right) - \sin\left(\frac{k \pi}{2}(x + y + 1)\right) \right]
\sin\left(\frac{k \pi}{2}\right)
\right\} \, \, ,
\end{align*}
for {\it even} $C$, $k \in {\mathbb{N}}^{*}$, $x \in I$ that 
\begin{align*}
& (e_k(x))^{*} \cdot \braket{C|e_k}_2 = 
\cos\left(\frac{k \pi}{2}(x+1)\right) \int_{-1}^{1} C^{*}(y) \cos\left(\frac{k \pi}{2}(y+1)\right) dy \\
& = \frac{1}{2} \, \cos\left(\frac{k \pi}{2}\right) \left[ 
\int_{-1}^{1} C^{*}(y) \,
\cos\left(\frac{k \pi}{2}(x - y + 1)\right)
dy \right. \\
& \left. \qquad \qquad \qquad \quad \, \, +
\int_{-1}^{1} C^{*}(y) \,
\cos\left(\frac{k \pi}{2}(x + y + 1)\right) dy
\right] \\
& = \cos\left(\frac{k \pi}{2}\right) \int_{-1}^{1} {\hat C_{\p}}^{*}(y) \,
\cos\left(\frac{k \pi}{2}(x - y + 1)\right) 
dy \\
& = \cos\left(\frac{k \pi}{2}\right) \int_{x-1}^{x+1} {\hat C_{\p}}^{*}(y) \,
\cos\left(\frac{k \pi}{2}(x - y + 1)\right) 
dy \\
& =  \cos\left(\frac{k \pi}{2} \right) \int_{-1}^{1} {\hat C_{\p}}^{*}(x - y) \,
\cos\left(\frac{k \pi}{2}(y + 1)\right) 
dy \\
& = \cos\left(\frac{k \pi}{2} \right)  \braket{e_k|{\hat{C_{\p}}}^{*}(x - {\textrm{id}}_{\mathbb{R}})}_2
\, \, , 
\end{align*}
where ${\hat C_{\p}}$ denotes the extension of $C$ to a $2$-{\it periodic} function 
on ${\mathbb{R}}$, i.e., such that 
\begin{equation*}
{\hat C_{\p}}(x + 2) = {\hat C_{\p}}(x)
\end{equation*}
for every $x \in I$. Here, it has been used that
$ \cos\left(\frac{k \pi}{2}({\textrm{id}}_{\mathbb{R}} + 1)\right)$
is $2$-{\it periodic} for {\it even} $k \in {\mathbb{N}}^{*}$.
\newline
\linebreak
We note for $k \in {\mathbb{N}}^{*}$ that 
\begin{align*}
\cos\left(\frac{k \pi}{2}\right) =
\begin{cases}
\, \, \, \, 0 & \text{if $k$ is odd,} \\
\, \, \, \, 1 & \text{if $k$ is even and $k/2$ is even,} \\
-1 & \text{if $k$ is even and $k/2$ is odd.} 
\end{cases} \, \, . 
\end{align*}
In the next step, we decompose $C$ into $C_1,C_2 \in 
L^2_{\mathbb{C}}(I)$, where 
\begin{align*}
C_1(x) := \frac{1}{2} \left[C(|x|) + C(1 - |x|)
 \right] \, \, , \, \, C_2(x) := \frac{1}{2} \left[ C(|x|) - C(1 - |x|) 
 \right] \, \, , 
\end{align*}
such that 
\begin{equation*}
C = C_1 + C_2 \, \, .
\end{equation*} 
Note that $C_1,C_2$ are even and 
have a so called ``half-wave symmetry,'' i.e, that for every 
$x \in [0,1/2]$:
\begin{align*}
& C_1(1-x) = \frac{1}{2} \left[C(|1-x|) + C(1 - |1-x|)
 \right] =
\frac{1}{2} \left[C(1-x) + C(x) 
 \right] = C_1(x) \, \, , \\
& C_2(1-x) = \frac{1}{2} \left[C(|1-x|) - C(1 - |1-x|)
 \right] =
\frac{1}{2} \left[C(1-x) - C(x) 
 \right] = - C_2(x) \, \, .
\end{align*}
As a consequence, for even $k \in {\mathbb{N}}^{*}$, $j \in \{1,2\}$,
\begin{align*}
& \braket{e_k|C_j}_2 = \int_{-1}^{1} \cos\left(\frac{k \pi}{2}(y + 1)\right) C_{j}(y) \, dy \\
& = \int_{-1}^{1} \left[\cos\left(\frac{k \pi}{2} y \right)
\cos\left(\frac{k \pi}{2}\right) - \sin\left(\frac{k \pi}{2} y \right)
\sin\left(\frac{k \pi}{2}\right)\right] C_{j}(y) \, dy \\
& = \cos\left(\frac{k \pi}{2}\right) \int_{-1}^{1} \cos\left(\frac{k \pi}{2} y \right)
C_{j}(y) \, dy = 2 \cos\left(\frac{k \pi}{2}\right) \int_{0}^{1} \cos\left(\frac{k \pi}{2} y \right)
C_{j}(y) \, dy \\
& = 2 \cos\left(\frac{k \pi}{2}\right) \left[ 
\int_{0}^{1/2} \cos\left(\frac{k \pi}{2} y \right)
C_{j}(y) \, dy + \int_{1/2}^{1} \cos\left(\frac{k \pi}{2} y \right)
C_{j}(y) \, dy \right] \\
& = 2 \cos\left(\frac{k \pi}{2}\right) \left[ 
\int_{0}^{1/2} \cos\left(\frac{k \pi}{2} y \right)
C_{j}(y) \, dy + \int_{0}^{1/2} \cos\left(\frac{k \pi}{2} (1 - y) \right)
C_{j}(1 - y) \, dy \right] \\
& = 2 \cos\left(\frac{k \pi}{2}\right) \left[ 
\int_{0}^{1/2} \cos\left(\frac{k \pi}{2} y \right)
C_{j}(y) \, dy + (-1)^{j+1} (-1)^{k/2} 
\int_{0}^{1/2} \cos\left(y\right)
C_{j}(y) \, dy \right] \\
& = 2 \cos\left(\frac{k \pi}{2}\right) \left[1 + (-1)^{j+1} (-1)^{k/2} \right] \int_{0}^{1/2} \cos\left(y\right)
C_{j}(y) \, dy \, \, , 
\end{align*}
and hence 
\begin{align*}
\braket{e_k|C_j}_2 = 
\begin{cases}
0 & \text{if $k$ odd} \\
0 & \text{if $k$ even, $k/2$ odd and $j=1$} \\
0 & \text{if $k$ even, $k/2$ is even and $j=2$}
\end{cases} \, \, .
\end{align*}
From this, we conclude for even $k \in {\mathbb{N}}^{*}$, $x \in I$
that 
\begin{align} \label{evenk}
(e_k(x))^{*} \cdot \braket{C|e_k}_2 & = (e_k(x))^{*} \cdot \braket{C_1|e_k}_2 + (e_k(x))^{*} \cdot \braket{C_2|e_k}_2 \nonumber \\
& =
\braket{e_k|{\hat{C}_{1,\p}}^{*}(x - {\textrm{id}}_{\mathbb{R}})}_2
- \braket{e_k|{\hat{C}_{2,\p}}^{*}(x - {\textrm{id}}_{\mathbb{R}})}_2
\, \, , 
\end{align}
where for every $j \in \{1,2\}$, ${\hat C}_{j,\p}$ denotes the extension of $C_j$ to a $2$-{\it periodic} function 
on ${\mathbb{R}}$.
\newline
\linebreak 
In the following, we extend (\ref{evenk}) to {\it odd} $k \in {\mathbb{N}}^{*}$. For this purpose, 
we note for every even 
$g \in L^2_{\mathbb{C}}(I)$ that its 2-periodic extension ${\hat g}$ is even, too. For the proof, let $l \in {\mathbb{N}}$, 
$x \in [-1-2l,-2l]$. Then 
\begin{equation*}
{\hat u}(x) = u(x+2l) = u(-x-2l) = 
{\hat u}(-x-2l + 2l) = {\hat u}(-x) \, \, .
\end{equation*}
Hence it follows for even
$u \in L^2_{\mathbb{C}}(I)$, $x \in I$ and odd $k \in {\mathbb{N}}^{*}$ that 
\begin{align*}
\braket{e_k|{\hat u}(x - \cdot)}_2 & =
- \braket{e_k|{\hat u}(- x - \cdot)}_2 \, \, ,
\end{align*}
which implies that 
\begin{equation*}
\braket{e_k|{\hat u}(x - \cdot) + {\hat u}(- x - \cdot)}_2 = 0
\, \, .
\end{equation*}
On the other hand, for even $k \in {\mathbb{N}}^{*}$ 
\begin{align*}
\braket{e_k|{\hat u}(x - \cdot)}_2 &  =
\braket{e_k|{\hat u}(- x - \cdot)}_2 \, \, .
\end{align*}
As a consequence, 
\begin{equation*}
\braket{e_k|\frac{1}{2}\,[{\hat u}(x - \cdot)+ {\hat u}(- x - \cdot)]}_2
= 
\begin{cases}
\qquad \, \, \, \, 0 & \text{if $k \in {\mathbb{N}}^{*}$ is odd} \\
\braket{e_k|{\hat u}(x - \cdot)}_2 & \text{if $k \in {\mathbb{N}}^{*}$ is even}
\end{cases} \, \, .
\end{equation*}
Therefore, we conclude from (\ref{evenk}) that 
for {\it even} $C$ and $k \in {\mathbb{N}}^{*}$, $x \in I$ 
\begin{align*} 
& (e_k(x))^{*} \cdot \braket{C|e_k}_2
= \braket{e_k|\frac{1}{2} \, [{\hat{C}_{1,}}^{*}(x - \cdot
) + {\hat{C}_{1,\p}}(- x - \cdot
) - {\hat{C}_{2,\p}}^{*}(x - \cdot) -
{\hat{C}_{2,\p}}^{*}(- x - \cdot)
]}_2 \, \, .
\end{align*}
Furthermore,
\begin{align*}
& \braket{e_0|\frac{1}{2} \, [{\hat{C}_{1,\p}}^{*}(x - \cdot
) + {\hat{C}_{1,\p}}(- x - \cdot
) - {\hat{C}_{2,\p}}^{*}(x - \cdot) -
{\hat{C}_{2,\p}}^{*}(- x - \cdot)
]}_2 \\
& = 2^{-3/2} \, \int_{-1}^{1} {\hat{C}_{1,\p}}^{*}(x - y) \, dy  +  2^{-3/2} \,  \int_{-1}^{1} {\hat{C}_{1,\p}}(- x - y) \, dy \\
& \quad \, -  2^{-3/2} \, \int_{-1}^{1} {\hat{C}_{2,\p}}^{*}(x - y) \, dy -
 2^{-3/2} \, \int_{-1}^{1} {\hat{C}_{2,\p}}^{*}(- x - y)
]\, dy  \\
& = 2^{-3/2} \, \int_{-1}^{1} {\hat{C}_{1,\p}}^{*}(y - x) \, dy  +  2^{-3/2} \,  \int_{-1}^{1} {\hat{C}_{1,\p}}(y + x) \, dy \\
& \quad \, -  2^{-3/2} \, \int_{-1}^{1} {\hat{C}_{2,\p}}^{*}(y - x) \, dy -
 2^{-3/2} \, \int_{-1}^{1} {\hat{C}_{2,\p}}^{*}(y + x)
\, dy  \\
& = 2^{-3/2} \, \int_{-1}^{1} C_1^{*}(y) \, dy  +  2^{-3/2} \,  \int_{-1}^{1} C_1^{*}(y) \, dy \\
& \quad \, -  2^{-3/2} \, \int_{-1}^{1} C_2^{*}(y) \, dy -
 2^{-3/2} \, \int_{-1}^{1} C_2^{*}(y)
\, dy  \\
& = 2^{-1/2} \, \int_{-1}^{1} C_1^{*}(y) \, dy  -  2^{-1/2} \, \int_{-1}^{1} C_2^{*}(y) \, dy \\
& = (e_0(x))^{*} \cdot \braket{f|e_0}_2 + \frac{\sqrt{2} - 1}{2} \, \int_{-1}^{1} C_1^{*}(y) \, dy  - \frac{\sqrt{2} + 1}{2} \, \int_{-1}^{1} C_2^{*}(y) \, dy 
\end{align*}
and hence 
\begin{align*}
(e_0(x))^{*} \cdot \braket{C|e_0}_2 e_0 = & \braket{e_0|\frac{1}{2} \, [{\hat{C}_{1,\p}}^{*}(x - \cdot
) + {\hat{C}_{1,\p}}(- x - \cdot
) - {\hat{C}_{2,\p}}^{*}(x - \cdot) -
{\hat{C}_{2,\p}}^{*}(- x - \cdot)
]}_2 e_0 + k_{N,C} e_0 \, \, , 
\end{align*}
where 
\begin{align*}
k_{N,C} :=  - \frac{\sqrt{2} - 1}{2} \, \int_{-1}^{1} C_1^{*}(y) \, dy  + \frac{\sqrt{2} + 1}{2} \, \int_{-1}^{1} C_2^{*}(y) \, dy 
\, \, . 
\end{align*} 
As a consequence, 
\begin{align*}
C *_{\N} u (x) & =
\sum_{k \in {\mathbb{N}}} \braket{e_k|C}_2 \braket{e_k|u}_2 . \, e_k(x) =
\big \langle \sum_{k \in {\mathbb{N}}} (e_k(x))^{*} \cdot \braket{C|e_k}_2 . e_k|u \big \rangle_2 \\
& = 
\big \langle \sum_{k \in {\mathbb{N}}} \langle e_k|
\frac{1}{2} \, [C_1^{*}(x - \cdot) + C_1^{*}(-x - \cdot)
- C_2^{*}(x - \cdot) - C_2^{*}(-x - \cdot)
]
\rangle_2 . e_k|g \big \rangle_2 + k_{N,C} \braket{e_0|u}_2 \\
& = \frac{1}{2} \, 
\int_{-1}^{1} C_1(x - y) u(y) \, dy 
+ \frac{1}{2} \, 
\int_{-1}^{1} C_1(- x - y) u(y) \, dy \\
& \quad \, \, - \frac{1}{2} \,
\int_{-1}^{1} C_2(x - y) u(y) \, dy 
- \frac{1}{2} \, \int_{-1}^{1} C_2(- x - y) u(y) \, dy + k_{N,C} \braket{e_0|u}_2
\, \, .
\end{align*}

\subsubsection{Connections to the Standard Fourier Expansion}
\label{standardFourierNeumann}

In the following, we connect the expansion with respect 
to the Hilbert basis $(e_k)_{k \in {\mathbb{N}}^{*}}$ to 
that expansion with respect to the Hilbert basis 
$(e^{\p}_k)_{k \in {\mathbb{Z}}}$ of $L^2_{\mathbb{C}}((-2,2))$, where 
\begin{equation*}
e_{k}^{\p}(x) := \frac{1}{2} \, e^{i \pi k x / 2} \, \, , 
\end{equation*}
for every $x \in (-2,2)$. We note for $k \in {\mathbb{N}}^{*}$, $x \in I$ that 
\begin{align*}
e_{k}(x) & = e^{\frac{i k \pi}{2}} \cdot e_{k}^{\p}(x) +
e^{-\frac{i k \pi}{2}} \cdot e_{-k}^{\p}(x) \, \, .
\end{align*}
This implies for $u \in L^2_{\mathbb{C}}(I)$ that 
\begin{align*}
\braket{e_k|u}_2 & = e^{\frac{-i k \pi}{2}} \braket{e_{k}^{\p}|{\bar{u}}\,}_{{\bar{I}},2} + e^{\frac{i k \pi}{2}} \braket{e_{-k}^{\p}|{\bar{u}}\,}_{{\bar{I}},2}
\end{align*}
where ${\bar{I}} := (-2,2)$, $\braket{\,|\,}_{{\bar{I}},2}$ denotes the scalar product 
on $L^2_{\mathbb{C}}({\bar{I}})$ and 
${\bar{u}} \in L^2_{\mathbb{C}}((-2,2))$ is defined 
by 
\begin{equation*}
{\bar{u}}(x) := 
\begin{cases}
\, \, \, \, 0 & \text{if $x \in (-2,-1)$}\\
u(x) & \text{if $x \in (-1,1)$} \\
\, \, \, \, 0 & \text{if $x \in (1,2)$}
\end{cases}
\, \, ,
\end{equation*}
for a.e. $x \in (-2,2)$.
Furthermore, for $u \in L^2_{\mathbb{C}}(I)$ and on $I$
\begin{align*}
& u = \braket{e_0|u}_2 \, e_0 + \sum_{k=1}^{\infty} \braket{e_k|u}_2 \, e_k  
= \sum_{k \in {\mathbb{Z}}} \braket{e_{k}^{\p}|u}_2 \cdot e_{k}^{\p} + \sum_{k \in {\mathbb{Z}}} e^{i k \pi} \braket{e_{-k}^{\p}|u}_2 \cdot e_{k}^{\p} \, \, .
\end{align*}
We note for $k \in {\mathbb{Z}}$ that 
\begin{align*}
& e^{i k \pi} \braket{e_{-k}^{\p}|u}_2 =
\int_{-2}^{-1} \frac{1}{2} \, e^{-i \pi k y / 2} u(-y - 2) \, dy
+ \int_{1}^{2} \frac{1}{2} \, e^{-i \pi k y / 2} u(- y + 2) \, dy 
\, \, .
\end{align*}
Hence on $I$
\begin{equation*}
u = \sum_{k=0}^{\infty} \braket{e_k|u}_2 \, e_k  = 
\sum_{k \in {\mathbb{Z}}} 
\left[\int_{-2}^2 \left(e_{k}^{\p}(y)\right)^{*} \cdot u_e(y)\, dy \right] \cdot e_{k}^{\p} \, \, , 
\end{equation*}
where 
\begin{equation*}
u_e(x) := 
\begin{cases}
u(-x - 2) & \text{if $x \in (-2,-1)$}\\
\quad \, \, u(x) & \text{if $x \in (-1,1)$} \\
u(-x + 2) & \text{if $x \in (1,2)$}
\end{cases}
\, \, . 
\end{equation*}
for a.e. $x \in (-2,2)$.

\subsection{Dirichlet Boundary Conditions}
\label{dirichletboundaryconditions}
\begin{flushleft}
We define the operator $A_{0,\D} : D(A_{0,\D}) \rightarrow L^2_{\mathbb{C}}(I)$
by 
\end{flushleft}
\begin{equation*}
D(A_{0,\D}) := \left\{ u \in C^2(\bar{I},{\mathbb{C}}): 
\lim_{x \rightarrow -1} u(x) = \lim_{x \rightarrow 1} u(x) = 0
 \right\}
\end{equation*} 
and  
\begin{equation*}
A_{0,\D} u := - \frac{4}{\pi^2} \, u^{\, \prime \prime} 
\end{equation*}  
for every $u \in D(A_{0,\D})$, where $I := (-1,1)$,
$C^2(\bar{I},{\mathbb{C}})$ 
consists of the restrictions of the elements of
$C^2(J,{\mathbb{C}})$ to $I$, 
where $J$ runs through 
all open intervals of ${\mathbb{R}}$ containing $\bar{I}$. 
Note that 
$C^2(\bar{I},{\mathbb{C}})$ is a {\it dense subspace} of $X$.
$A_{0,\D}$ is densely-defined, linear and positive symmetric.

\subsubsection{Associated Hilbert Basis and Properties}

We note that $A_{0,\D}$ is a special case of a regular Sturm-Liouville
operator. In particular, $A_{0,\D}$ is essentially self-adjoint.  The
closure $A_{\D}$ of $A_{0,\D}$, given by
\begin{equation*}
A_{\D} u = - \frac{4}{\pi^2} \, u^{\, \prime \prime}, 
\end{equation*}
where $u \in W^2_0(I,\CMPLX)$ and $\prime$ denotes the weak derivative.
$A_{\D}$ has a purely discrete spectrum
$\sigma(A_{\D})$  consisting of simple eigenvalues,
\begin{equation*}
\sigma(A_{\D}) = \left\{k^2 : k \in {\mathbb{N}}^{*}
\right\} \, \, . 
\end{equation*}
For every $k \in {\mathbb{N}}^{*}$, a normalized eigenvector corresponding 
to the eigenvalue $k^2$ is given by  
\begin{equation*}
e_{k}(x) := \sin\left(\frac{k \pi}{2}(x+1)\right)
\, \, , 
\end{equation*}
for every $x \in I$,  
\begin{equation*}
A_{0,\D} e_{k} = \frac{4}{\pi^2} \, \frac{\pi^2 k^2}{4} \, e_{k}  =
k^2 \, e_{k}  \, \, ,
\end{equation*}
Hence $e_1,e_2,\dots $ is a Hilbert basis of $L^2_{\mathbb{C}}(I)$.
Furthermore, we note for $k \in {\mathbb{N}}^{*}$,
$l \in {\mathbb{N}}$
and $x \in I$ that 
\begin{align*}
& e_{2k}(x) = (-1)^{k} \sin(k \pi x) \, \, , \\
& e_{2l + 1}(x) = 
(-1)^l \cos\left( \pi \left(l +  \frac{1}{2}\right) x \right) 
\end{align*}
and hence that 
\begin{align*}
& \text{$e_k$ is odd and periodic with period $2$, for even  $k \in {\mathbb{N}}^{*}$} \, \, , \\
& \text{$e_k$ is even and antiperiodic with period $2$, for odd $k \in {\mathbb{N}}^{*}$} \, \, .
\end{align*}
Also, $\sin\left(\frac{k \pi}{2}({\textrm{id}}_{\mathbb{R}}+1)\right)$
is periodic with period $4$ for every $k \in {\mathbb{N}}^{*}$.

\subsubsection{Compactness of $f(A_{\D})$}
For every $f \in B(\sigma(A_{\D}),{\mathbb{C}})$, if 
\begin{equation*}
(|f(k^2)|^2)_{k \in {\mathbb{N}}^{*}}
\end{equation*}
is summable, then $f(A_{\D})$ is a Hilbert-Schmidt operator and hence compact. The latter
is the case if 
\begin{equation*}
|f(\lambda)| \leqslant  c \, \lambda^{-\alpha} 
\end{equation*}
for every $\lambda \in \sigma(A_{\D})$, where $\alpha > 1/2$, $c \geqslant 0$.

\subsubsection{Properties of Simple Convolutions and Integral Representations} 
\label{simpleConvoDirichlet}

In the following, we give connections to the 
convolutions $*_{\p}$ from Section~\ref{periodicboundaryconditions},
for periodic BCs,
and
$*_{\aBC}$ from Section~\ref{antiperiodicboundaryconditions}, for
antiperiodic BCs.  For
every $k \in {\mathbb{Z}}$, $x \in I$,  the
corresponding eigenfunctions are as follows:
\begin{equation*}
e_{k}^{\p}(x) := \frac{1}{\sqrt{2}} \, e^{i \pi k x} \, \, , \, \,
e_{k}^{\aBC}(x) := \frac{1}{\sqrt{2}} \, e^{i \pi \left(k + \frac{1}{2}\right) x} \, \, .
\end{equation*}  
We note for even $C \in L^2(I)$, odd $u \in L^2_{\mathbb{C}}(I)$, $x \in I$, $k \in {\mathbb{N}}^{*}$ that 
\begin{align*}
\braket{e_k^{\p}|C}_2 & = \braket{e_{-k}^{\p}|C}_2 = 
\frac{1}{\sqrt{2}} \, \int_{-1}^{1} \cos(\pi k y) C(y) \, dy 
\, \, \, \, \left( \leqslant 
\frac{1}{\sqrt{2}} \int_{-1}^{1} C(y) \, dy
\right) \, \, , \\
\braket{e_k^{\p}|u}_2  & = - \braket{e_{-k}^{\p}|u}_2 = 
\frac{(-1)^{k+1} i}{\sqrt{2}} \braket{e_{2k}|u}_2 \, \, .
\end{align*}
As a consequence, for $k \in {\mathbb{N}}^{*}$
\begin{align*}
& \braket{e_k^{\p}|C}_2 \braket{e_k^{\p}|u}_2 e_k^{\p}(x) + 
\braket{e_{-k}^{\p}|C}_2 \braket{e_{-k}^{\p}|u}_2 e_{-k}^{\p}(x) \\
& = i \sqrt{2} \, (-1)^{k} \braket{e_k^{\p}|C}_2 \braket{e_k^{\p}|u}_2 e_{2k}(x) = \braket{e_k^{\p}|C}_2 \braket{e_{2k}|u}_2 e_{2k}(x) \, \, .
\end{align*}
Hence 
\begin{align*}
C *_{\p} u  & = \sum_{k \in {\mathbb{Z}}} \braket{e_k^{\p}|C}_2 \braket{e_k^{\p}|u}_2 e_{k}^{\p}
=  
\sum_{k = 1}^{\infty} \braket{e_k^{\p}|C}_2 \braket{e_{2k}|u}_2 e_{2k} \\
& = \sum_{k = 1}^{\infty} \varphi_1(k^2) \braket{e_{k}|u}_2 e_{k} \, \, , 
\end{align*}
where $\varphi_1 \in B(\sigma(A_{\D}),{\mathbb{C}})$ is defined by 
\begin{equation*}
\varphi_1(k^2) := 
\begin{cases}
\quad \, \, 0 & \text{if $k \in {\mathbb{N}}^{*}$ is odd} \\
\braket{e_{k/2}^{\p}|C}_2 & \text{if $k \in {\mathbb{N}}^{*}$ is even} 
\end{cases} \, \, ,
\end{equation*}
and 
\begin{equation*}
C *_{\p} P_{\text{odd}} = \varphi_1(A_{\D}) \, \, .
\end{equation*}

Also, we note for even $C \in L^2(I)$ and 
even $u \in L^2_{\mathbb{C}}(I)$, $x \in I$, $k \in {\mathbb{N}}$ that 
\begin{align*}
\braket{e_k^{\aBC}|C}_2  & = 
\braket{e_{-k-1}^{\aBC}|C}_2 
= \frac{(-1)^{k}}{\sqrt{2}} \braket{e_{2k+1}|C}_2 \, \, \, \, \left( \leqslant 
\frac{1}{\sqrt{2}} \int_{-1}^{1} C(y) \, dy
\right) 
\, \, , \\
\braket{e_k^{\aBC}|u}_2 & = \braket{e_{-k-1}^{\aBC}|u}_2 = \frac{(-1)^{k}}{\sqrt{2}} \braket{e_{2k+1}|u}_2  \, \, .
\end{align*}
As a consequence, 
\begin{align*}
& \braket{e_k^{\aBC}|C}_2 \braket{e_k^{\aBC}|u}_2 e_k^{\aBC}(x) + 
\braket{e_{-k-1}^{\aBC}|C}_2 \braket{e_{-k-1}^{\aBC}|u}_2 e_{-k-1}^{\aBC}(x) \\
& = \sqrt{2} \, (-1)^{k} \braket{e_k^{\aBC}|C}_2 \braket{e_k^{\aBC}|u}_2 e_{2k+1}(x) = \braket{e_k^{\aBC}|C}_2 \braket{e_{2k+1}|u}_2 e_{2k+1}(x)  \, \, .\\
\end{align*}
Hence 
\begin{align*}
C *_{\aBC} u  & = \sum_{k \in {\mathbb{Z}}} \braket{e_k^{\aBC}|C}_2 \braket{e_k^{\aBC}|u}_2 e_{k}^{\aBC} = \sum_{k=0}^{\infty} \braket{e_k^{\aBC}|C}_2 \braket{e_{2k+1}|u}_2 e_{2k+1}  \\
& = 
\sum_{k=1}^{\infty} \varphi_2(k^2) \braket{e_{k}|u}_2 e_{k} \, \, ,
\end{align*}
where $\varphi_2 \in B(\sigma(A_{\D}),{\mathbb{C}})$ is defined by 
\begin{equation*}
\varphi_2(k^2) := 
\begin{cases}
\braket{e_{(k-1)/2}^{\aBC}|C}_2 & \text{if $k \in {\mathbb{N}}^{*}$ is odd} \\
\qquad \, \, \, \, 0 & \text{if $k \in {\mathbb{N}}^{*}$ is even} 
\end{cases} \, \, ,
\end{equation*}
and 
\begin{equation*}
C *_{\aBC} P_{\text{even}} = \varphi_2(A_{\D}) \, \, .
\end{equation*}

\subsubsection{Properties of Canonical Convolutions and Integral Representations}

In the following, $*_{\D}$ denotes the convolution in $L^2_{\mathbb{C}}(I)$
that, according to Theorem~\ref{convolutionsinhilbertspaces}, is associated
to the Hilbert basis $(e_{k})_{k \in {\mathbb{N}}^{*}}$. In particular, for $C \in L^2(I)$ and $c \in {\mathbb{R}}$,
\begin{equation*}
\braket{e_k|C}_2 = \int_{-1}^{1} 
\sin\left(\frac{k \pi}{2}(y+1)\right) C(y)  \, dy \, \, , 
\end{equation*}
is real-valued for every $k \in {\mathbb{N}}^{*}$ and $c - C *_{\D} \cdot$ is a bounded self-adjoint function of $A_{\D}$. Furthermore, since 
\begin{equation*}
\int_{-1}^{1} C \, dy - 
\braket{e_k|C}_2 = \int_{-1}^{1} \left[1 -
\sin\left(\frac{k \pi}{2}(y+1)\right) \right] C(y)  \, dy \, \, , 
\end{equation*}
for every $k \in {\mathbb{N}}^{*}$,  
if $C \geqslant 0$ and 
\begin{equation*}
c = \int_{-1}^{1} C(y) \, dy \, \, , 
\end{equation*}
then the operator 
$c - C *_{\D} \cdot$
is in particular positive.
\newline
\linebreak
If $C$ is even, 
\begin{equation*}
\braket{e_k|C}_2 = 0 \, \, , 
\end{equation*}
for every even $k \in {\mathbb{N}}^{*}$. 
\newline
\linebreak
In addition, since for every $C, g \in L^2_{\mathbb{C}}(I)$, $x \in I$ and every finite subset
$S \subset {\mathbb{N}}^{*}$
\begin{align*}
& \sum_{k \in S} |(e_k(x))^{*} \cdot \braket{C|e_k}_2|^2 \leqslant \sum_{k \in S}
|\braket{e_k|C}_2|^2  \leqslant 
\sum_{k \in {\mathbb{N}}^{*}} |\braket{e_k|C}_2|^2
\, \, , 
\end{align*}
we note that 
\begin{equation*}
(C *_{\D} u)(x) = \sum_{k \in {\mathbb{N}}^{*}} \braket{e_k|C}_2 \braket{e_k|u}_2 . \, e_k(x) =
\big \langle \sum_{k \in {\mathbb{N}}^{*}} (e_k(x))^{*} \cdot \braket{C|e_k}_2 . e_k|u \rangle_2
\end{equation*}
and, 
since for $k \in {\mathbb{N}}^{*}$, $x, u \in {\mathbb{R}}$
\begin{align*}
& \sin\left(\frac{k \pi}{2}(x+1)\right) \sin\left(\frac{k \pi}{2}(y+1)\right) \\
& = \frac{1}{2} \left\{\left[\cos\left(\frac{k \pi}{2}(x - y + 1)\right) 
- \cos\left(\frac{k \pi}{2}(x + y + 1)\right) \right]
\cos\left(\frac{k \pi}{2}\right) \right. \\
& \left.\qquad \, \, \, \, + \left[\sin\left(\frac{k \pi}{2}(x - y + 1)\right) + \sin\left(\frac{k \pi}{2}(x + y + 1)\right) \right]
\sin\left(\frac{k \pi}{2}\right)
\right\} \, \, ,
\end{align*}
for {\it even} $C$, $k \in {\mathbb{N}}^{*}$, $x \in I$ that 
\begin{align*}
& (e_k(x))^{*} \cdot \braket{C|e_k}_2 = 
\sin\left(\frac{k \pi}{2}(x+1)\right) \int_{-1}^{1} C^{*}(y) \sin\left(\frac{k \pi}{2}(y+1)\right) dy \\
& = \sin\left(\frac{k \pi}{2}\right) \int_{-1}^{1} {\hat{C}_{\aBC}}^{*}(y) \,
\sin\left(\frac{k \pi}{2}(x - y + 1)\right) 
dy \\
& = \sin\left(\frac{k \pi}{2}\right) \int_{x-1}^{x+1} {\hat{C}_{\aBC}}^{*}(y) \,
\sin\left(\frac{k \pi}{2}(x - y + 1)\right) 
dy \\
& =  \sin\left(\frac{k \pi}{2} \right) \int_{-1}^{1} {\hat{C}_{\aBC}}^{*}(x - y) \,
\sin\left(\frac{k \pi}{2}(y + 1)\right) 
dy \\
& = \sin\left(\frac{k \pi}{2} \right)  \braket{e_k|{\hat{C}_{\aBC}}^{*}(x - {\textrm{id}}_{\mathbb{R}})}_2
\, \, , 
\end{align*}
where ${\hat{C}_{\aBC}}$ denotes the extension of $C$ to a $2$-{\it antiperiodic} function 
on ${\mathbb{R}}$, i.e., such that 
\begin{equation*}
{\hat{C}_{\aBC}}(x + 2) = - {\hat{C}_{\aBC}}(x)
\end{equation*}
for every $x \in I$. Here, it has been used that 
$\sin\left(\frac{k \pi}{2}({\textrm{id}}_{\mathbb{R}} + 1)\right)$
is $2$-{\it antiperiodic} for {\it odd} $k \in {\mathbb{N}}^{*}$.
\newline
\linebreak
We note for $k \in {\mathbb{N}}^{*}$ that 
\begin{align*}
\sin\left(\frac{k \pi}{2}\right) =
\begin{cases}
\, \, \, \, 0 & \text{if $k$ is even} \\
\, \, \, \, 1 & \text{if $k$ is odd and $(k-1)/2$ is even} \\
-1 & \text{if $k$ is odd and $(k-1)/2$ is odd} 
\end{cases} \, \, . 
\end{align*}
As a consequence, if we denote by $\mbox{P}$ the orthogonal projection 
onto the closure of the subspace 
\begin{equation*}
\mathrm{Span}(\{e_{4l + 1}: l \in {\mathbb{N}}\}) \, \, ,
\end{equation*} 
then 
for {\it even} $C$ and {\it odd} $k \in {\mathbb{N}}^{*}$, $x \in I$ 
\begin{align} \label{oddk}
(e_k(x))^{*} \cdot \braket{C|e_k}_2  
& = (e_k(x))^{*} \braket{\mbox{PC}|e_k}_2
+ (e_k(x))^{*} \braket{C - {\mbox{PC}}|e_k}_2 \nonumber \\
& = \braket{e_k|({\widehat{\mbox{PC}}_{\aBC}})^{*}(x - \cdot
) - ({\hat{C}_{\aBC} - \widehat{\mbox{PC}}_{\aBC}})^{*}(x - \cdot)}_2
\, \, ,
\end{align}
where $\widehat{\mbox{PC}}_{\aBC}$ denotes the extension of $\mbox{PC}$ to a $2$-{\it antiperiodic} function 
on ${\mathbb{R}}$. Here, we used for every odd $k \in {\mathbb{N}}^{*}$ 
that $e_{k}$ is even and hence that $\mbox{PC}$ and 
$C - \mbox{PC}$ are even. 
\newline
\linebreak 
In the following, we extend (\ref{oddk}) to {\it even} $k \in {\mathbb{N}}^{*}$. For this purpose, 
we note for every even 
$u \in L^2_{\mathbb{C}}(I)$ that its 2-anti-periodic extension ${\hat u}$ is even, too. For the proof, let $l \in {\mathbb{N}}$, 
$x \in [-1-2l,-2l]$. Then 
\begin{align*}
{\hat u}(x) & = (-1)^l u(x+2l) = (-1)^l u(-x-2l) = (-1)^l \cdot (-1)^l
{\hat u}(-x-2l + 2l) = {\hat u}(-x) \, \, .
\end{align*}
Hence it follows for even
$u \in L^2_{\mathbb{C}}(I)$, $x \in I$ and even $k \in {\mathbb{N}}^{*}$ that 
\begin{align*}
\braket{e_k|{\hat u}(x - \cdot)}_2 & = 
- \braket{e_k|{\hat u}(- x - \cdot)}_2 \, \, ,
\end{align*}
which implies that 
\begin{equation*}
\braket{e_k|{\hat u}(x - \cdot) + {\hat u}(- x - \cdot)}_2 = 0
\, \, .
\end{equation*}
On the other hand, for odd $k \in {\mathbb{N}}^{*}$ 
\begin{align*}
\braket{e_k|{\hat u}(x - \cdot)}_2 & = 
\braket{e_k|{\hat u}(- x - \cdot)}_2 \, \, .
\end{align*}
As a consequence, 
\begin{equation*}
\braket{e_k|\frac{1}{2}\,[{\hat u}(x - \cdot)+ {\hat u}(- x - \cdot)]}_2
= 
\begin{cases}
\qquad \, \, \, \, 0 & \text{if $k \in {\mathbb{N}}^{*}$ is even} \\
\braket{e_k|{\hat u}(x - \cdot)}_2 & \text{if $k \in {\mathbb{N}}^{*}$ is odd}
\end{cases} \, \, .
\end{equation*}
Therefore, we conclude from (\ref{oddk}) that 
for {\it even} $C$ and $k \in {\mathbb{N}}^{*}$, $x \in I$ 
\begin{align*} 
& (e_k(x))^{*} \cdot \braket{C|e_k}_2 \\
& = \braket{e_k|\frac{1}{2} \, [({\widehat{\mbox{PC}}_{\aBC}})^{*}(x - \cdot
) + ({\widehat{\mbox{PC}}_{\aBC}})^{*}(- x - \cdot
) - ({\hat{C}_{\aBC} - \widehat{\mbox{PC}}_{\aBC}})^{*}(x - \cdot) -
({\hat{C}_{\aBC} - \widehat{\mbox{PC}}_{\aBC}})^{*}(- x - \cdot)
]}_2
\end{align*}
and hence that 
\begin{align*}
& (C *_{\D} u)(x) = \sum_{k \in {\mathbb{N}}^{*}} \braket{e_k|C}_2 \braket{e_k|u}_2 . \, e_k(x) \\
& =
\big \langle \sum_{k \in {\mathbb{N}}^{*}} (e_k(x))^{*} \cdot \braket{C|e_k}_2 . e_k|u \big \rangle_2 \\
& = 
\big \langle \sum_{k \in {\mathbb{N}}^{*}} \langle e_k|
\frac{1}{2} \, [({\widehat{\mbox{PC}}_{\aBC}})^{*}(x - \cdot
) + ({\widehat{\mbox{PC}}_{\aBC}})^{*}(- x - \cdot
) \\
& \qquad \qquad \, \, \, \, \, \, \, \, \, \, - ({\hat{C}_{\aBC} - \widehat{\mbox{PC}}_{\aBC}})^{*}(x - \cdot) -
({\hat{C}_{\aBC} - \widehat{\mbox{PC}}_{\aBC}})^{*}(- x - \cdot)
]
\rangle_2 . e_k|g \big \rangle_2 \\
& = \frac{1}{2} \, \braket{({\widehat{\mbox{PC}}_{\aBC}})^{*}(x - \cdot
)|u}_2 + \frac{1}{2} \, \braket{({\widehat{\mbox{PC}}_{\aBC}})^{*}(- x - \cdot
)|u}_2 \\
& \quad \, \, - \frac{1}{2} \,
\braket{(\hat{C}_{\aBC} - \widehat{\mbox{PC}}_{\aBC})^{*}(x - \cdot
)|u}_2 - \frac{1}{2} \, \braket{(\hat{C}_{\aBC} - \widehat{\mbox{PC}}_{\aBC})^{*}(-x - \cdot
)|u}_2 \\
& = \frac{1}{2} \, 
\int_{-1}^{1} ({\widehat{\mbox{PC}}_{\aBC}})(x - y) u(y) \, dy 
+ \frac{1}{2} \, 
\int_{-1}^{1} ({\widehat{\mbox{PC}}_{\aBC}})(- x - y) u(y) \, dy \\
& \quad \, \, - \frac{1}{2} \,
\int_{-1}^{1} (\hat{C}_{\aBC} - {\widehat{\mbox{PC}}_{\aBC}})(x - y) u(y) \, dy 
- \frac{1}{2} \, \int_{-1}^{1} (\hat{C}_{\aBC} - {\widehat{\mbox{PC}}_{\aBC}})(- x - y) u(y) \, dy
\, \, .
\end{align*}

\subsubsection{Representation of the Projection Present in the Canonical Convolution}
In the following, we give a representation of $\mbox{P}$, which
is independent of the orthonormal system used in its definition. 
\newline
\linebreak 
For $u \in L^2_{\mathbb{C}}(I)$, $x \in I$, we note that   
\begin{align*}
\sum_{k=0}^n \braket{e_{4k+1}|u} e_{4k+1}(x) =
\big \langle 
\sum_{k=0}^n 
e_{4k+1}(x) . e_{4k+1} | u
\big \rangle .
\end{align*}
Furthermore, for $y \in I$,
\begin{align*}
& e_{4k+1}(x) \, e_{4k+1}(y) = \sin\left(\frac{(4k+1) \pi}{2}(x+1)\right) \sin\left(\frac{(4k+1) \pi}{2}(y+1)\right) \\
& = \frac{1}{2} \, \left[\sin\left(\frac{(4k+1) \pi}{2}(x - y + 1)\right) + \sin\left(\frac{(4k+1) \pi}{2}(x + y + 1)\right) \right] \, \, .
\end{align*}
Since for $a \in {\mathbb{R}}$,
\begin{align*}
\sum_{k=0}^n \sin((4 k + 1) a) & =
\sum_{k=0}^n [\sin(4 k a) \cos(a) + \cos(4 k a) \sin(a)] \\
& = \cos(a) \sum_{k=0}^n \sin(4 k a) + \sin(a) 
\sum_{k=0}^n \cos(4 k a) \, \, ,
\end{align*}
for $n \in {\mathbb{N}}^{*}$, $b \in {\mathbb{C}}$ 
satisfying $b \neq 2 \pi l, l \in {\mathbb{Z}}$,
\begin{align*}
& \sum_{k=0}^n \sin(k b) = \frac{1}{2i} 
\left[ 
\sum_{k=0}^n e^{i k b} - \sum_{k=0}^n e^{- i k b} \right] =
\frac{1}{2i} \left[ 
\sum_{k=0}^n (e^{i b})^{k} - \sum_{k=0}^n (e^{- i b})^{k} \right] \\
& = \frac{1}{2 \sin{\!(b/2)}} \left\{ \sin\left(\frac{b}{2}\right)
+ \sin\left[\left(n + \frac{1}{2}\right) b \right] 
\right\}
\end{align*}
and hence if $a \neq l \pi/2, l \in {\mathbb{Z}}$,
\begin{align*}
\sum_{k=0}^n \sin((4 k + 1) a) & 
= \frac{\sin[(2n+1)a] \sin[2(n+1)a]}{\sin{\!(2a)}} \, \, ,
\end{align*}
we conclude that 
\begin{align*}
& \sum_{k=0}^n e_{4k+1}(x) \, e_{4k+1}(y) \\
& = \frac{1}{2} \, \frac{\sin\left(\frac{(2n+1) \pi}{2}(x - y + 1)\right)
\sin\left(\frac{2(n+1) \pi}{2}(x - y + 1)\right)}{\sin{\![\pi(x - y + 1)]}} 
\\
& \quad \, \, + \frac{1}{2} \, \frac{\sin\left(\frac{(2n+1) \pi}{2}(x + y + 1)\right)
\sin\left(\frac{2(n+1) \pi}{2}(x + y + 1)\right)}{\sin{\![\pi(x + y + 1)]}} \, \, ,
\end{align*}
if both
\begin{equation*}
x - y \, \, , \, \, x + y  \notin {\mathbb{Z}} \, \, .
\end{equation*}
As a consequence, 
\begin{align*}
& \sum_{k=0}^n \braket{e_{4k+1}|u} e_{4k+1}(x) \\
& \quad \, \, + \frac{1}{2} \, \int_{-1}^{1}
\frac{\sin\left(\frac{(2n+1) \pi}{2}(x + y + 1)\right)
\sin\left(\frac{2(n+1) \pi}{2}(x + y + 1)\right)}{\sin{\![\pi(x + y + 1)]}} \, u(y) \, dy \\
& = \frac{1}{2} \, \int_{-1}^{1}
\frac{\sin\left(\frac{(2n+1) \pi}{2}(x - y + 1)\right)
\sin\left(\frac{2(n+1) \pi}{2}(x - y + 1)\right)}{\sin{\![\pi(x - y + 1)]}} \, u(y) \, dy \\
& \quad \, \, + \frac{1}{2} \, \int_{-1}^{1}
\frac{\sin\left(\frac{(2n+1) \pi}{2}(x - y + 1)\right)
\sin\left(\frac{2(n+1) \pi}{2}(x - y + 1)\right)}{\sin{\![\pi(x - y + 1)]}} \, u(-y) \, dy \, \, , 
\end{align*}
and hence 
\begin{align*}
& \mbox{Pu} = 
\lim_{n \rightarrow \infty} 
\frac{\sin\left[\pi (n + \frac{1}{2}) 
({\textrm{id}}_{I} + 1)\right]
\sin\left[\pi (n+1) ({\textrm{id}}_{I}  + 1)\right]}{\sin{\![\pi({\textrm{id}}_{I}  + 1)]}} * \frac{1}{2} \, [u + u \circ (-{\textrm{id}}_{I})] \, \, , 
\end{align*}
where $*$ denotes the integral convolution on $I$, and the limit is to be 
performed in $L^2_{\mathbb{C}}(I)$.

\subsubsection{Connections to the Standard Fourier Expansion}
\label{standardFourierDirichlet}
In the following, we connect the expansion with respect 
to the Hilbert basis $(e_k)_{k \in {\mathbb{N}}^{*}}$ to 
that expansion with respect to the Hilbert basis 
$(e^{\p}_k)_{k \in {\mathbb{Z}}}$ of $L^2_{\mathbb{C}}((-2,2))$, where 
\begin{equation*}
e_{k}^{\p}(x) := \frac{1}{2} \, e^{i \pi k x / 2} \, \, , 
\end{equation*}
for every $x \in (-2,2)$. We note for $k \in {\mathbb{N}}^{*}$, $x \in I$ that 
\begin{align*}
e_{k}(x) & = \frac{e^{\frac{i k \pi}{2}}}{i} \cdot e_{k}^{\p}(x) -
\frac{e^{-\frac{i k \pi}{2}}}{i} \cdot e_{-k}^{\p}(x) \, \, .
\end{align*}
This implies for $u \in L^2_{\mathbb{C}}(I)$ that 
\begin{align*}
\braket{e_k|u}_2 & 
= - \frac{e^{\frac{-i k \pi}{2}}}{i} \braket{e_{k}^{\p}|{\bar{u}}\,}_{{\bar{I}},2} + \frac{e^{\frac{i k \pi}{2}}}{i} \braket{e_{-k}^{\p}|{\bar{u}}\,}_{{\bar{I}},2}
\end{align*}
where ${\bar{I}} := (-2,2)$, $\braket{\,|\,}_{{\bar{I}},2}$ denotes the scalar product 
on $L^2_{\mathbb{C}}({\bar{I}})$ and 
${\bar{u}} \in L^2_{\mathbb{C}}((-2,2))$ is defined 
by 
\begin{equation*}
{\bar{u}}(x) := 
\begin{cases}
\, \, \, \, 0 & \text{if $x \in (-2,-1)$}\\
u(x) & \text{if $x \in (-1,1)$} \\
\, \, \, \, 0 & \text{if $x \in (1,2)$}
\end{cases}
\, \, ,
\end{equation*}
for a.e. $x \in (-2,2)$.
Furthermore, for $u \in L^2_{\mathbb{C}}(I)$ and on $I$
\begin{align*}
& u = \sum_{k=1}^{\infty} \braket{e_k|u}_2 \, e_k  = 
\sum_{k \in {\mathbb{Z}}} \braket{e_{k}^{\p}|u}_2 \cdot e_{k}^{\p} - \sum_{k \in {\mathbb{Z}}} e^{i k \pi} \braket{e_{-k}^{\p}|u}_2 \cdot e_{k}^{\p} \, \, .
\end{align*}
We note for $k \in {\mathbb{Z}}$ that 
\begin{align*}
& e^{i k \pi} \braket{e_{-k}^{\p}|u}_2 = 
\int_{-2}^{-1} \frac{1}{2} \, e^{-i \pi k y / 2} u(-y - 2) \, dy
+ \int_{1}^{2} \frac{1}{2} \, e^{-i \pi k y / 2} u(- y + 2) \, dy 
\, \, .
\end{align*}
Hence on $I$
\begin{equation*}
u = \sum_{k=1}^{\infty} \braket{e_k|u}_2 \, e_k  = 
\sum_{k \in {\mathbb{Z}}} 
\left[\int_{-2}^2 \left(e_{k}^{\p}(x)\right)^{*} \cdot u_e(x)\, dx \right] \cdot e_{k}^{\p} \, \, , 
\end{equation*}
where 
\begin{equation*}
u_e(x) := 
\begin{cases}
- u(-x - 2) & \text{if $x \in (-2,-1)$}\\
\qquad u(x) & \text{if $x \in (-1,1)$} \\
- u(-x + 2) & \text{if $x \in (1,2)$}
\end{cases}
\, \, . 
\end{equation*}
for a.e. $x \in (-2,2)$.

\section{Numerical Experiments}

\label{sec:numerics}

\begin{figure}[t]
\centering
\subfigure%[Micromodulus function.]
{
\scalebox{0.30}{\includegraphics{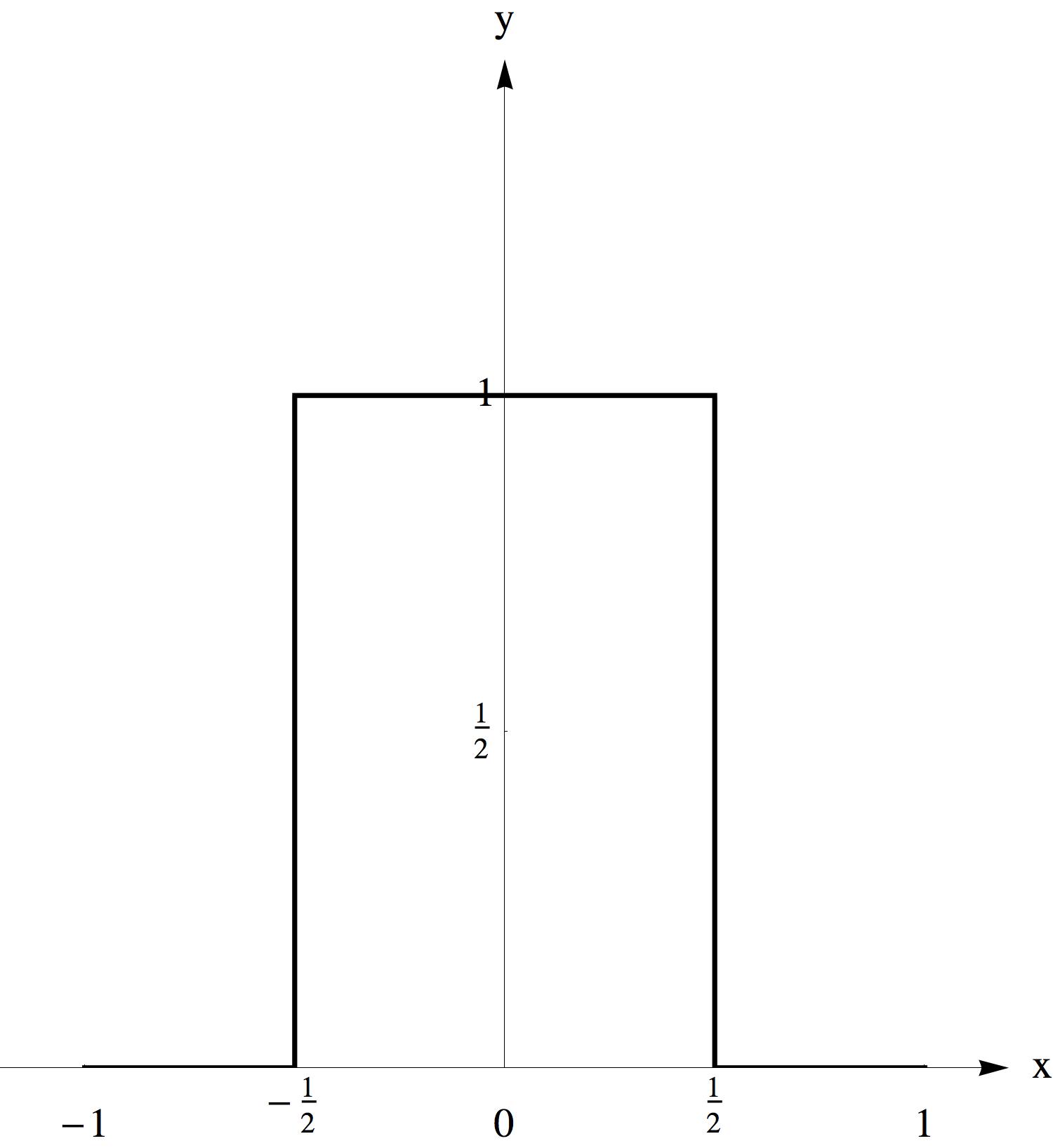}}
}
\hspace{.35cm}
\subfigure%[Discontinuous initial displacement function.]
{
\scalebox{0.30}{\includegraphics{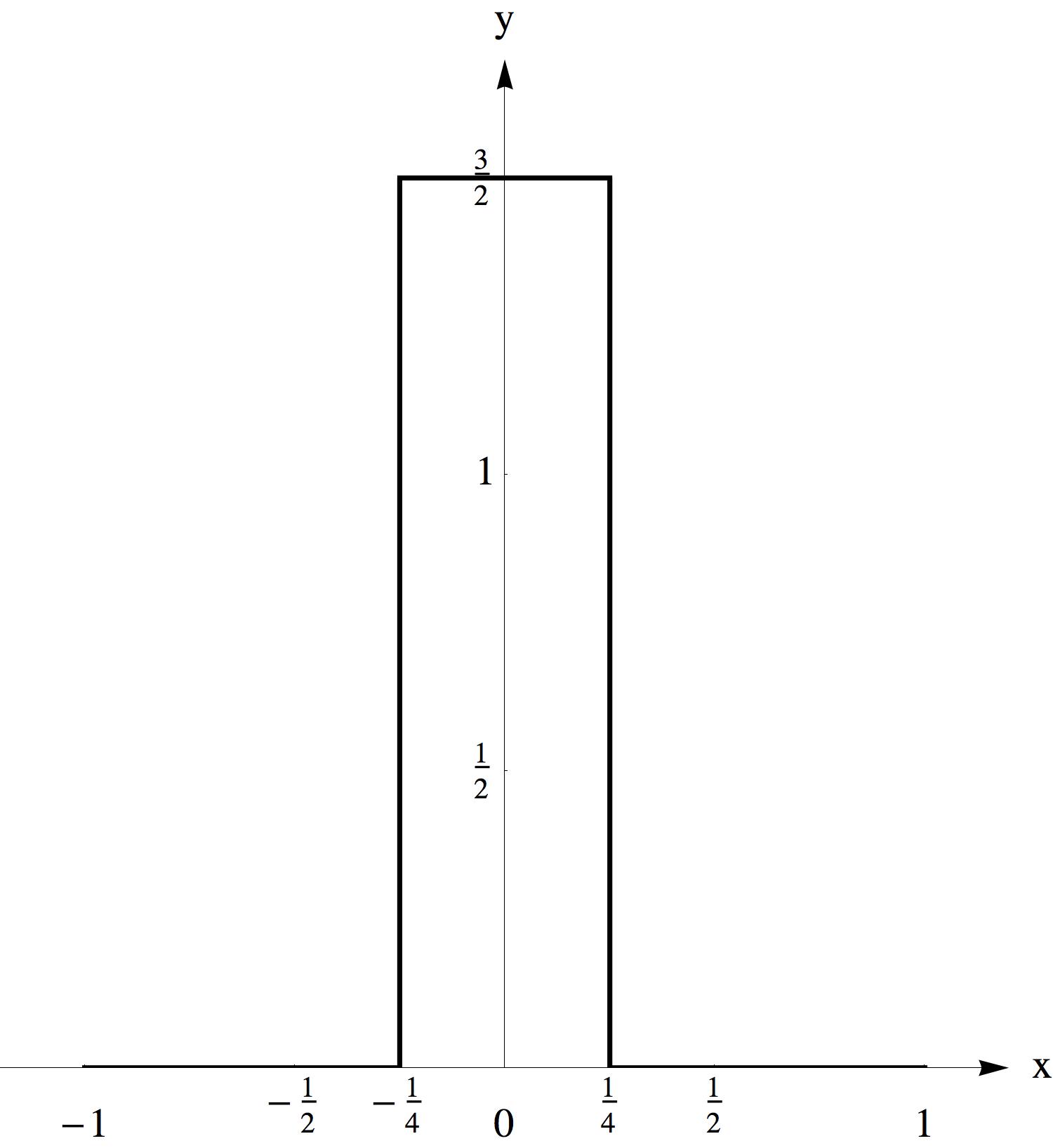}}
}
\hspace{.35cm}
\subfigure%[Initial displacement function.]
{
\scalebox{0.30}{\includegraphics{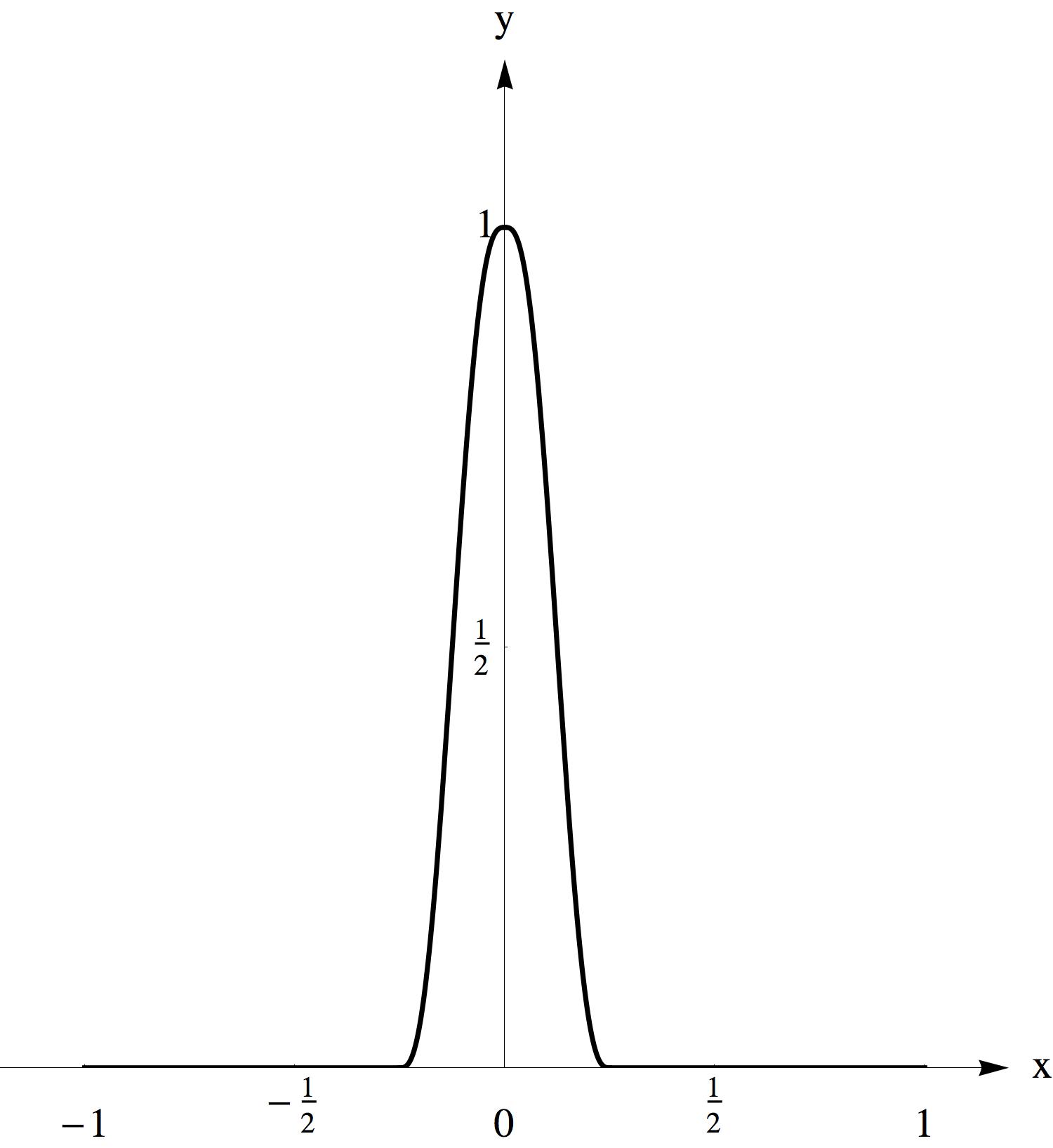}}
}
\caption{Micromodulus function $C(x)$ (left). Discontinuous (middle) and continuous (right)
initial displacement functions $u_{0,\textrm{disc}}(x)$ and $u_{0,\textrm{cont}}(x)$, respectively.}
\label{fig:micromodulusData}
\end{figure}

Recalling the governing equation \eqref{governingEqu}, we numerically
solve the following nonlocal equation
\begin{subequations}
\label{governing}
\begin{alignat}{2}
u_{tt}(x,t) + \varphi(A_{\BC})u(x,t) &= b(x,t),\quad &&(x,t) \in \Omega  \times J, \\
u(x,0) &= u_0(x), \quad &&x \in \Omega, \\
u_t(x,0) &= v_0(x) \quad &&x \in \Omega,
\end{alignat}
\end{subequations}
where $J = (0,T)$ is some finite time interval, $\Omega=(-1,1)$, $b$ is a given source term,
and $u_0$ and $v_0$ are given initial conditions.
The choice of the subscript
${\BC} \in \{\p, \aBC, \N, \D\}$ is determined by the BCs that are
to be satisfied at the boundary of the physical domain $(-1,1)$.
This, in turn, determines the function of the classical operator
$\varphi(A_{\tt{BC}})$
as described in Table \ref{table:nonlocalFunc} where we have defined
\begin{equation}
\label{sillingConst}
c := \frac{1}{\sqrt{2}} \int_{-1}^{1}C(y) \, dy.
\end{equation}
Furthermore, the abstract convolutions, for fixed $t \in J$, are given in terms of the eigenbasis as follows:\\
\begin{eqnarray*}
C *_{\p} u(x,t)  & = &
\sum_{k \in {\mathbb{Z}}} \braket{e_k^{\p}|C}_2 \braket{e_k^{\p}|u}_2 e_{k}^{\p}(x), \\
C *_{\aBC} u(x,t)  & = &
\sum_{k \in {\mathbb{Z}}} \braket{e_k^{\aBC}|C}_2 \braket{e_k^{\aBC}|u}_2 e_{k}^{\aBC}(x).
\end{eqnarray*}
For our numerical approximation of the solution of the nonlocal
problem \eqref{governing}, we need to write these infinite sums as
integral convolutions. This has been accomplished
for all types of BCs in Section \ref{sec:convoBC}.
We provide these integral convolutions below
\begin{eqnarray*}
C *_{\p} u(x,t)   &= &\frac{1}{\sqrt{2}} \int_{-1}^{1} {\hat{C}_{\p}}(x - y) \, u(y,t) \, dy, \\
C *_{\aBC} u(x,t) &= &\frac{1}{\sqrt{2}} \int_{-1}^{1} {\hat{C}_{\aBC}}(x - y) \, u(y,t) \, dy.
\end{eqnarray*}

\begin{table}[t]
\centering
\begin{tabular}{c|ccc}
\hline
&&& \\
$\BC$  &$\varphi(A_{\BC})u(x,t)$ &BCs enforced \\ [1ex]
\hline  \\ [-.5ex]
$\p$   &$(c-C *_{\p}) u(x,t)$
&$u(-1,t) =  u(1,t),\quad u_x(-1,t) =  u_x(1,t)$ \\ [2ex]
$\aBC$ &$(c-C *_{\aBC}) u(x,t)$
&$u(-1,t) = -u(1,t),\; u_x(-1,t) = -u_x(1,t)$ \\ [2ex]
$\N$ &$\sqrt{2} \{[(c-C*_{\p}P_{\textrm{even}})+(c-C*_{\aBC}P_{\textrm{odd}})]u(x,t)\}$
&$u_x(-1,t) = u_x(1,t) = 0$ \\ [2ex]
$\D$ &$\sqrt{2} \{[(c-C*_{\p}P_{\textrm{odd}})+(c-C *_{\aBC}P_{\textrm{even}})]u(x,t)\}$
&$u(-1,t) = u(1,t) = 0$ \\ [2ex]
\hline
\end{tabular}
\caption{
The choice of nonlocal operators based on the boundary conditions enforced.
}
\label{table:nonlocalFunc}
\end{table}

\subsection{Discretization in Space}
%% Discontinuous data plots
%% periodic
\begin{figure}[t]
\centering
\subfigure[Regulating function.]{
\scalebox{0.35}{\includegraphics{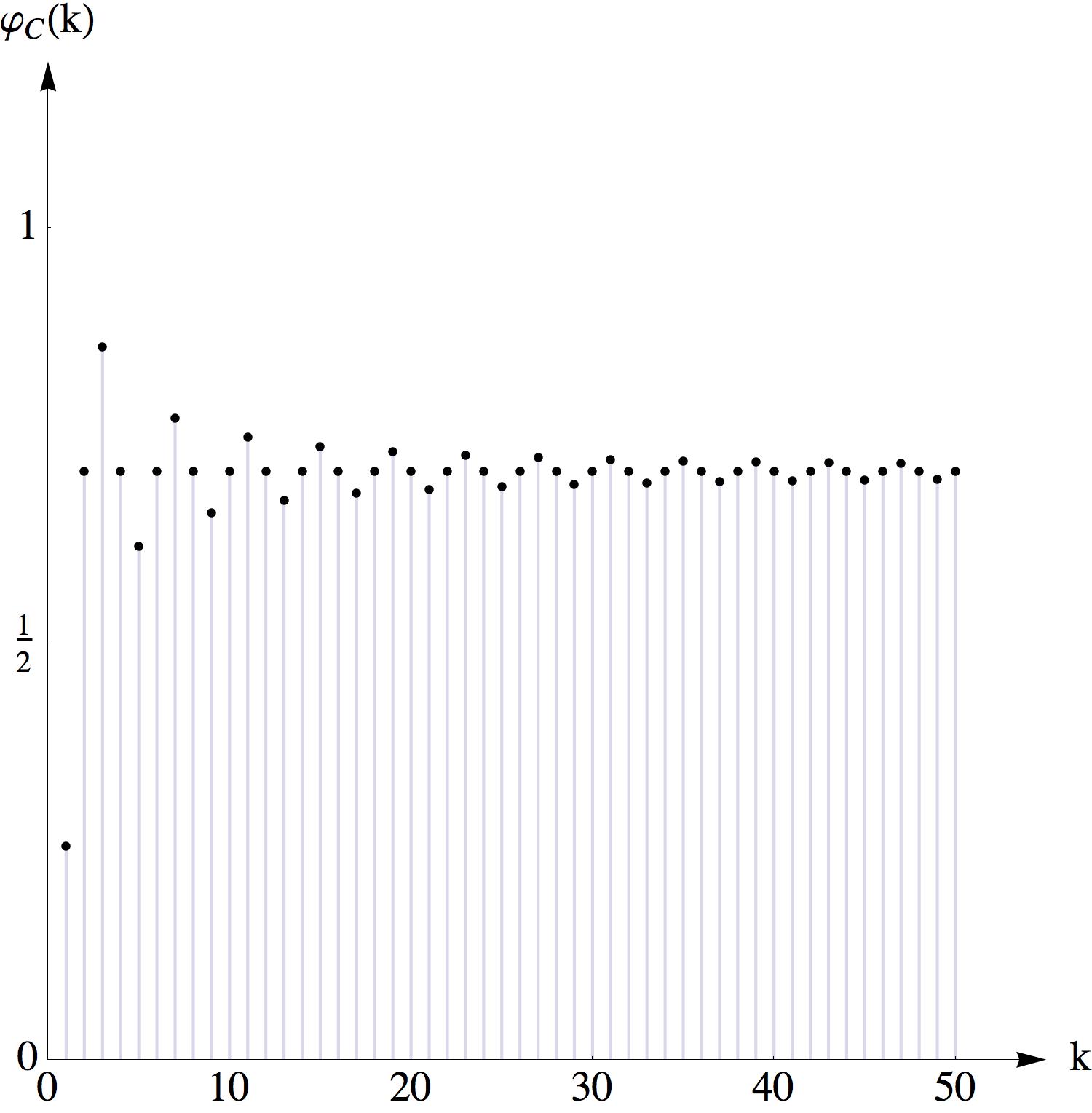}}
\label{periodicRegulatingFunc}
}
\hspace{.35cm}
\subfigure[Contour plot of $u$ from Fig.~\ref{periodicDataView}.]{
\scalebox{0.70}{\includegraphics{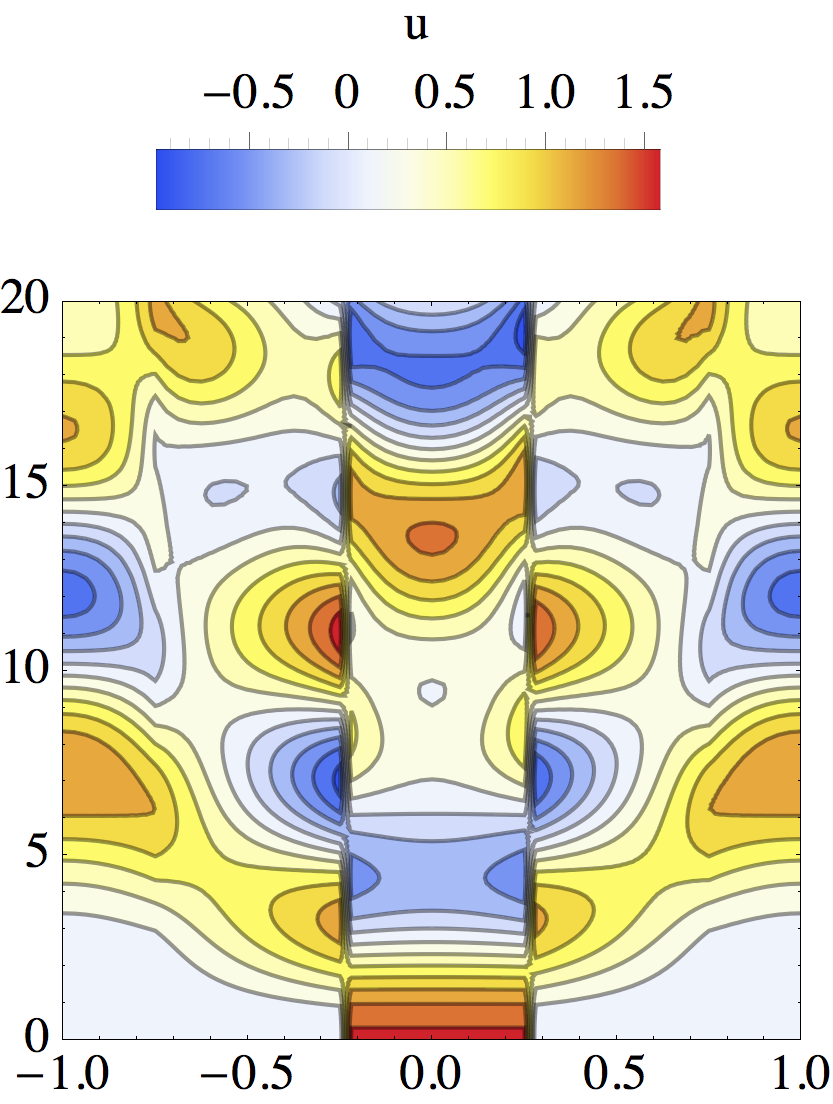}}
\label{periodicContourPlot}
}
\\
\subfigure[Solution $u$ to the nonlocal wave equation
with initial data $u(x,0) = u_{0,\textrm{disc}}(x)$
and $u_t(x,0)=0,~x \in (-1,1)$. Initial data view.]{
\scalebox{0.50}{\includegraphics{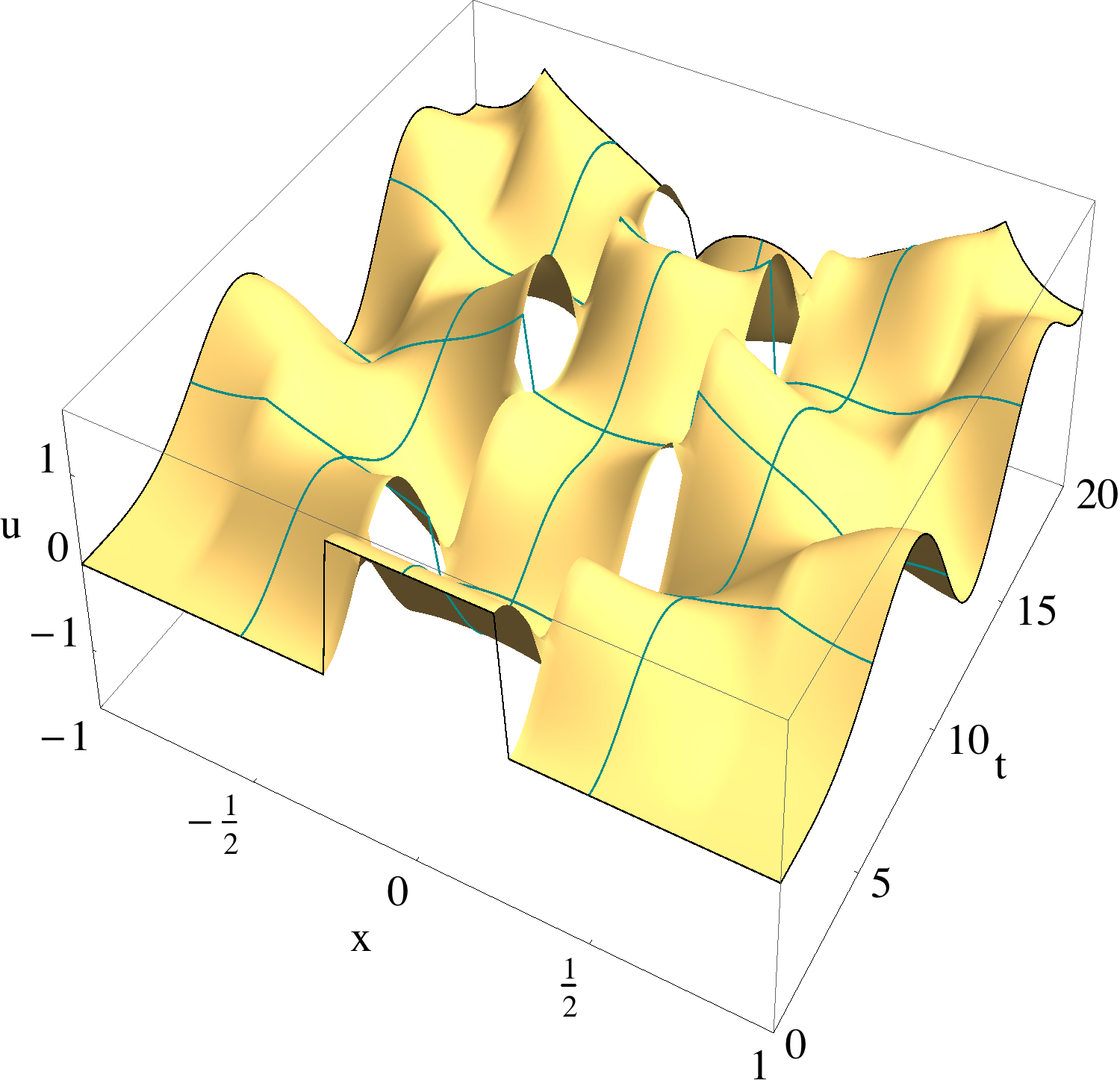}}
\label{periodicDataView}
}
\hspace{.35cm}
\subfigure[The same solution from Fig.~\ref{periodicDataView} from a
boundary point of view.]{
\scalebox{0.50}{\includegraphics{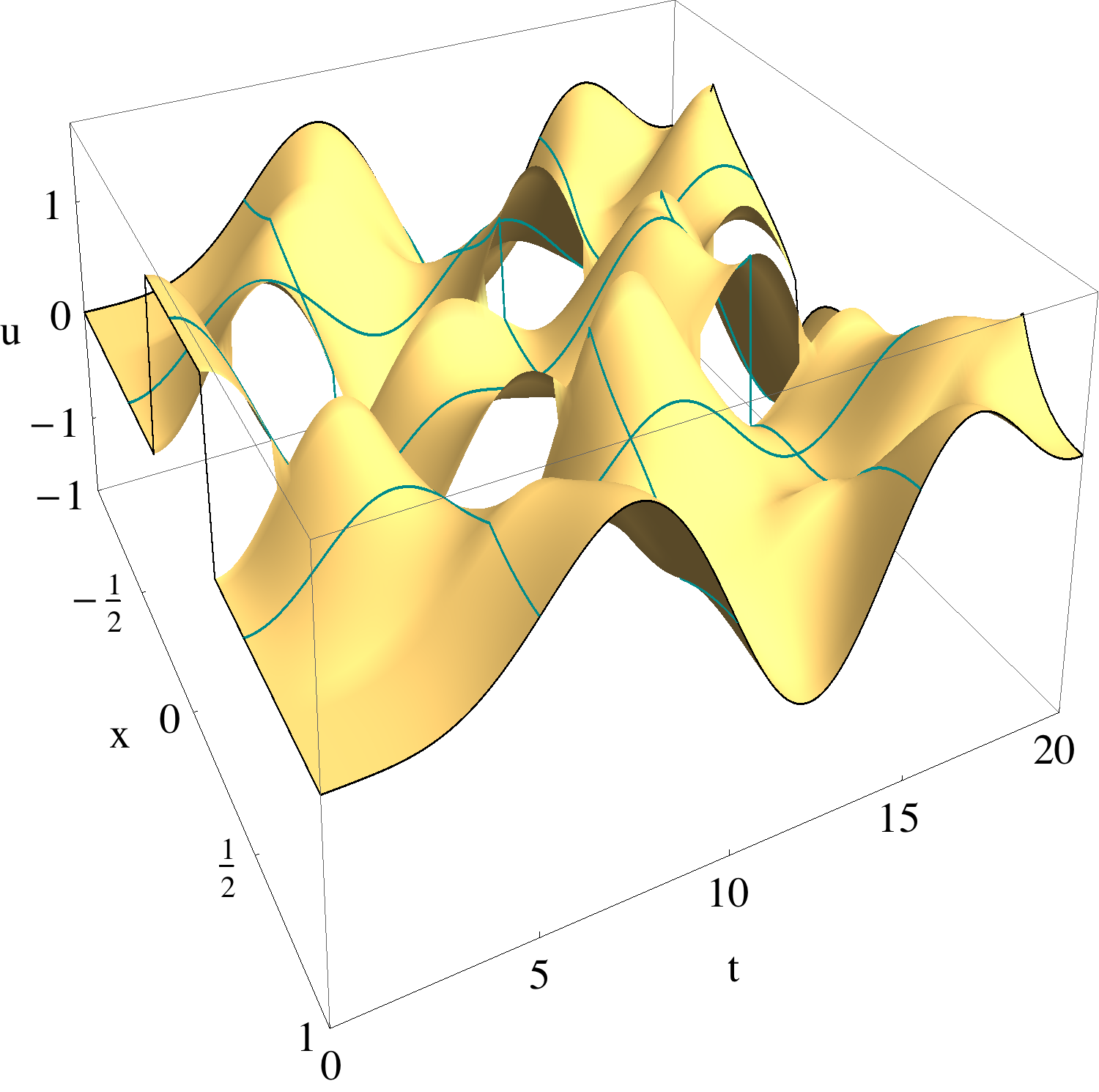}}
\label{periodicBoundaryView}
}
\caption{Solution to the nonlocal wave equation solution with
periodic boundary conditions and vanishing initial velocity.}
\label{fig:periodic}
\end{figure}

%%% antiperiodic
\begin{figure}[t]
\centering
\subfigure[Regulating function.]{
\scalebox{0.35}{\includegraphics{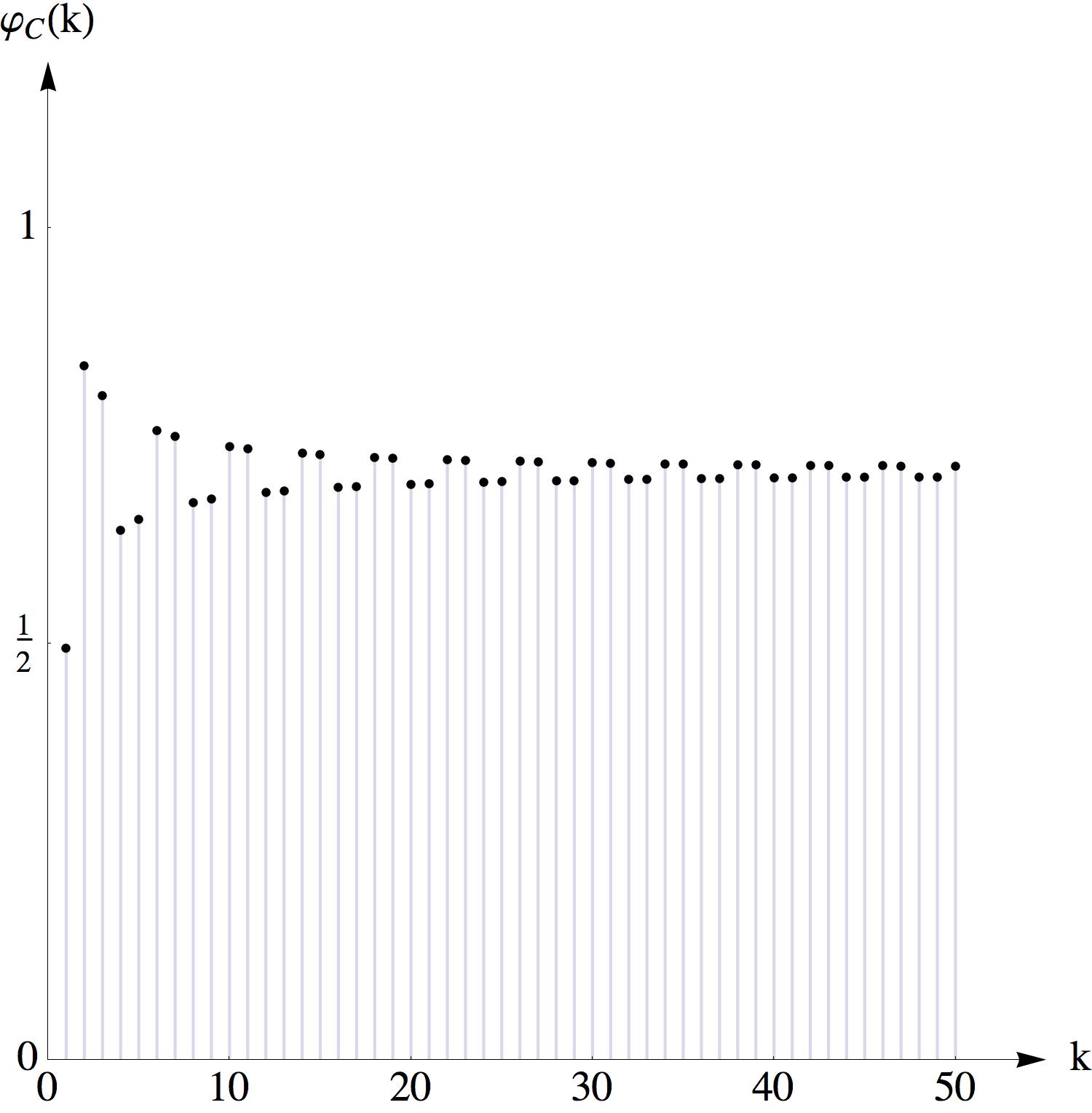}}
\label{antiperiodicRegulatingFunc}
}
\hspace{.35cm}
\subfigure[Contour plot of $u$ from Fig.~\ref{antiperiodicDataView}.]{
\scalebox{0.70}{\includegraphics{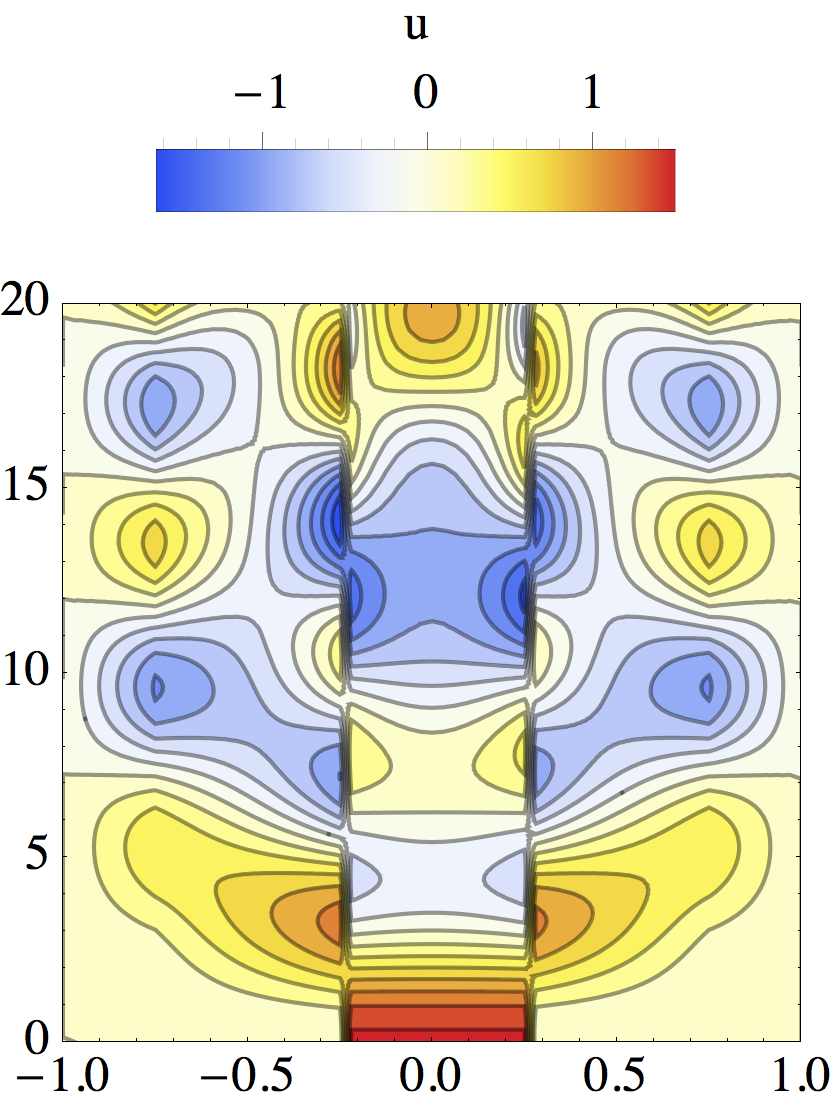}}
\label{antiperiodicContourPlot}
}
\\
\subfigure[Solution $u$ to the nonlocal wave equation
with initial data $u(x,0) = u_{0,\textrm{disc}}(x)$
and $u_t(x,0)=0,~x \in (-1,1)$. Initial data view.]{
\scalebox{0.50}{\includegraphics{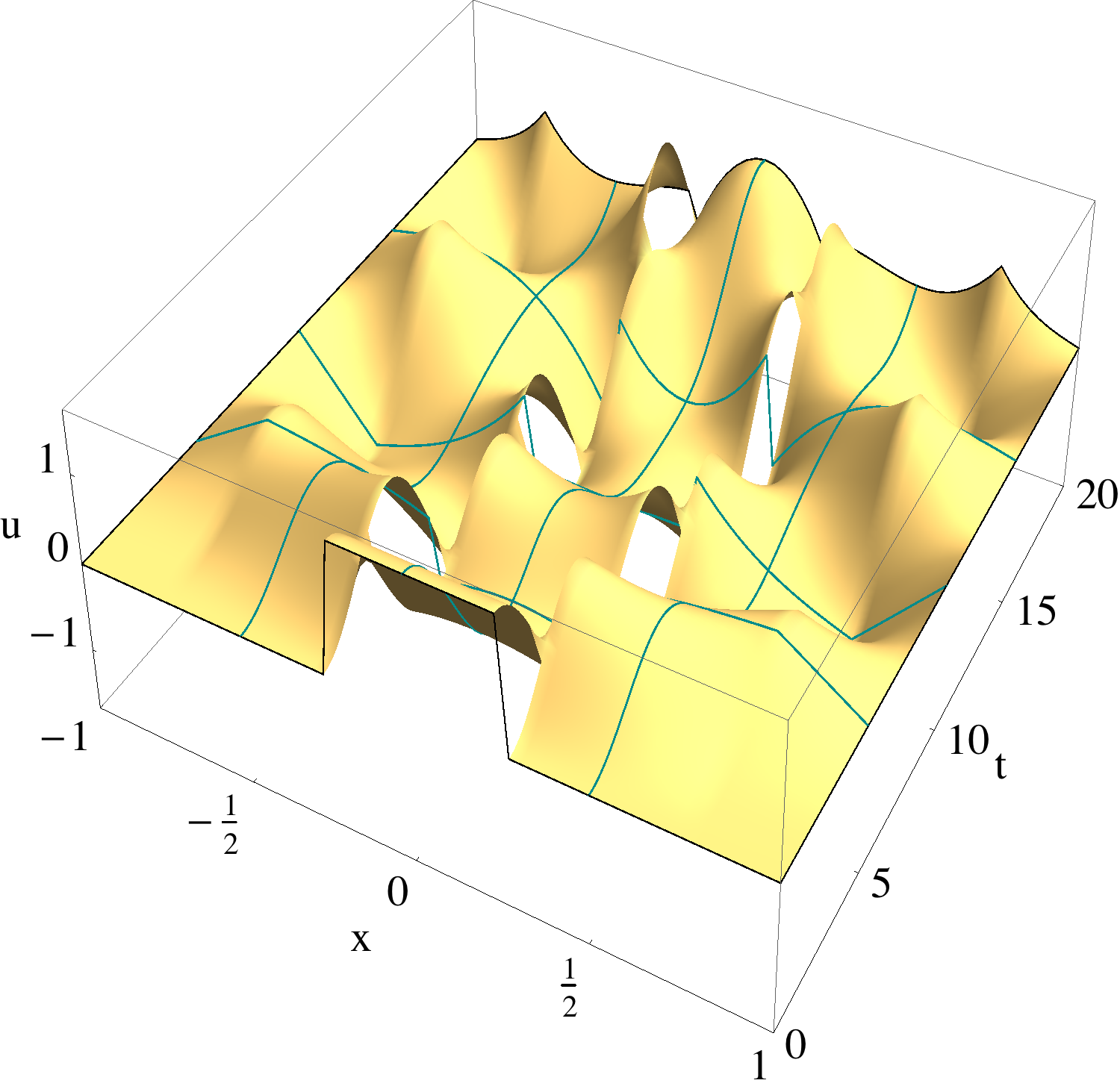}}
\label{antiperiodicDataView}
}
\hspace{.35cm}
\subfigure[The same solution from Fig.~\ref{antiperiodicDataView} from a
boundary point of view.]{
\scalebox{0.50}{\includegraphics{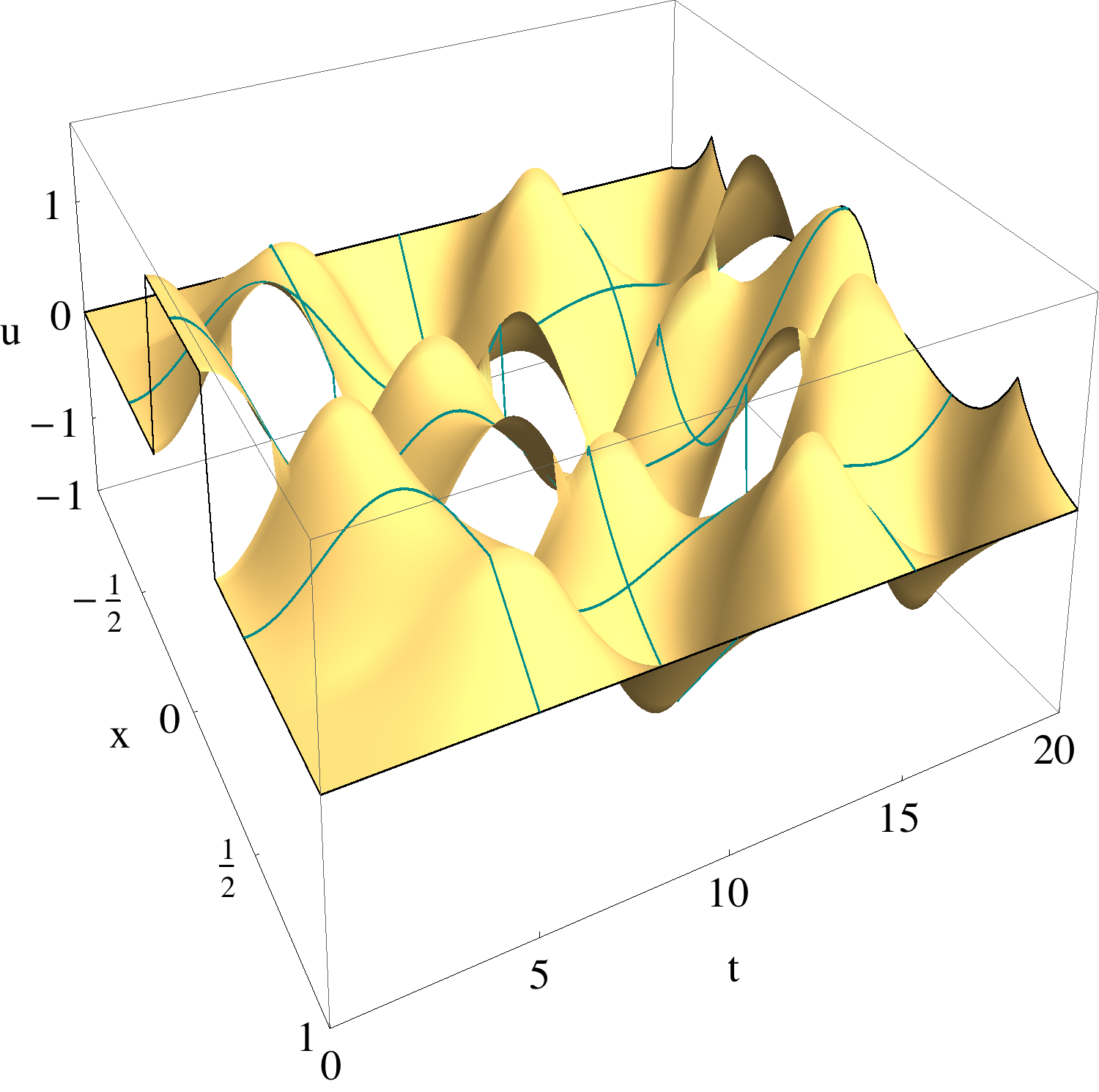}}
\label{antiperiodicBoundaryView}
}
\caption{Solution to the nonlocal wave equation solution with
antiperiodic boundary conditions and vanishing initial velocity.}
\label{fig:antiperiodic}
\end{figure}

%%% neumann
\begin{figure}[t]
\centering
\subfigure[Regulating function.]{
\scalebox{0.35}{\includegraphics{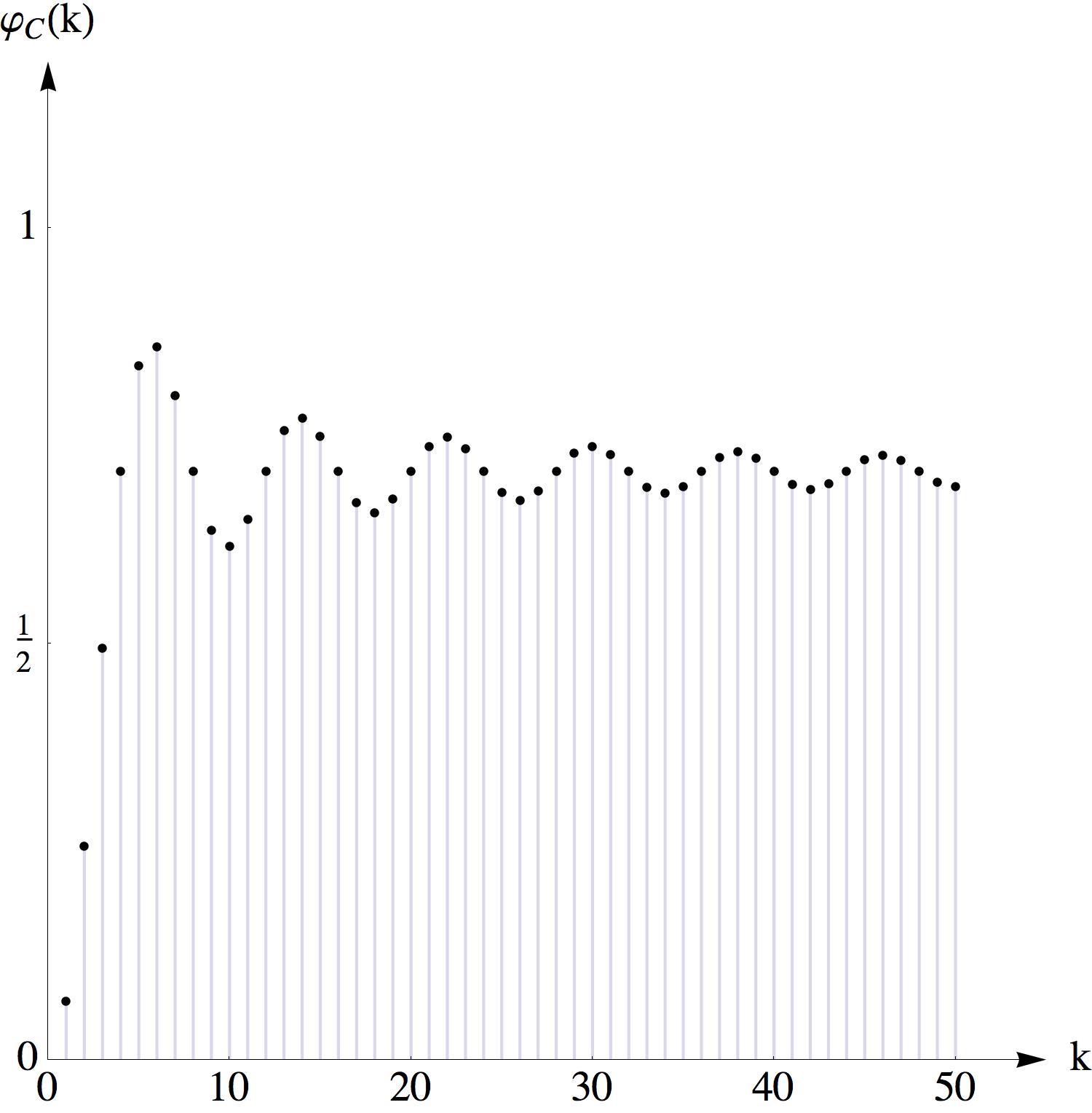}}
\label{neumannRegulatingFunc}
}
\hspace{.35cm}
\subfigure[Contour plot of $u$ from Fig.~\ref{neumannDataView}.]{
\scalebox{0.70}{\includegraphics{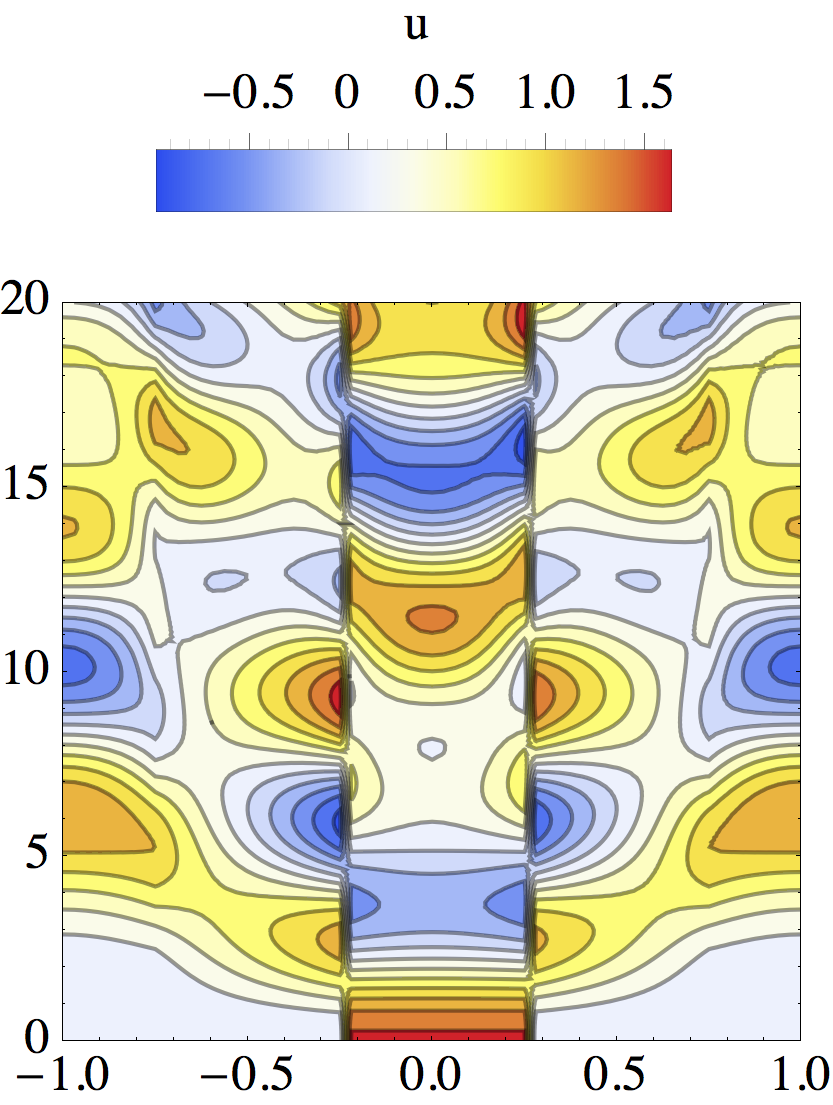}}
\label{neumannContourPlot}
}
\\
\subfigure[Solution $u$ to the nonlocal wave equation
with initial data $u(x,0) = u_{0,\textrm{disc}}(x)$
and $u_t(x,0)=0,~x \in (-1,1)$. Initial data view.]{
\scalebox{0.50}{\includegraphics{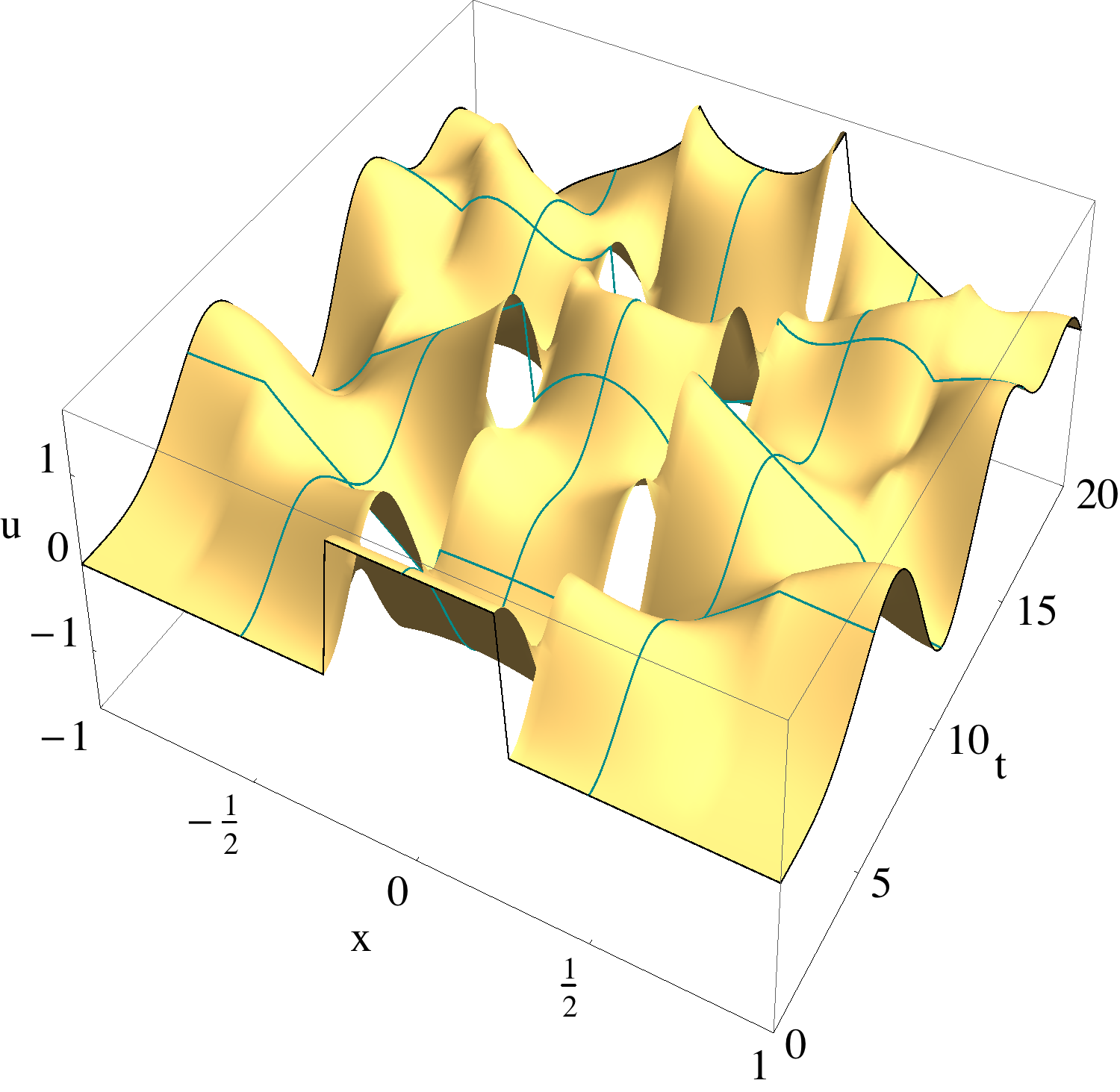}}
\label{neumannDataView}
}
\hspace{.35cm}
\subfigure[The same solution from Fig.~\ref{neumannDataView} from a
boundary point of view.]{
\scalebox{0.50}{\includegraphics{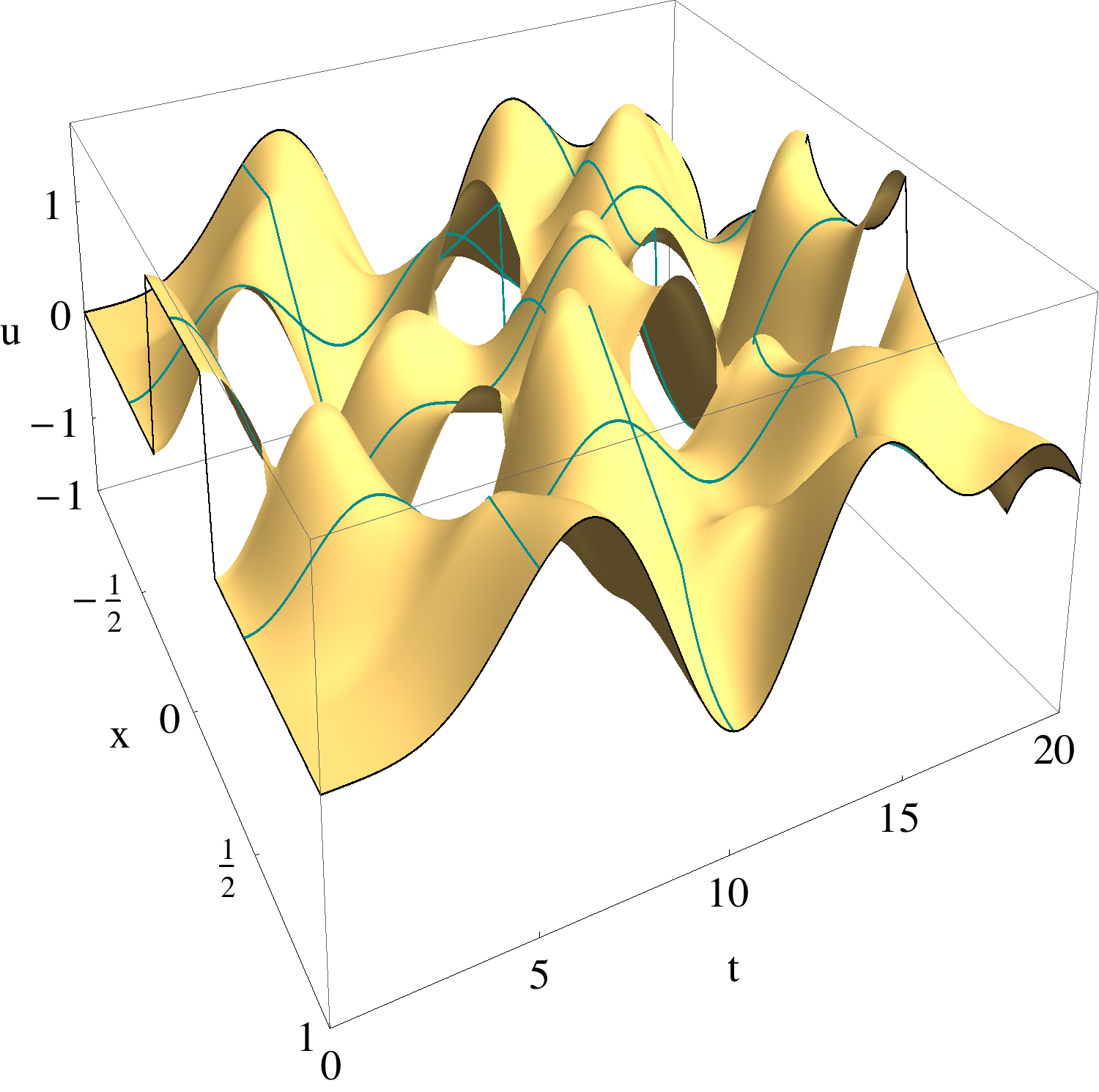}}
\label{neumannBoundaryView}
}
\caption{Solution to the nonlocal wave equation solution with
Neumann boundary conditions and vanishing initial velocity.}
\label{fig:neumann}
\end{figure}

%%% dirichlet
\begin{figure}[t]
\centering
\subfigure[Regulating function.]{
\scalebox{0.35}{\includegraphics{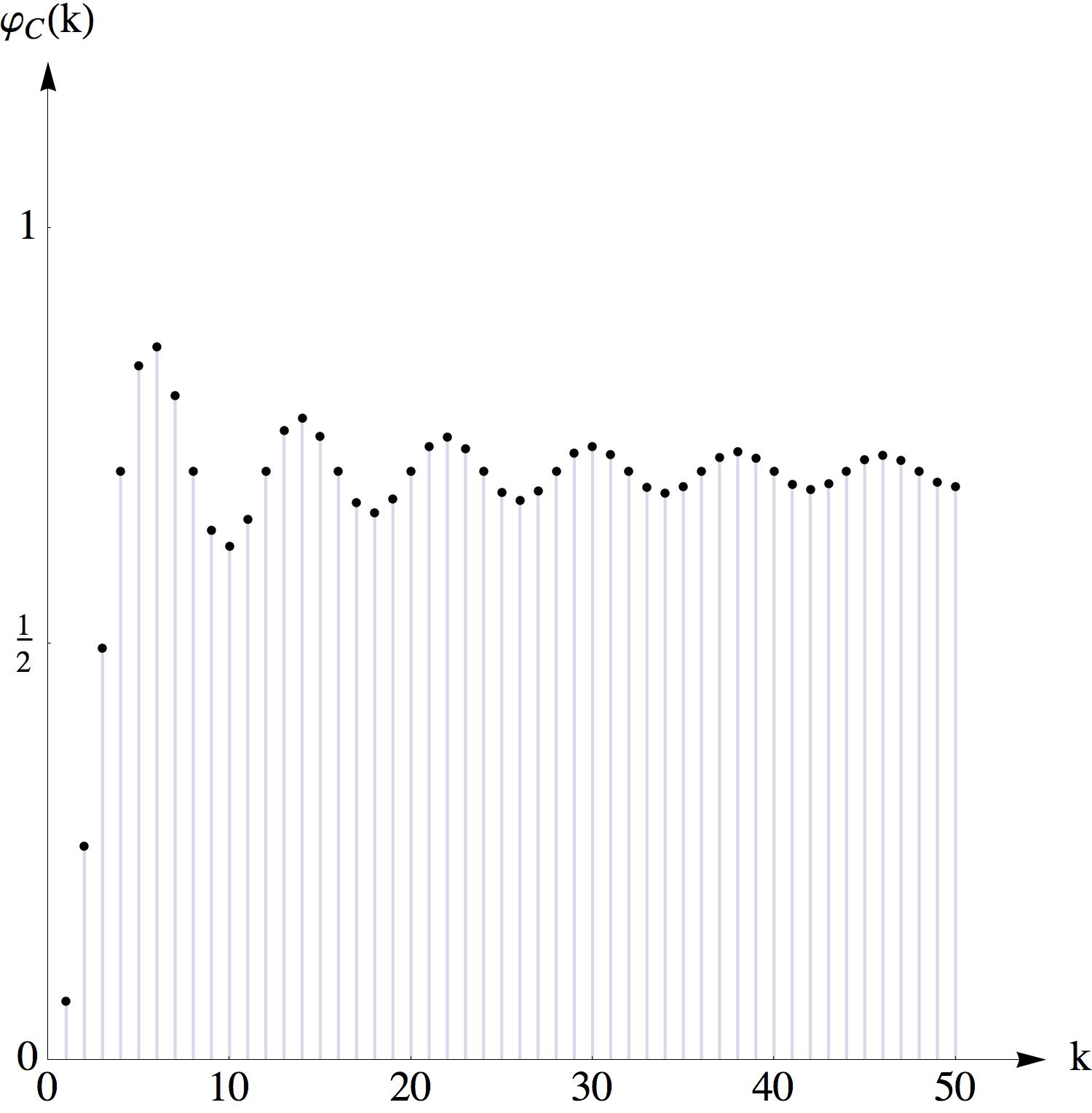}}
\label{dirichletRegulatingFunc}
}
\hspace{.35cm}
\subfigure[Contour plot of $u$ from Fig.~\ref{dirichletDataView}.]{
\scalebox{0.70}{\includegraphics{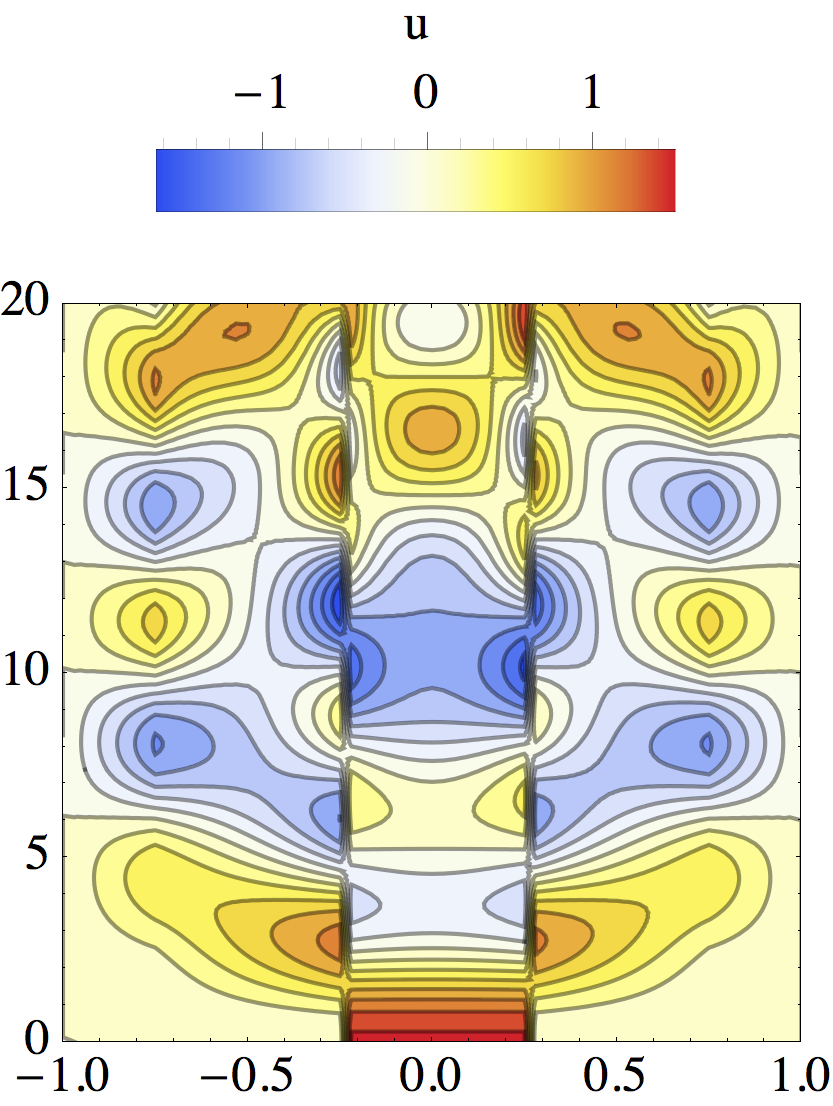}}
\label{dirichletContourPlot}
}
\\
\subfigure[Solution $u$ to the nonlocal wave equation
with initial data $u(x,0) = u_{0,\textrm{disc}}(x)$
and $u_t(x,0)=0,~x \in (-1,1)$. Initial data view.]{
\scalebox{0.50}{\includegraphics{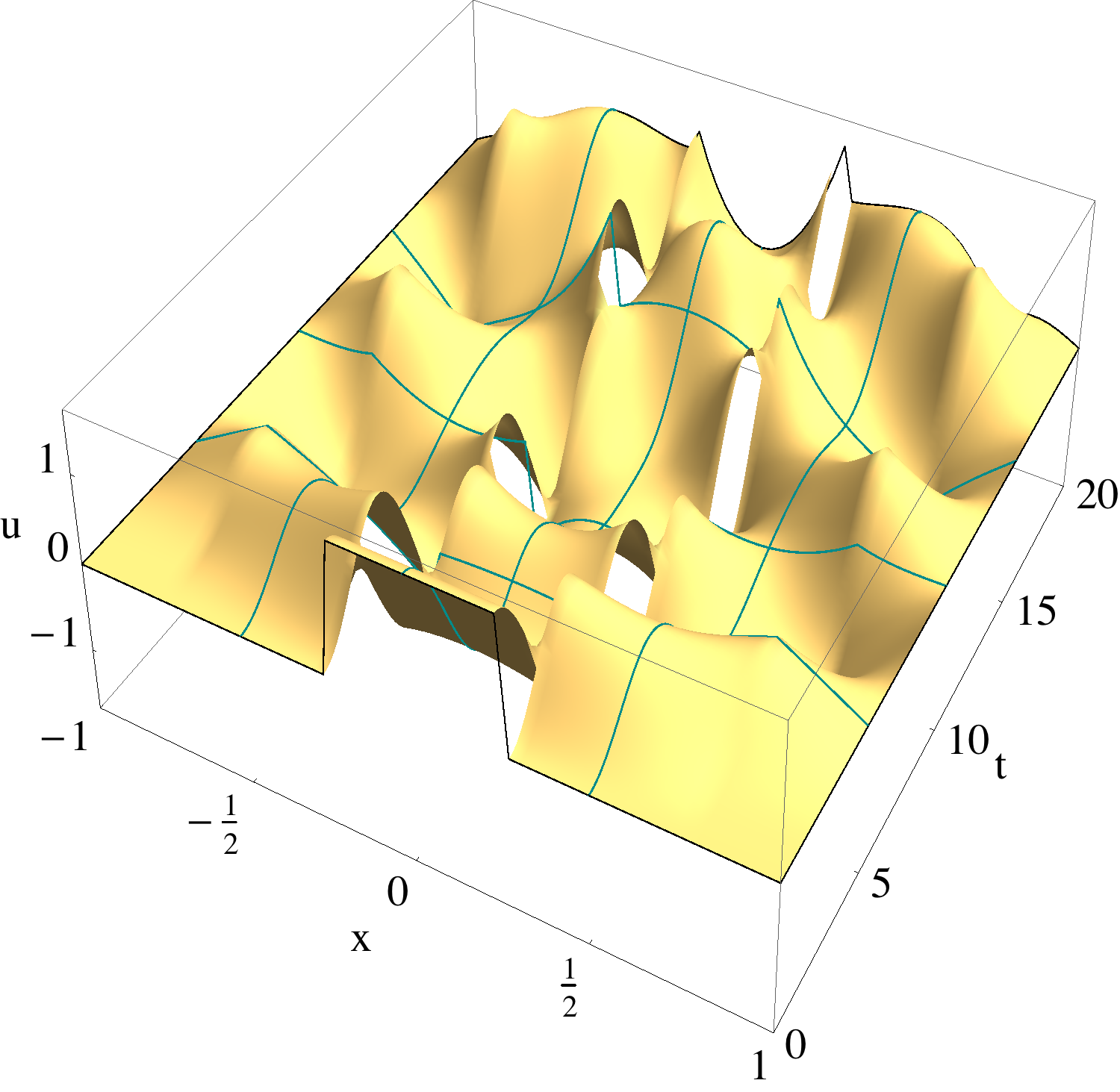}}
\label{dirichletDataView}
}
\hspace{.35cm}
\subfigure[The same solution from Fig.~\ref{dirichletDataView} from a
boundary point of view.]{
\scalebox{0.50}{\includegraphics{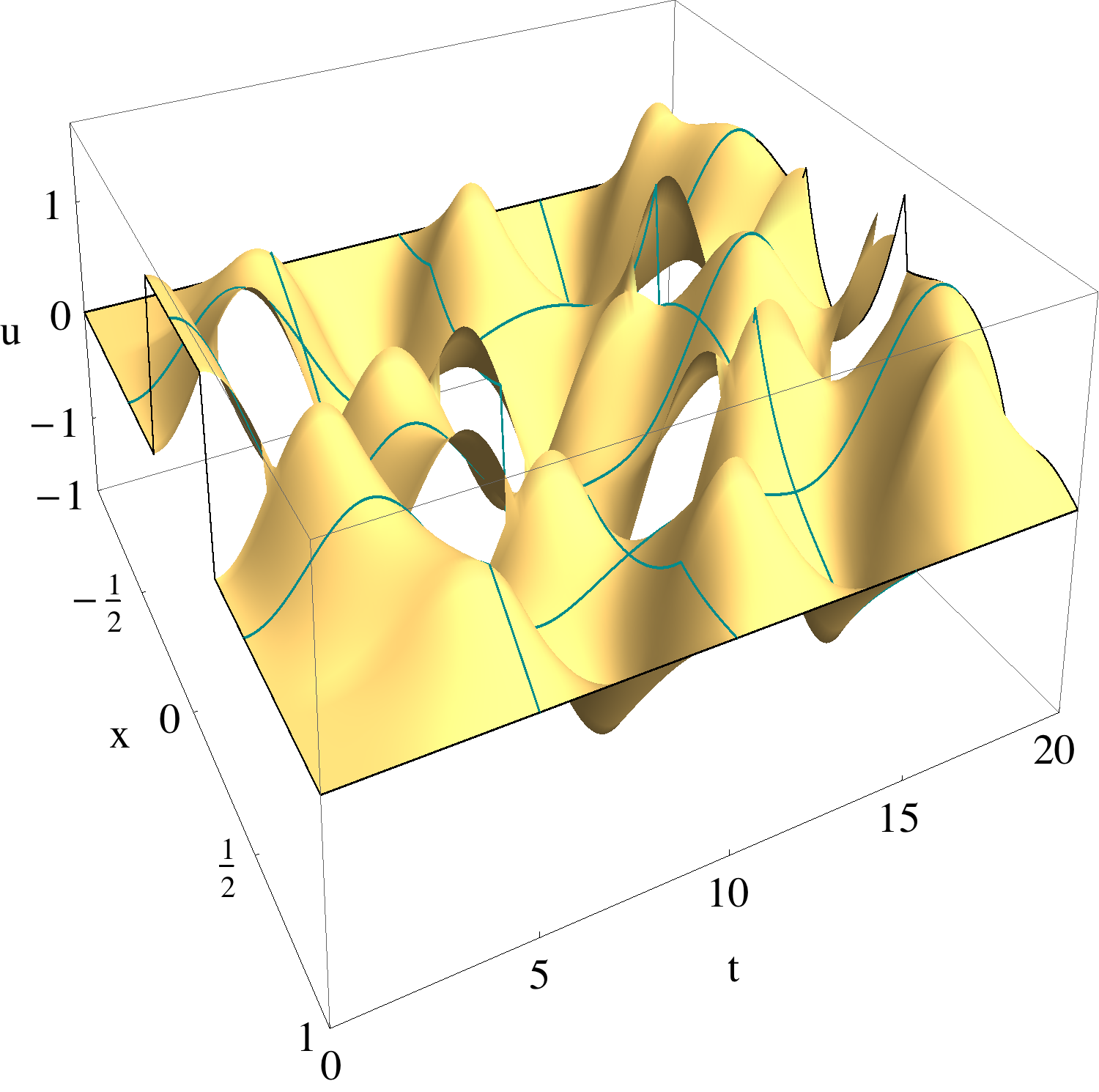}}
\label{dirichletBoundaryView}
}
\caption{Solution to the nonlocal wave equation solution with
Dirichlet boundary conditions and vanishing initial velocity.}
\label{fig:dirichlet}
\end{figure}

To approximate the solution of \eqref{governing} we begin with
discretizing the domain $\Omega$ into $N$ subintervals by defining
$\oh = \{K_1,K_2,\ldots,K_N\}$ where $K_i = (x_{i-1},x_i)$ with
$-1 = x_0 < x_1 < \cdots < x_{N-1} < x_N = 1$. We let
$h_i = |K_i| = x_i - x_{i-1}$ for $i=1,\ldots,N$.
Given a polynomial degree $\ell \ge 0$, we wish to approximate the solution
$u(x,t)$ of \eqref{governing} for a fixed $t$ in the finite element space
\begin{equation*}
V_h = \{ v \in L^2(\Omega) : v|_K \in P_{\ell}(K) \mbox{ for all } K \in \oh\}
\end{equation*}
where $P_{\ell}(K)$ is the space of polynomials of degree at most $\ell$ on $K$.

We define the $L^2$-inner product on an element $K\in \oh$ as
\[
(u,v)_K = \int_K u(x) v(x) \, dx \quad\mbox{ and set }\quad
(u,v)_{\oh} = \sum_{K\in\oh}(u,v)_K.
\]
For approximation of \eqref{governing}, we use a Galerkin projection
as used in \cite{aksoyluParks2011,aksoyluUnlu2014_nonlocal} and
consider the following (semidiscrete) approximation:
Find $u^h : J \times V_h \to \mathbb{R}$ such that
\begin{subequations}
\label{DG}
\begin{alignat}{1}
(u_{tt}^h, v)_{\oh} + (\varphi(A_{\BC})u^h,v)_{\oh} &= (b,v)_{\oh} \quad \mbox{ for all } v \in V_h, \\
u^h|_{t=0}   &= \Pi_h u_0, \\
u_t^h|_{t=0} &= \Pi_h v_0.
\end{alignat}
\end{subequations}
Here, $\Pi_h$ denotes the $L^2$-projection onto $V_h$.

\subsection{Discretization in Time}
The discretization of \eqref{governing} by the Galerkin method \eqref{DG}
leads to the second-order system of ordinary differential equations
\begin{equation}
\label{DG-ODE}
\bM \ddot{\bu}(t) + \bA\bu^h(t) = \bb^h(t), \qquad t \in J,
\end{equation}
with initial conditions
\begin{equation}
\label{DG-ODE-IC}
\bM\bu^h(0) = \bu_0^h, \qquad
\bM\dot{\bu}^h = \bv_0^h.
\end{equation}
Here, $\bM$ denotes the mass matrix and $\bA$ denotes the stiffness matrix.
To discretize \eqref{DG-ODE}-\eqref{DG-ODE-IC} in time, we employ the
Newmark time-stepping scheme as described in
\cite{groteSchneebeliSchotzau2006}, also see, e.g.,
\cite{raviartThomas1983_book}.

Let $k$ denote the time step and set $t_n = n\cdot k$ for $n = 1,2,\ldots$ \, .
The Newmark scheme we employ consists in finding approximations
$\{\bu_n^h\}_n$ to $\bu^h(t_n)$ such that
\begin{eqnarray}
\bM \ddot{\bu}_1^h & = & \left( \bM - \frac12 k^2\bA \right)\bu_0^h
+ k \bM \bv_0^h + \frac12 k^2 \bb_0^h, \label{Newmark-1} \\
\bM \ddot{\bu}_{n+1}^h & = & \left( 2\bM - k^2\bA \right)\bu_n^h
- \bM \bu_{n-1}^h + k^2 \bb_n, \label{Newmark-n}
\end{eqnarray}
for $n=1,2,\ldots,N_{t}-1$ where $N_{t}k = T$, and $\bb_n = \bb(t_n)$.
Although there is a more general version of the Newmark time-stepping scheme,
we made this particular choice due to the fact that it is second-order accurate
and is explicit in the sense that at each time step we only have to solve
a linear system with a coefficient matrix $\bM$ that is block diagonal.
Hence, $\bM$ can be inverted at a very low computational cost.
For other Newmark schemes
the coefficient matrix of the linear system would be
$\bM + k^2 \beta \bA$ for some $\beta > 0$ which needs
to be inverted at each time step.
For a detailed discussion of more general Newmark time integration schemes
we refer to \cite{groteSchneebeliSchotzau2006}.

\subsection{Implementation Details}
Let us describe a few details regarding the computation of
the stiffness matrix $\bA$. Let $K \in \oh$ and let
$\{\phi_j^K \;:\; j=1,\ldots,\ell+1\}$ be a basis for $P_{\ell}(K)$.
To fix ideas, let us consider the case $\BC = \p$ so that
\[
\varphi(A_{\BC})u(x,t) = (c-C*_{\p})u(x,t).
\]
The remaining cases ${\BC}=\aBC,\N,\D$ are similar.

First of all, we need to compute the constant $c$ in \eqref{sillingConst}.
In the cases where $C$ is an elementary function such as a (piecewise) polynomial,
the exact value of this constant can be computed by direct integration.
However, in the general case, we have to resort to numerical quadrature.
We simply compute
\[
c =  \frac{1}{\sqrt{2}} \sum_{K \in \oh} \int_K C(x)dx
\]
where the integral on each element $K \in \oh$ is approximated by a quadrature rule.
In this case, if $C$ happens to have discontinuities or kinks in $\Omega$,
in order to obtain an accurate approximation to $c$, we have to ensure that the
nodes of the discrete domain $\oh$ are aligned with these discontinuities.

The matrix $\bA$ is of size $N(\ell+1) \times N(\ell+1)$ and has a block structure.
Each block-row of size $(\ell+1) \times N(\ell+1)$ corresponding to an element $K \in \oh$
is determined by the equations
\[
(\varphi(A_{\p})u^h,\phi_i^K)_{K} = (b,\phi_i^K)_{K},
\qquad\mbox{for } i=1,2,\ldots,\ell+1.
\]
Inserting the definition of $\varphi(A_{\p})$, we get
\[
\begin{aligned}
(\varphi(A_{\p})u^h,\phi_i^K)_{K}
&= ((c-C*_{\p})u^h,\phi_i^K)_K \\
&= c(u^h,\phi_i^K)_K - (C*_{\p} u^h,\phi_i^K)_K .
\end{aligned}
\]
The computation of the first term is standard, but we would like elaborate
on a few details regarding the computation of the second term.
At any fixed time $t \in J$ and for a fixed element $T \in \oh$,
we have the restriction, $u^{h,T}$, of $u^h$ on $T$ has the expansion
\[
u_h^{h,T}(x,t) = \sum_{j=1}^{\ell+1} u_j^T(t)\phi_j^T(x).
\]
Then, since
\[
\begin{aligned}
C*_{\p} u^h(x,t) &= \int_{-1}^1 \hat{C}_{\p}(x-y)\,u^h(y,t)\,dy \\
&= \sum_{T \in \oh} \sum_{j=1}^{\ell+1} u_j^T(t)\int_T \hat{C}_{\p}(x-y) \phi_j^T(y)\,dy,
\end{aligned}
\]
we have
\begin{equation}
\label{A_TK}
(C*_{\p} u^h,\phi_i^K)_K
= \sum_{T \in \oh} \sum_{j=1}^{\ell+1} u_j^T(t) \int_K R_j^T(x) \phi_j^K(x) \, dx
\end{equation}
where
\[
R_j^T(x) := \int_T \hat{C}_{\p}(x-y) \phi_j^T(y)\,dy.
\]
Thus, we need to compute pointwise values of $R_j^T$ which will be achieved
through numerical quadrature.
Note that the micromodulus function $C$ may have points of discontinuities
or kinks (or higher order derivatives of $C$ may not be continuous) in $\Omega$.
Hence, when computing $R_j^T(x)$, we need to take these points into account, for example,
when using Gaussian quadrature which requires the smoothness of the integrand
for optimal order accuracy. Furthermore, even if $C$ is arbitrarily smooth in $\Omega$,
its extension $\hat{C}_{\p}$ may not be smooth in $[-2,2]$.
Since the integrand involves $\hat{C}_{\p}(x-y)$ which is a translation of
$\hat{C}_{\p}(-y) = \hat{C}_{\p}(y)$ by $x$ units to the left, we always have to account for
possible singularities of $\hat{C}_{\p}(y)$ at the end points, $\{-1,1\}$, of the domain
$\Omega$.
Suppose $y_s \in T$ is such that $\hat{C}_{\p}(x-y_s)$ has a singularity in $K$.
Then the integral defining $R_j^T(x)$ has to be computed
by writing $T = T_1 \cup T_2$ where $T=(x_L,x_R)$,
$T_1 = (x_L,y_s)$ and $T_2 = (y_s,x_R)$, and applying
numerical quadrature on both subintervals.
A similar treatment is needed when computing the integral
$\int_K R_j^T(x) \phi_j^K(x) \, dx$.

Due to the nonlocal nature of the problem, the stiffness
matrix $\bA$ is not necessarily sparse. This can be seen from
\eqref{A_TK} by observing that $R_j^T$ does not necessarily vanish on
the element $K$ for $T \ne K$. The sparsity structure of $\bA$ is
determined by the support of the micromodulus function $C$. More
explicitly, the wider the support of $C$, the less sparse $\bA$ is.
Symmetry and positive definiteness of the stiffness matrix are
the consequences of the self-adjointness and positivity of the
governing operator, respectively; see Theorem \ref{thm:positivity}.
For the case of periodic and Neumann BCs, the stiffness matrix becomes
positive semidefinite and these systems can be solved by using
numerical methods described in
\cite{bochevLehoucq2005,gockenbach2006_book}.  Finally, we would like
to point out that the assembly of the stiffness matrix as well as the
mass matrix is independent of the time step and is performed only once.

\subsection{Approximations to Explicitly Known Exact Solutions}

Note that, since the operator
$\varphi(A_{\BC})$ is different for each BC, 

In order to ascertain the convergence performance of the scheme
described above, we display some numerical results corresponding to
explicitly known exact solutions.  We solve one example corresponding
to each BC type.  We take the exact solution
corresponding to each BC as given in Table \ref{table:exactSolutions}
and compute the corresponding right-hand side function $b(x,t)$.

%%% Continuous data plots
%%% neumann and dirichlet
\begin{figure}[h]
\centering
\subfigure[Contour plot of $u$ from Fig.~\ref{neumannDataViewCont}.]{
\scalebox{0.70}{\includegraphics{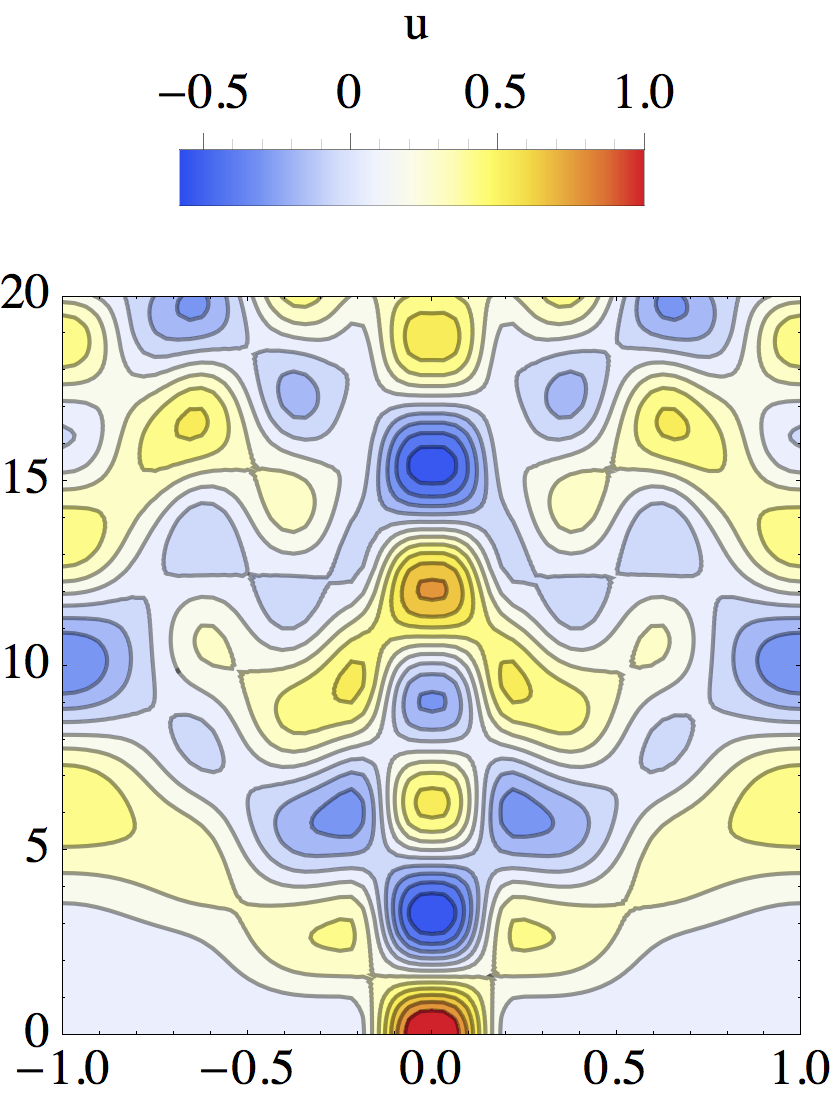}}
\label{neumannContourPlotCont}
}
\hspace{.35cm}
\subfigure[Contour plot of $u$ from Fig.~\ref{dirichletDataViewCont}.]{
\scalebox{0.70}{\includegraphics{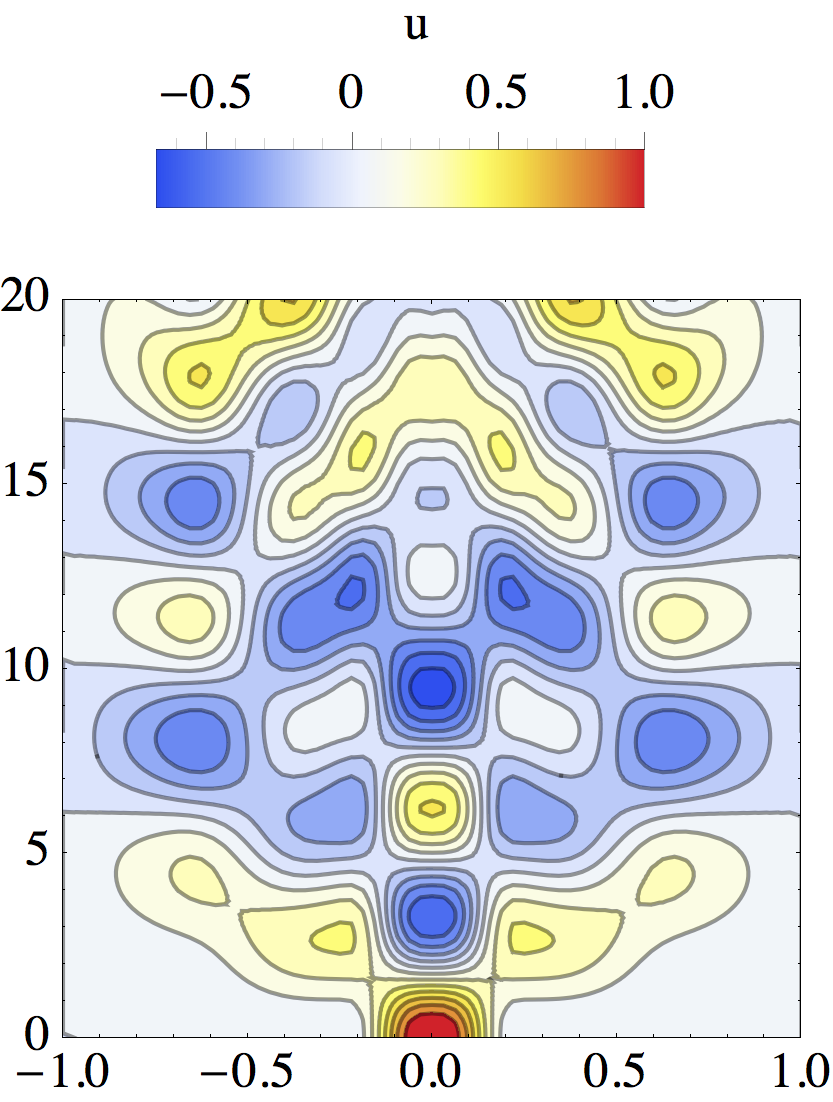}}
\label{dirichletContourPlotCont}
}
\\
\subfigure[Solution $u$ to the nonlocal wave equation
with Neumann boundary conditions and initial data
$u(x,0) = 0$
and $u_t(x,0)=u_{0,\textrm{cont}},~x \in (-1,1)$. Initial data view.]{
\scalebox{0.50}{\includegraphics{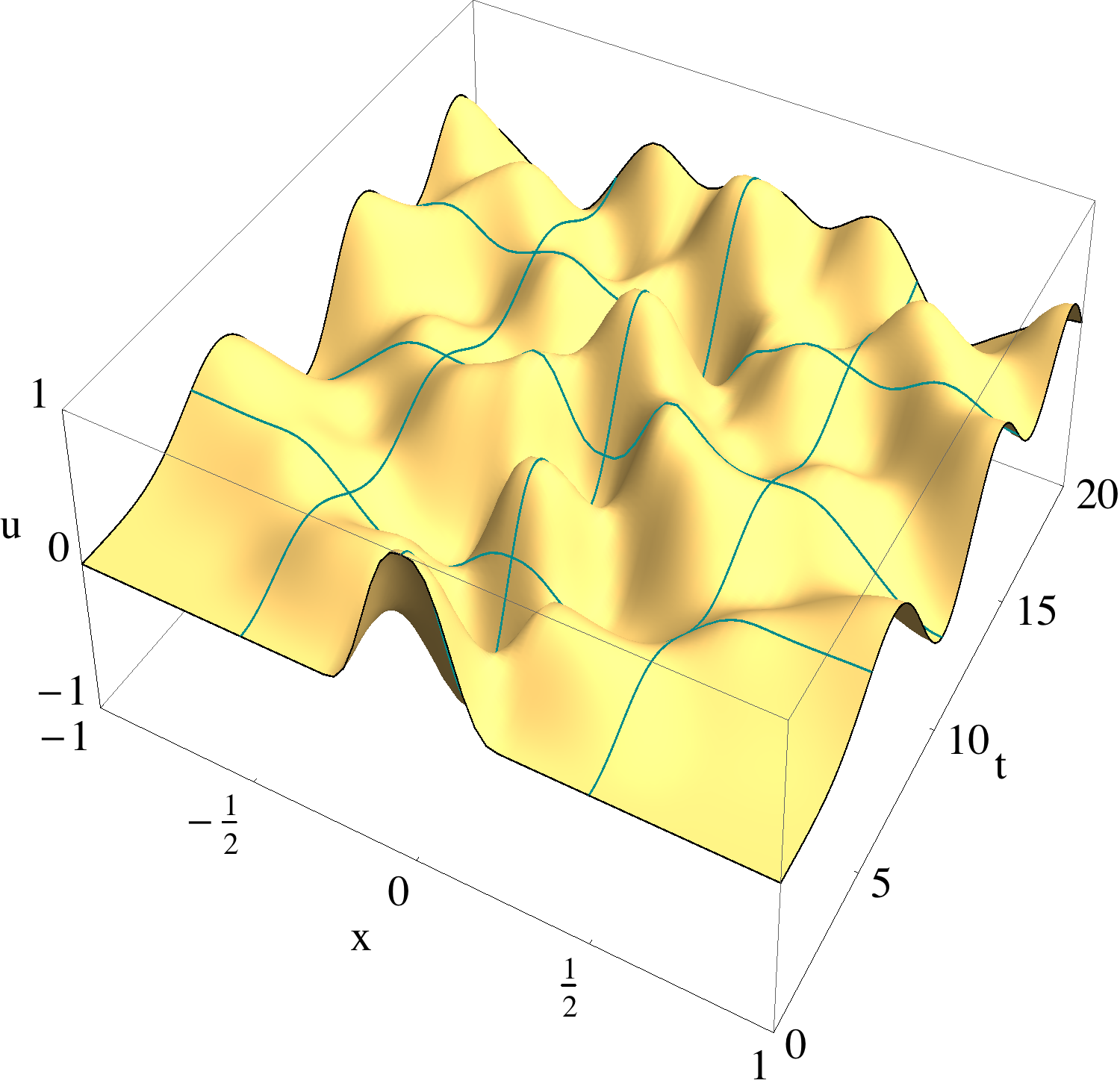}}
\label{neumannDataViewCont}
}
\hspace{.35cm}
\subfigure[Solution $u$ to the nonlocal wave equation
with Dirichlet boundary conditions and initial data
$u(x,0) = u_{0,\textrm{cont}}(x)$
and $u_t(x,0)=0,~x \in (-1,1)$. Initial data view.]{
\scalebox{0.50}{\includegraphics{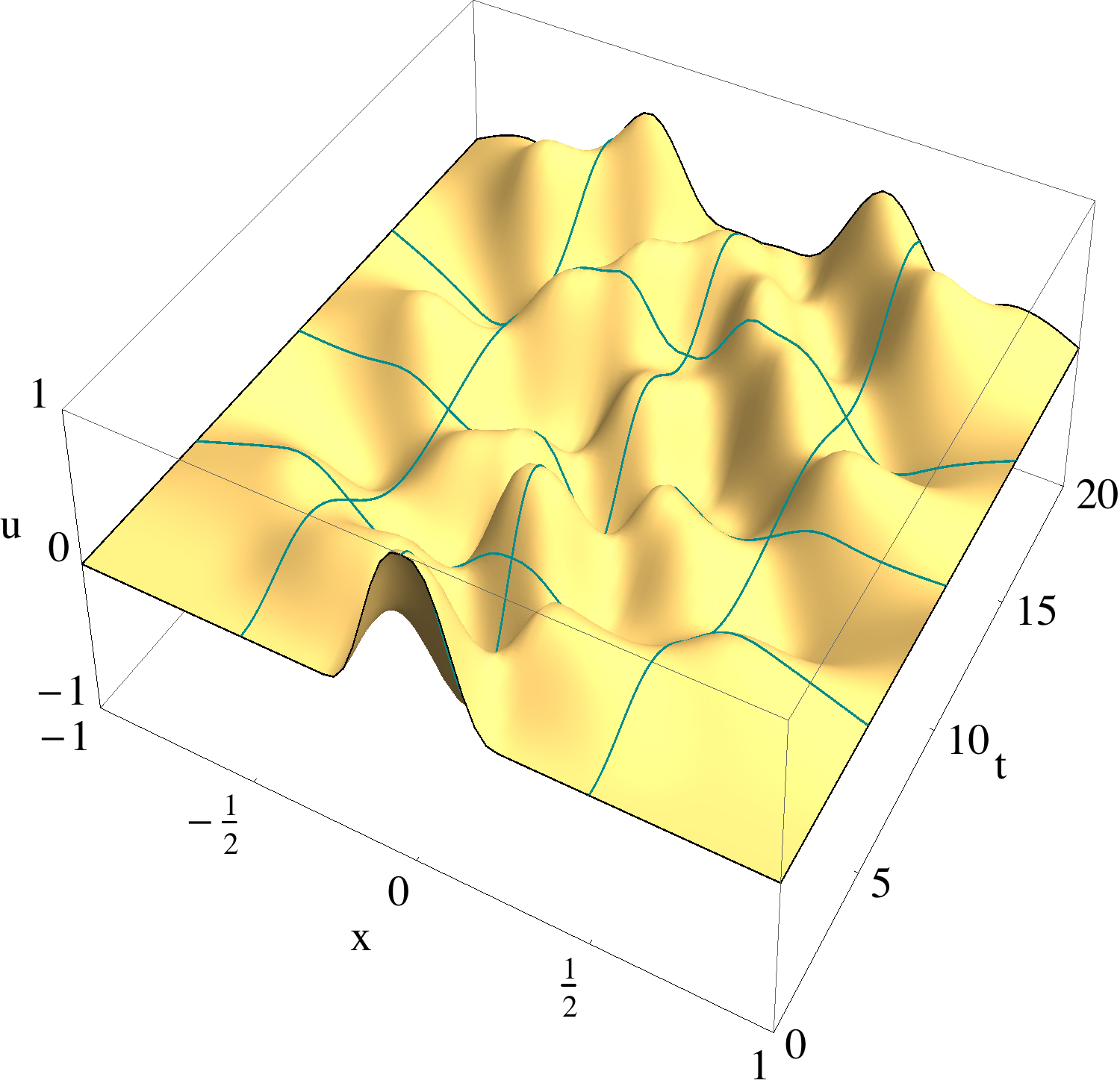}}
\label{dirichletDataViewCont}
}\\
\subfigure[The same solution from Fig. ~\ref{neumannDataViewCont} from a
boundary point of view.]{
\scalebox{0.50}{\includegraphics{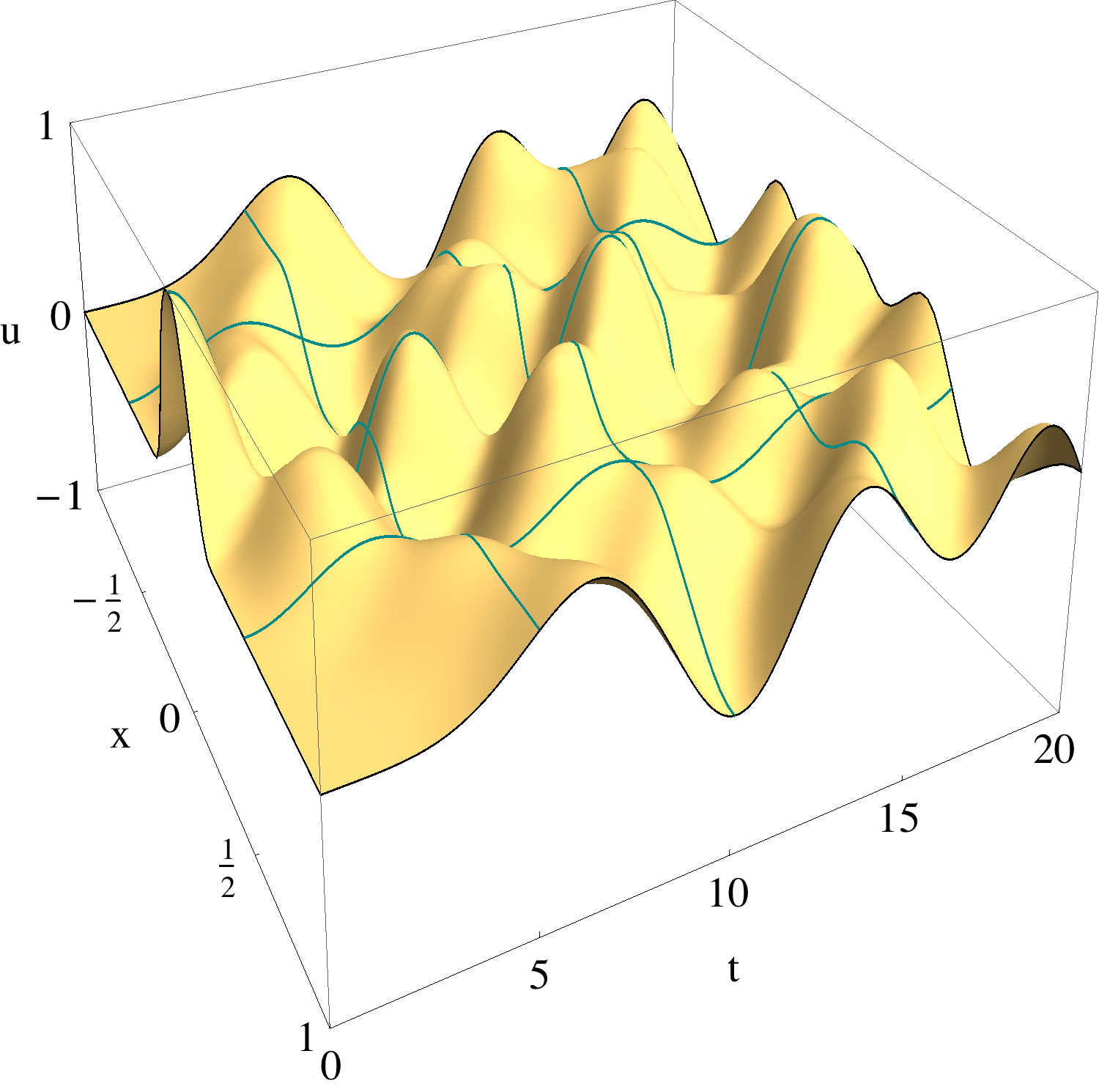}}
\label{neumannBoundaryViewCont}
}
\hspace{.35cm}
\subfigure[The same solution from Fig.~\ref{dirichletDataViewCont} from a
boundary point of view.]{
\scalebox{0.50}{\includegraphics{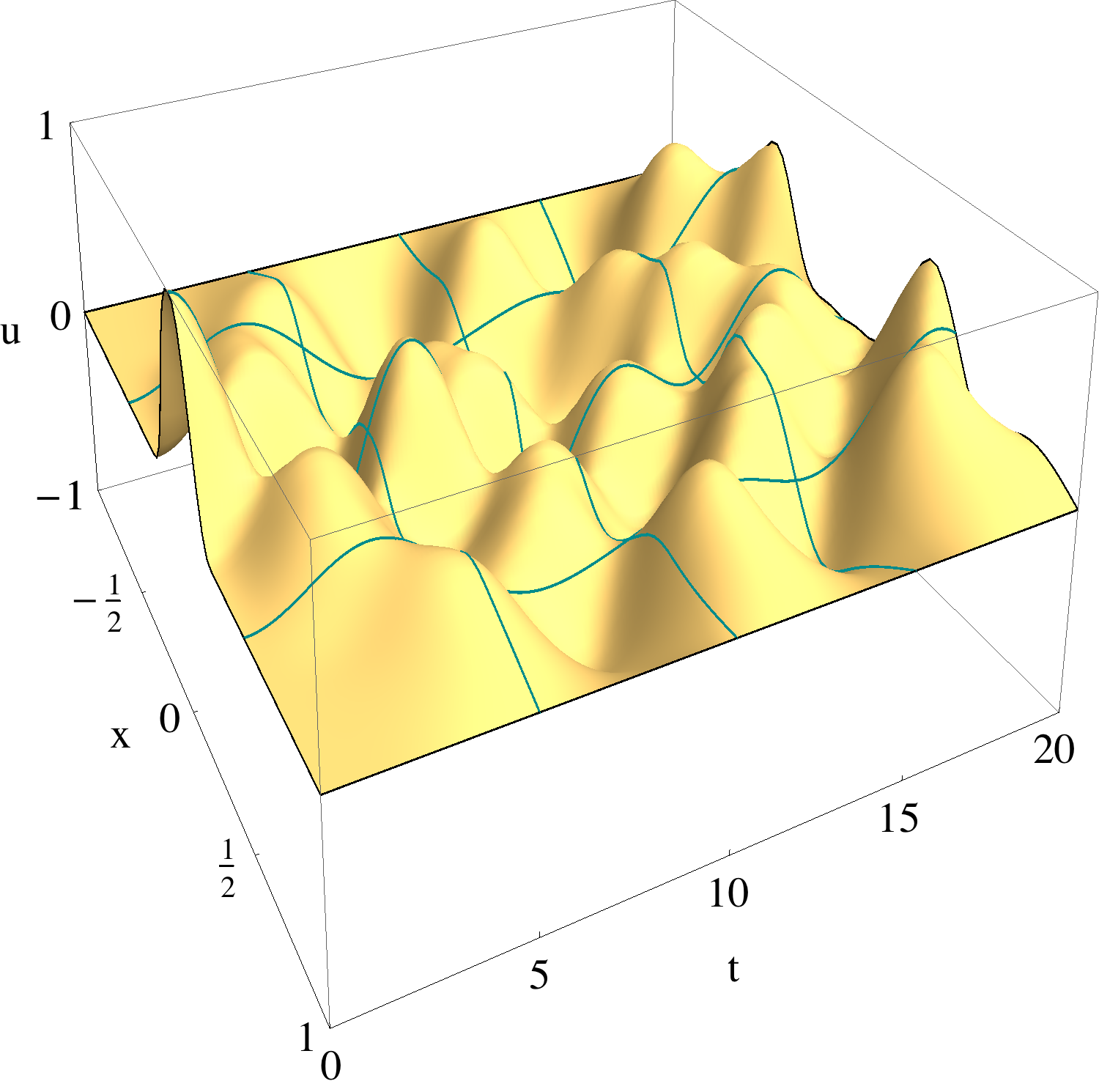}}
\label{dirichletBoundaryViewCont}
}
\caption{Solution to the nonlocal wave equation with
Neumann ((a), (c), and (e))and Dirichlet ((b), (d), (f)) boundary conditions,
continuous initial displacement and vanishing initial velocity.}
\label{fig:neumannDirichletCont}
\end{figure}

\clearpage

\begin{table}[h]
\centering
\begin{tabular}{c|l}
\hline
& \\ [-1ex]
$\BC$  &$\qquad\quad u(x,t)$ \\ [1ex]
\hline & \\ [-.5ex]
$\p$   &$t^2(\sin(\pi x) + \cos(\pi x))$ \\ [1ex]
$\aBC$ &$t^2(x^4-1)$ \\ [1ex]
$\N$   &$t^2((x^2-1)^2-8/15)$ \\ [1ex]
$\D$   &$t^2(1+\sin(\pi x) + \cos(\pi x))$ \\ [1ex]
\hline
\end{tabular}
\caption{Known exact solutions used in numerical experiments.}
\label{table:exactSolutions}
\end{table}

%% Classical wave equation with Dirichlet BC
\begin{figure}[h]
\centering
%%% Classical wave equation with Neumann BC
\subfigure[Solution $u$ to the classical (local) wave equation with
initial data $u(x,0) = u_{0,\textrm{cont}}(x)$ defined in \eqref{initialDisplacement} and
$u(x,0)=0.$]{
\scalebox{0.50}{\includegraphics{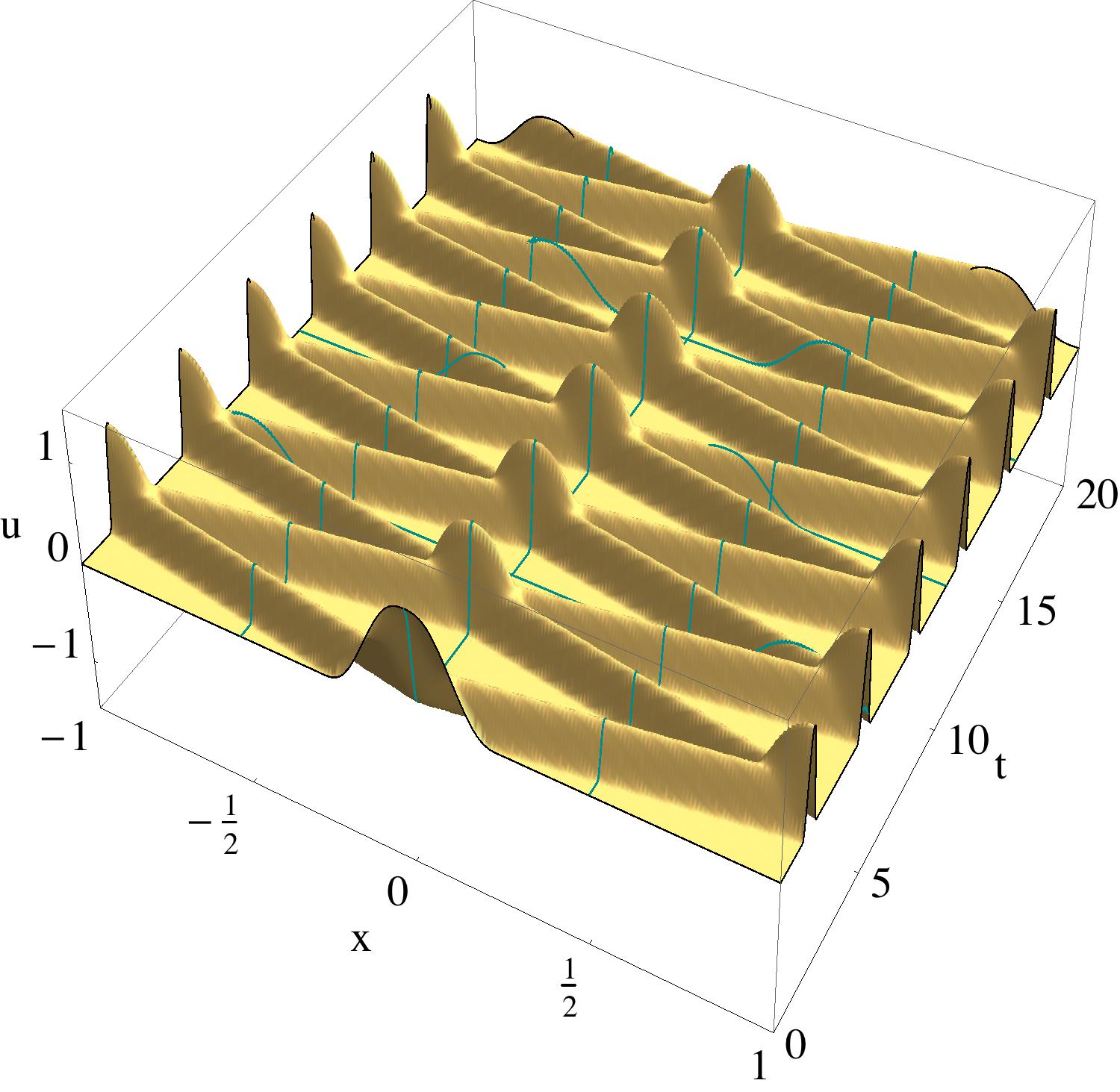}}
\label{neumannClassicalCosinus}
}
\hspace{.35cm}
\subfigure[Solution $u$ to the classical (local) wave equation with
initial data $u(x,0) = u_{0,\textrm{cont}}(x)$ defined in \eqref{initialDisplacement} and
$u_t(x,0)=0.$]{
\scalebox{0.50}{\includegraphics{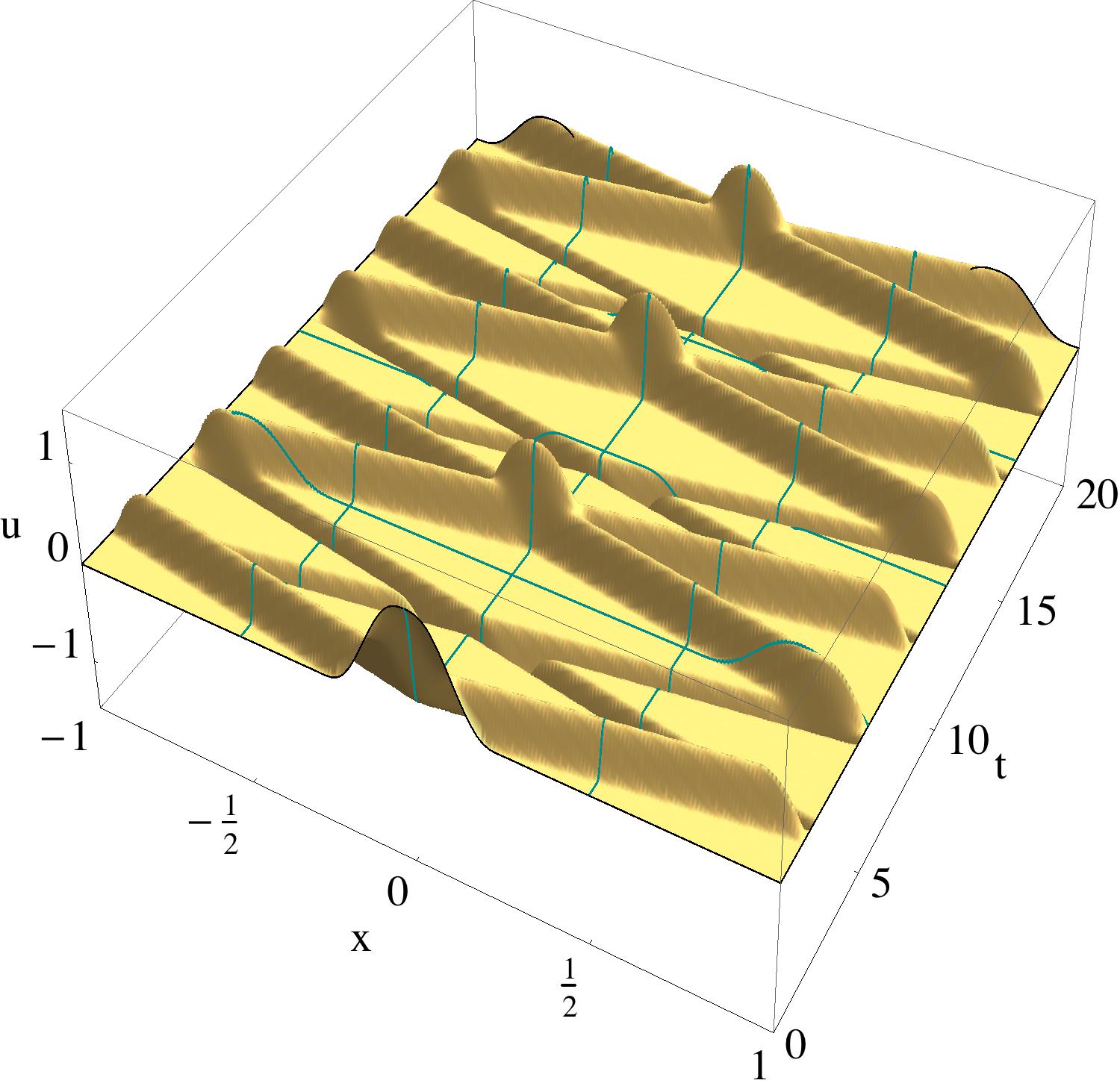}}
\label{dirichletClassicalCosinus}
}
\\
\subfigure[Solution $u$ to the classical (local) wave equation with
initial data $u(x,0) = 0$ and $u_t(x,0)= u_{0,\textrm{cont}}(x)$ defined
in \eqref{initialDisplacement}.]{
\scalebox{0.50}{\includegraphics{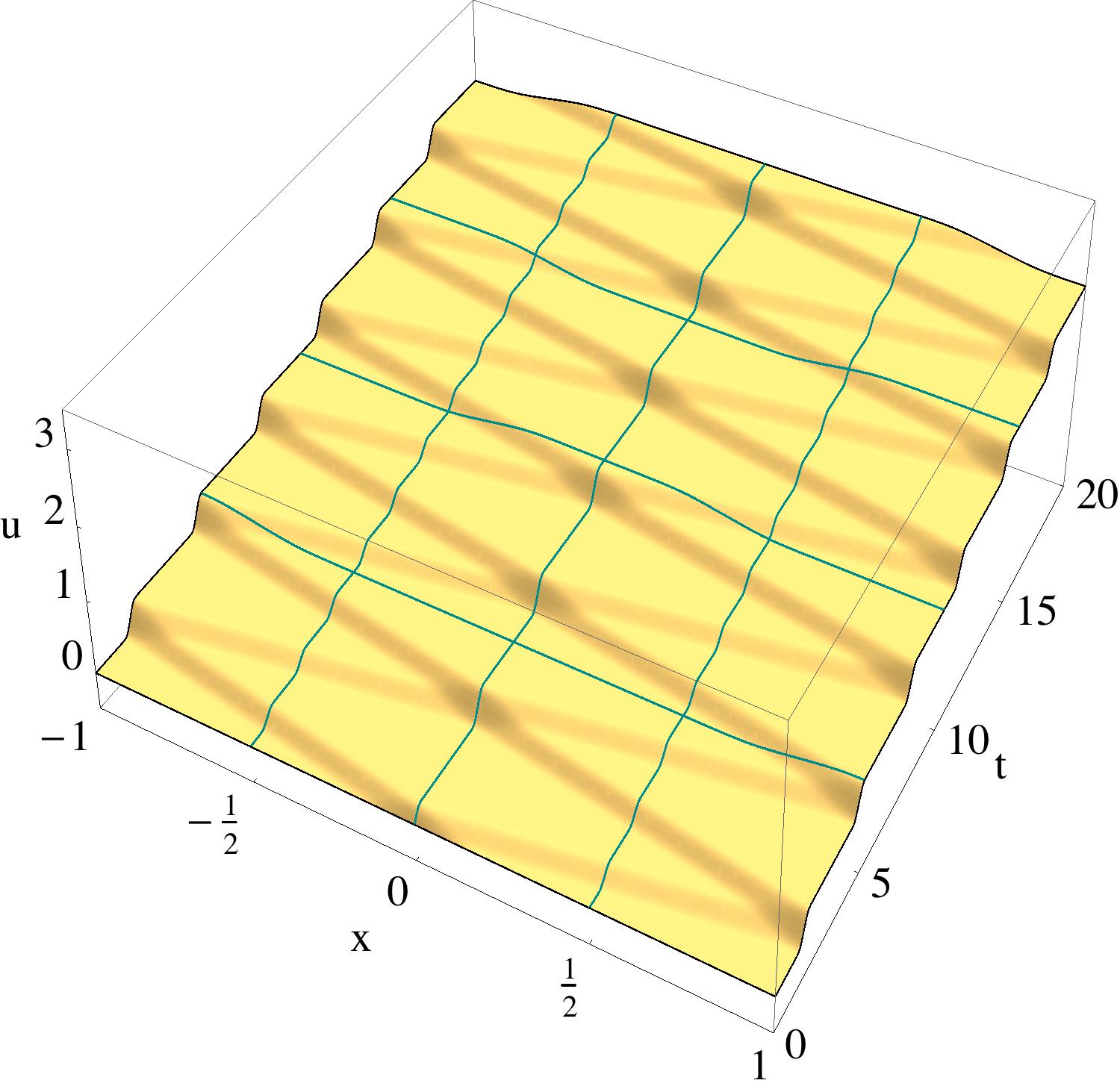}}
\label{neumannClassicalSinus}
}
\hspace{.35cm}
\subfigure[Solution $u$ to the classical (local) wave equation with
initial data $u(x,0) = 0$ and $u_t(x,0)= u_{0,\textrm{cont}}(x)$ defined
in \eqref{initialDisplacement}.]{
\scalebox{0.50}{\includegraphics{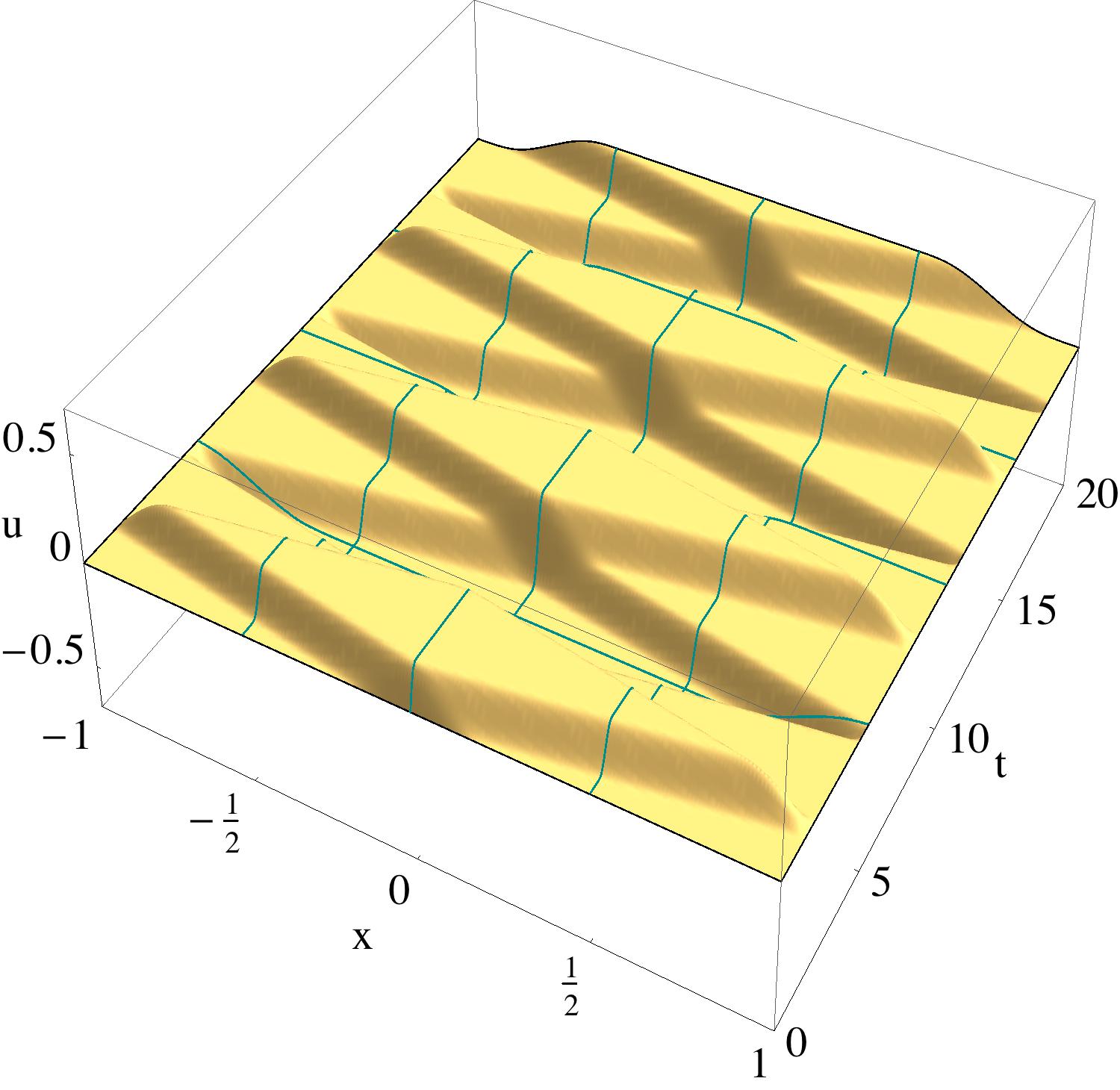}}
\label{dirichletClassicalSinus}
}
\caption{Solution to the classical wave equation with
Neumann ((a) and (c)) and Dirichlet ((b) and (d)) boundary conditions
with vanishing initial velocity ((a) and (b)) and  vanishing
initial displacement ((c) and (d)).}
\label{fig:neumannDirichletClassical}
\end{figure}

%\clearpage
\noindent
the
corresponding source term $b(x,t)$ is also expected to differ.
We take the micromodulus function $C$ to be the unit box on $\Omega$, namely,
\begin{equation}
\label{micromodulusFunc}
C(x) =
\left\{
\begin{aligned}
1, \quad & x \in [-1/2,1/2], \\
0, \quad &\mbox{otherwise,}
\end{aligned}
\right.
\end{equation}
which is depicted in Figure \ref{fig:micromodulusData}.

For each case, we compute the exact solution until the final time
$T=20$ and compute the relative $L^2$-error
$\norm{(u-u^h)(T,\cdot)}_0/\norm{u(T,\cdot)}_0$.
We first compute an approximate solution with a uniform coarse mesh with
$N=2^3$ elements and then refine the mesh by subdividing each element into
two elements of equal size. In each case, as the time step of the Newmark scheme
we take $\Delta t = 0.005$ so that the explicit Newmark time integration scheme is stable.
In all of our examples, we found out that taking $\Delta t$ so that
$\Delta t < h/10$ is sufficient for stability.
Note that since the Newmark scheme is second order accurate in time, and
all of the exact solutions in Table \ref{table:exactSolutions} is of the form
$u(x,t) = T(t)X(x)$ with $T(t)=t^2$, a second order polynomial, it is guaranteed
that the dominant error is that in the space variable.

\begin{table}[htbp]
\caption{History of convergence with known exact solutions for all BC types.}
\label{table:convergenceOrders}
\scriptsize
%\small
\begin{center}
\begin{tabular}{ccc@{\hspace{.2cm}}cc@{\hspace{.2cm}}cc@{\hspace{.2cm}}cc@{\hspace{.2cm}}c}
\hline \\
\vspace{-.1in} \\
& & \multicolumn{2}{c}{\underline{\quad\quad periodic \quad\quad}}
  & \multicolumn{2}{c}{\underline{\quad\quad antiperiodic \quad\quad}}
  & \multicolumn{2}{c}{\underline{\quad\quad Neumann \quad\quad}}
  & \multicolumn{2}{c}{\underline{\quad\quad Dirichlet \quad\quad}} \\
\vspace{-.04in} \\
$\ell$ &mesh &error &order &error &order &error &order &error &order \\
\vspace{-.05in} \\
\hline \\
\vspace{-.1in} \\
  &3 &2.32E-01 &-- &1.53E-01 &-- &2.34E-01 &-- &1.83E-01 &-- \\
  &4 &1.14E-01 &1.02 &6.88E-02 &1.15 &1.15E-01 &1.03 &8.35E-02 &1.13 \\
0 &5 &5.68E-02 &1.01 &3.29E-02 &1.06 &5.72E-02 &1.01 &4.05E-02 &1.04 \\
  &6 &2.84E-02 &1.00 &1.62E-02 &1.02 &2.85E-02 &1.00 &2.01E-02 &1.01 \\
  &7 &1.42E-02 &1.00 &8.10E-03 &1.00 &1.43E-02 &1.00 &1.00E-02 &1.00 \\
  &&&&&&&& \\
  &3 &2.28E-02 &-- &1.46E-02 &-- &2.30E-02 &-- &1.62E-02 &-- \\
  &4 &5.74E-03 &1.99 &3.69E-03 &1.99 &5.91E-03 &1.96 &4.06E-03 &1.99 \\
1 &5 &1.44E-03 &2.00 &9.25E-04 &2.00 &1.49E-03 &1.99 &1.02E-03 &2.00 \\
  &6 &3.59E-04 &2.00 &2.32E-04 &2.00 &3.73E-04 &2.00 &2.54E-04 &2.00 \\
  &7 &8.98E-05 &2.00 &5.79E-05 &2.00 &9.32E-05 &2.00 &6.35E-05 &2.00 \\
  &&&&&&&& \\
  &3 &1.52E-03 &-- &8.03E-04 &-- &2.05E-03 &-- &1.07E-03 &-- \\
  &4 &1.90E-04 &2.99 &1.01E-04 &2.99 &2.47E-04 &3.05 &1.35E-04 &2.99 \\
2 &5 &2.38E-05 &3.00 &1.26E-05 &3.00 &3.06E-05 &3.01 &1.69E-05 &3.00 \\
  &6 &2.98E-06 &3.00 &1.58E-06 &3.00 &3.82E-06 &3.00 &2.11E-06 &3.00 \\
  &7 &3.73E-07 &3.00 &1.97E-07 &3.00 &4.77E-07 &3.00 &2.63E-07 &3.00 \\
  &&&&&&&& \\
  &3 &7.51E-05 &-- &2.21E-05 &-- &5.03E-04 &-- &5.31E-05 &-- \\
3 &4 &4.71E-06 &3.99 &1.38E-06 &4.00 &3.16E-05 &3.99 &3.33E-06 &3.99 \\
  &5 &2.95E-07 &4.00 &8.62E-08 &4.00 &1.98E-06 &4.00 &2.08E-07 &4.00 \\
  &6 &1.84E-08 &4.00 &5.39E-09 &4.00 &1.25E-07 &3.99 &1.30E-08 &4.00 \\
\vspace{-.05in} \\
\hline \\
\end{tabular}
\end{center}
\end{table}
We display our numerical results in Table \ref{table:convergenceOrders}.
Therein, the column labeled $\ell$ indicates the polynomial degree we used to compute
$u^h$, and the column labeled ``mesh'' denotes the mesh we used to compute the
relative error displayed in the corresponding row, more explicitly,
mesh$=i$ means we used a uniform mesh with $N=2^i$ elements.
In the column labeled ``order'' we display an approximate order
of convergence as follows. If $e_i$ denotes the relative error with
mesh$=i$, then we display the quantity
\[
r_{i+1} = -\frac{1}{\log 2}\frac{e_{i+1}}{e_i}
\]
at the row corresponding to mesh$=i+1$. The results displayed in
Table \ref{table:convergenceOrders} suggest an error estimate of the form
\[
\frac{\norm{(u-u^h)(\cdot,T)}_0}{\norm{u(\cdot,T)}_0} \le D\, h^{\ell+1}
\]
for some constant $D$ independent of $u$ and $h$, that is, the method
converges \emph{optimally} with respect to the mesh size.

\subsection{Approximations to Solutions}

Here we display some numerical results in which we solve \eqref{DG}
with $b = 0$ on $\Omega \times J$. In this case, we do not have an
explicit representation of the solution and merely rely on numerical
computing.  We consider two initial displacement functions

\begin{equation} \label{initialDisplacementDiscont}
u_{0,\textrm{disc }}(x) =
\left\{
\begin{aligned}
3/2, \quad & x \in [-1/4,1/4] ,\\
0, \quad &\mbox{otherwise},
\end{aligned}
\right.
\end{equation}
and
\begin{equation} \label{initialDisplacement}
u_{0,\textrm{cont}}(x)= \left\{
\begin{aligned}
0, \quad & x \in (-1, -1/4), \\
(1 + 4 x)^3 (96 x^2 - 12 x + 1), \quad & x \in [-1/4, 0), \\
(1 - 4 x)^3 (96 x^2 + 12 x + 1), \quad & x \in [0, 1/4], \\
0,  \quad & x \in (1/4, 1).\\
\end{aligned}
\right.
\end{equation}
These functions are displayed in Fig. \ref{fig:micromodulusData}.
In all cases, the initial velocity $v_0(x) = 0$ for all $x \in \Omega$.
The micromodulus function $C(x)$ is again taken to be the unit box given in \eqref{micromodulusFunc}.
We use the polynomial degree $\ell=2$ on a mesh with $N=64$ elements.
For each BC case, we depict the regulating functions $\varphi_C(k)$ utilized to define the
nonlocal operator as well as the associated wave propagation; see
Figures \ref{fig:periodic}, \ref{fig:antiperiodic}, \ref{fig:neumann},
and \ref{fig:dirichlet}.

For $t \in \REAL$, we have proved that the solution is discontinuous
if and only if the initial data is discontinuous; see Section
\ref{sec:strategy}.  Furthermore, the position of discontinuity is
determined by the initial data and should remain stationary.
Since we use vanishing initial velocity,
the explicit solution expression given in
\eqref{representationofthesolution} is as follows:
\begin{equation*}
u(x,t) =
\left[\overline{\cos \left(t \sqrt{\varphi(c)} \right)} \,\right]  u_0(x)
+ \delta u(x,t),
\end{equation*}
where $\delta u(\cdot,t)$ is a continuous function for $t \in \REAL.$
As seen in Figures \ref{fig:periodic}, \ref{fig:antiperiodic},
\ref{fig:neumann}, and \ref{fig:dirichlet}, discontinuities of the
initial data remain stationary at $x=-1/4$ and $x=1/4$.
We also numerically verify that the prescribed BCs are
satisfied for all $t \in [0,20]$.  For instance, it is easy to see
that homogeneous Dirichlet BCs are satisfied in Figure \ref{fig:dirichlet}.
Furthermore, the governing operator preserves the reflection symmetry.  In other
words, since initial data (both $u_{0,\textrm{disc}}$ and $u_{0,\textrm{cont}}$) are
symmetric with respect to $x=0$, the displacement is symmetric
with respect to $x=0$, which
can easily be observed by the symmetry in contour plots; see Figures
\ref{periodicContourPlot}, \ref{antiperiodicContourPlot},
\ref{neumannContourPlot}, and \ref{dirichletContourPlot}.

%%% interpretaion of local experiments
We also report solutions of local and nonlocal equations with
continuous initial data $u_{0,\textrm{cont}}(x)$ given in
\eqref{initialDisplacement}.  We observe several common features. Wave
separation behavior is similar to that from the unbounded domain case
as reported in the companion paper; see
\cite{beyerAksoyluCeliker2014_unbounded}. Namely, in the classical
case, as expected, we observe the propagation of waves along
characteristics; see Figure \ref{fig:neumannDirichletClassical} .  In
the nonlocal case, we observe oscillatory recurrent wave separation
and oscillations are located at the center of the initial
pulse. Hence, the wave patterns are symmetric with respect to $x=0$;
see Figure \ref{fig:neumannDirichletCont}.  As far as the boundary
behavior goes, in the classical case, we see that the Dirichlet BC
creates reflections of opposite signs; see Figures
\ref{dirichletClassicalCosinus} and \ref{dirichletClassicalSinus}.  In
the case of Neumann BC, reflections are of the same sign; see Figures
\ref{neumannClassicalCosinus} and \ref{neumannClassicalSinus}.  A
parallel behavior is observed for the nonlocal Dirichlet case. Such
parallel behavior is not obvious in the Neumann case.  Further
investigation of boundary behavior is needed.

\section{Conclusion}

This paper came about from the result in our companion paper
\cite{beyerAksoyluCeliker2014_unbounded} that the peridynamic
governing operator is a bounded function of the classical governing
operator on $\REAL^n$. The peridynamic operator contains a
convolution.  In this paper, we generalize convolutions to a bounded
domain with the help of a eigenbasis obtained from classical operator
on the bounded domain. This way, we can incorporate local boundary
conditions into nonlocal governing operators.  We study prominent
BCs such as periodic, antiperiodic, Neumann, and
Dirichlet.  In the case of periodic and antiperiodic boundary
conditions, integral representations of the abstract convolutions are
relatively straightforward to establish.  Such representations can
also be achieved for the case of Neumann BCs, but with
considerably more effort exploiting half-way symmetry.  For Dirichlet
BC, this integral representation involves an
orthogonal projection of the micromodulus function onto a closed
subspace defined in terms of the eigenbasis.  We give an integral
representation of this projection which does not depend on the
eigenbasis.  This representation involves a limit.
Applying convolutions of the periodic and antiperiodic cases,
we construct additional integral convolutions, what we call \emph{simple 
convolutions}, respecting Neumann and Dirichlet BCs. 

%%% Hilbert-Schmidt

For the homogeneous nonlocal wave equation with the considered BCs, we
prove that continuity is preserved by time evolution. Namely, for $t
\in \REAL$, we prove that the solution is discontinuous if and only if
the initial data is discontinuous. This is due to the fact that the
solution has a unique decomposition into two parts.  The first part is
the product of a function of time with the initial data.  This
decomposition is induced by the fact that the governing operator has a
unique decomposition into multiple of the identity and a
Hilbert-Schmidt operator; see \eqref{decompositionBC}.  Hence, the
second part is continuous.  The decomposition also implies that
discontinuities remain stationary. Whereas, in the classical case, it
is well-known that discontinuities propagate along characteristics.
We hold that this fundamentally difference is one of the most
distinguishing features of PD.

The paper presents a unique way of combining the powers of abstract
operator theory with numerical computing.  The abstractness of the
methods used in the paper allows generalization to other nonlocal
theories.  To substantiate the uniqueness of our treatment, we
provide a comprehensive numerical study of the solutions of the
nonlocal wave equation.  We accomplish to demonstrate two crucial
goals: For $t \in \REAL$ and each BC considered, discontinuities of
the initial data remain stationary and BCs are satisfied by the
solution.  We accomplish the two goals for all BCs and depict the
corresponding solutions.  For discretization, we employ a weak
formulation based on a Galerkin projection and use piecewise
polynomials on each element which allows discontinuities of the
approximate solution at the element borders.  We carry out a history
of convergence study to ascertain the convergence behavior of the
method with respect to the polynomial order and observe optimal
convergence.

Generically, an operator with regular coefficients on $\REAL^n$ has a
purely discrete spectrum, providing an eigenbasis of the underlying
space.  The methods provided in this paper can treat problems in $n$
spatial dimensions.  To our knowledge, it is the first systematic
approach of incorporating local BCs into nonlocal theories governed by
bounded operators. Regulating functions are the key for capturing the
essence of the underlying physics. The creation of a map from
regulating functions to physical situations is an exciting research
endeavor for future research.  In conclusion, we believe that we
added valuable tools to the arsenal of methods to treat nonlocal
problems and compute their solutions.

\fi
%\newpage
%\markright{Research statement of Burak Aksoylu}
%{\small
%\bibliographystyle{amsplain}

\bibliographystyle{spmpsci}
\bibliography{../../../data/digest/gemFEM2011/paper3_stokes/bib/burak}

%}

\end{document}